\documentclass[12pt, letterpaper, twoside]{article}
\usepackage[square,sort,comma,numbers]{natbib}
\usepackage{lineno}
\usepackage{lmodern}
\usepackage{amsmath,amssymb}
\usepackage{amsthm}
\usepackage{hyperref}
\usepackage{amsfonts}
\usepackage{verbatim}
\usepackage{url}
\usepackage{graphicx}
\usepackage{float}
\usepackage{booktabs}
\usepackage{mathrsfs}
\usepackage{epstopdf}
\usepackage{cancel}
\usepackage{MnSymbol}
\usepackage[dvipsnames]{xcolor}
\usepackage[margin=1in]{geometry}
\usepackage{algorithm}% http://ctan.org/pkg/algorithms
\usepackage{algpseudocode}% http://ctan.org/pkg/algorithmicx
\algrenewcommand\algorithmicrequire{\textbf{Input:}}
\algrenewcommand\algorithmicensure{\textbf{Output:}}
\newtheorem{theorem}{Theorem}[section]
\newtheorem{lemma}[theorem]{Lemma}

\newtheorem{definition}{Definition}[section]

\modulolinenumbers[5]
\allowdisplaybreaks

\usepackage{xcolor}
\usepackage[font=small]{caption}

\title{Characterizing Attitudinal Network Graphs through Frustration Cloud%\thanks{Grants or other notes
}
\author{Lucas Rusnak \and  Jelena Te\v{s}i\'{c}  %etc.
}
\date{March 15 2021}

\begin{document}

%\numberwithin{theorem}{section}
%\numberwithin{lemma}{section}
%\numberwithin{corollary}{section}
%\numberwithin{definition}{section}

\maketitle

\begin{abstract}
%Sentence of subjectivity and challenges and bias of social and content networks here} 

Attitudinal Network Graphs are signed graphs where edges capture an expressed opinion; two vertices connected by an edge can be agreeable (positive) or antagonistic (negative). A signed graph is called balanced if each of its cycles includes an even number of negative edges. Balance is often characterized by the frustration index or by finding a single convergent balanced state of network consensus. In this paper, we propose to expand the measures of consensus from a single balanced state associated with the frustration index to the set of nearest balanced states. We introduce the \emph{frustration cloud} as a set of all nearest balanced states and use a graph-balancing algorithm to find all nearest balanced states in a deterministic way. Computational concerns are addressed by measuring consensus probabilistically, and we introduce new vertex and edge metrics to quantify \emph{status}, \emph{agreement}, and \emph{influence}. We also introduce a new global measure of controversy for a given signed graph and show that vertex status is a zero-sum game in the signed network. We propose an efficient scalable algorithm for calculating frustration cloud-based measures in social network and survey data of up to 80,000 vertices and half-a-million edges. We also demonstrate the power of the proposed approach to provide discriminant features for community discovery when compared to spectral clustering and to automatically identify dominant vertices and anomalous decisions in the network. 

%These metrics are shown to distinguish between promotion and promotability. 
%A Law of Conservation of Controversy is also proven that proves that the increase in an individuals status comes at the cost of others, and can but used to predict a \emph{potential} maximum status in a given social network. 

% Practical considerations of spanning tree sampling methods, tree sample quantities, and social network size are examined. Finally, status/influence are examined relative to spectral clustering where radial distance in the status/influence-cone implies a correlation to spectral clustering over various $k$.  

%\keywords{Frustration cloud \and Balanced Graph \and Consensus \and Status \and Controversy \and Signed Graph }
%\subclass{MSC 05C22 \and MSC 05C38 \and MSC 05C85 \and MSC 05C90 \and MSC 91D30}
\end{abstract}

%\begin{table}
%\begin{tabular}{lll}
%\hline\noalign{\smallskip}
%\end{tabular}
%\end{table}

\section{Introduction}

Signed graph network representations of socio-technical networks offer richer modeling of relations between people, AI agents, products, and content. Attitudes captured by edges between two vertices can be agreeable (positive) or antagonistic (negative). Some social signed graph examples include team member evaluations in a company, student evaluations of instructors, movie recommendations based on common interests, or the ``trustworthiness'' of a product reviewer or seller in online stores. If there is a group decision to be made, consensus or majority voting guides the decisions and the final outcome in such networks. These sentiments have tangible real-world effects, such as annual performance scores and promotions in a corporation. Graph decision algorithms are rarely scrutinized, as consensus and majority voting are established social constructs \cite{snapnets}. Only when the outcomes are known can an individual's status be perceived as elevated or diminished relative to their peers, and even then only anecdotally questioned or explained. Such algorithms' sensitivity to bias, fraud, and falsehood has been put under a magnifying glass in the last couple of years, as state-of-art research examined the controversy of decisions in cases of status quo \cite{Li2005} or subgroup mobilization against other groups \cite{Kumar2018}. Research in the domain of consensus in signed graphs focuses on frustration index computation and algorithmic convergence to \emph{some} balanced state \cite{2019Altafini,She2020} and is one dimensional, as described in Section~\ref{sec:Related}. 

In this paper, we focus on finding multiple balanced states of the signed graph, as they represent different consensus outcomes. The main contribution of this paper is that it generalizes the notion of the frustration index to the {\bf frustration cloud}. The \emph{frustration cloud} is a set of all balanced states obtained by a minimal number of edge sign inversions. In comparison, the \emph{frustration index} characterizes the distance between the original signed graph and a single nearest balancing set. We also propose a spanning tree-based graph balancing algorithm that focuses on finding balanced states from spanning trees. The proposed approach works on {\bf any} signed graph, avoids the NP-hardness of finding the frustration index, and focuses on determining a basis of fundamental cycles to produce balanced states \cite{OH1}. A greedy approach to finding a basis of fundamental cycles is NP-hard in general, and Deo et al.~proposed some polynomial-time algorithms \cite{deo1982algorithms}. The \emph{graphB} algorithm introduces a deterministic methodology that finds all the nearest balanced states of a signed graph. To quantify levels of agreement in the network, we sample the frustration cloud via the associated family of balanced matroidal bases (spanning trees) \cite{SG}.  This statistically meaningful sampling of the frustration cloud produces a robust way to handle the brittleness of the data space for signed graph data and avoid challenges presented in \cite{Selbst2018}. 

The proposed spanning tree-based balancing method combines the requirement for statistical parity across the nearest balanced states with the requirement to consider \emph{all} vertices instead of few selected ones; this method relies on the spanning trees, not random walks \cite{Garimella2017}. The sentiments are reconstructed around a spanning tree to produce a set of nearest balanced states. The resulting balanced states are generalizations of bipartite graphs \cite{Berge1, Har1}, and the resulting negative edge cut defines two consensus-based sets. We use these consensus-based sets to characterize the importance of specific vertices and edges necessary to produce a majority consensus: \emph{status} measures an individual's contribution to reaching consensus over the frustration cloud; \emph{agreement} measures the edge's contribution to belonging to the majority consensus; and \emph{influence} measures the vertex based on its edge agreement scores.

\subsection{Contribution}
\label{sec:contr}

Social network analysis has not converged on how to assess the robustness, resilience, and reliability of the network algorithm outcomes, or how to identify anomalies in large signed network graphs. There is a clear need to measure the performance of algorithms that define outcomes, characterize consensus in social or multi-agent attitudinal networks as a unit, and assess vertex and edge contributions to graphs as a whole.  We propose a process that characterizes the impact of every vertex and edge in its entirety, and it may be used outside of social network analysis on any binary decision paradigm to examine the reliability of decision-making processes relative to some given ground state.  Researchers in multi-agent networks have focused on techniques to produce a \emph{single} convergent balanced state \cite{8062440,2019Altafini,7798380,Hu2013,Jiang2016,She2020}.

We propose a new discrete alternative to Laplacian dynamics, and we identify all nearest balanced states of a signed graph. Our main contribution is a novel signed graph methodology that (1) determines all the nearest balanced states via basis sampling via spanning trees, (2) quantifies the importance of each balanced state relative to the likelihood it will be become the consensus state, (3) quantifies an individual's \emph{status} relative to their peers, (4) characterizes the \emph{potential} maximum status of an individual over tie-break scenarios, (5) provides a constant metric of \emph{controversy} for the entire network that is subject to a Conservation Law, (6) quantifies an individual decision or opinion based on \emph{agreement}, (7) aggregates agreement to each individual to quantify \emph{influence} over others, (8) compares status and influence to quantify the positive/negative relationship of the entire network and provide a spectrum of status-influence, and (9) scales the tree-based balancing algorithm to {\bf graphB:} graph balancing using statistically significant sample of spanning trees \cite{graphB}. 

In the event that the sentiment data provided is related to promotions, and the outcome of those promotions are known, the research examines the efficacy of this new methodology by outcomes as status separates ``promoted'' from ``not promoted'' and identifies any outliers in either case, and where the influence separates ``voters'' from those ``voted on''. We show proof-of-concept implementation results on several large social networks and how status-influence measures of the vertex can discover contentious outcomes on Wikipedia administrator promotions. We also identify anomalous actors when outcomes are not known in the Slashdot dataset \cite{snapnets}, and we analyze the approach on a small survey dataset \cite{1954Read} which shows that in new status-influence space, vertices can be fully characterized in terms of community without the need to specify the number of clusters $k$ for spectral clustering.

\section{Background and Related Work}
\label{sec:Related}

In this section, we describe the state-of-art work in mathematical sociology (\ref{ssec:soc}) and signed graph frustration (\ref{ssec:frustration}), and we present related work in the fields of social network analysis and control (\ref{ssec:SNA}) and signed graph clustering (\ref{ssec:sign}).

\subsection{Related Work in Mathematical Sociology}
\label{ssec:soc}

A \emph{signed graph} consists of a collection of vertices that are linked together with undirected edges; positive sentiment between two vertices is modeled as positive edge ``+1'', and negative sentiment between two vertices is modeled as a negative edge ``-1''.  Fritz Heider introduced Balance Theory in 1948 \cite{Heider}. Balance theory examined consensus in triadic relationships in a signed triangle graph. Figure~\ref{fig:BalEx} shows all possible edge signs for a signed triangle graph. These eight graphs are different as they have different edge signs between specific vertices. Each of the eight signings is a \emph{state} of the triangle graph. \emph{Balance Theory} is the base model of attitude change analysis among three persons in a signed graph \cite{Heider}. A triangle state is considered \emph{balanced} if the product of the edge signs is positive and \emph{unbalanced} if the product of edge signs is negative. A balanced triadic relationship in a signed graph is captured as ``the enemy of my enemy is my friend'' paradigm in mathematical social modeling \cite{Leskovec2010a}.

\begin{figure}[!ht]
    \centering
    \includegraphics{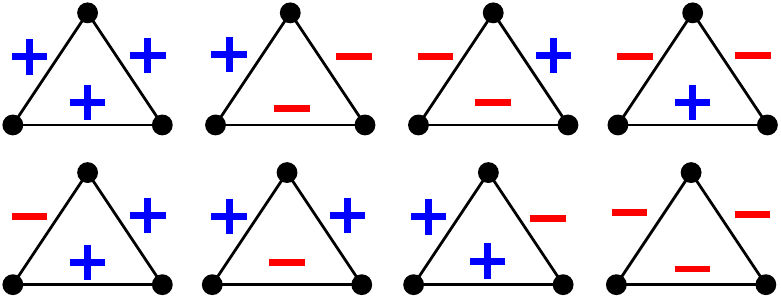}
    \caption{Sign graph triangle where the top row are balanced states, and the bottom row are unbalanced states.}
    \label{fig:BalEx}
\end{figure}

There are four different ways to achieve a balanced state in the triangle, as depicted in Figure~\ref{fig:BalEx}, and we emphasize that there is \emph{more than one} balanced state. These multiple consensus scenarios all capture different aspects of reaching consensus. A specific balanced state is a snapshot of the balanced assessment of the network, and it is \emph{not sufficient} as all sentiments are not necessarily equal, as illustrated in Figure~\ref{fig:BalEx} (top row).  In this paper, we implement an algorithm that considers the collection of nearest consensus states to characterize network graph behavior.  This type of analysis considers all possible consensus outcomes of attitudinal network graphs for a complete consensus characterization. 

\emph{Mathematical sociology} was introduced when Rashevsky characterized large social networks as graphs, where vertices are persons and edges measure the level of acquaintanceship \cite{Rashevsky1984}. The mathematical foundation of signed graphs \cite{Har0} and social balance theory \cite{Heider, Abelson1958} introduced the concepts of modeling balance and agreement in social networks using more complex mathematical models.  Harary introduced the \emph{frustration index} of a signed graph as a measure of how far the network graph is from a state of structural balance \cite{Har1}.  Harary's proposed measure is the smallest number of edges whose negation in the network graph results in a balanced signed graph.  Harary's attitudinal balanced model was formalized in graph-theoretic terms \cite{Har2} and fully characterized by Zaslavsky \cite{SG} in matroid-theoretic terms. Davis \cite{Davis1967} studied the necessary and sufficient conditions for clustering of attitudinal graphs. Mathematical sociological modeling has evolved to address sociological phenomena in various fields of social science and helps in understanding, evaluating, and predicting patterns of social relationships and interactions \cite{Hunter1984}. 

\subsection{Balance and Frustration in Signed Graph}
\label{ssec:frustration}

A \emph{signed graph} $\Sigma$ is a pair $(G,\sigma)$ that consists of a graph $G = (V,E)$ and an edge-signing function $\sigma : E \rightarrow \{+1,-1\}$. For a set of edges $E$ in a signed graph $\Sigma$, let $E^{+}$ (resp. $E^{-}$) denote the set of positive (resp. negative) edges of $G$ --- the signs of the edges are regarded as sentiments between two vertices. The \emph{sign of a subgraph} is the product of the signs of the edges in that subgraph. A signed graph is \emph{balanced} if the sign of every circle is positive \cite{Har0,Har2}. If the graph $\Sigma$ is not balanced, there exists a set of edges whose sign reversal produces a balanced signed graph, and that set is called a \emph{balancing set}.

\begin{figure}[!ht]
    \centering
    \includegraphics[width=3.5in]{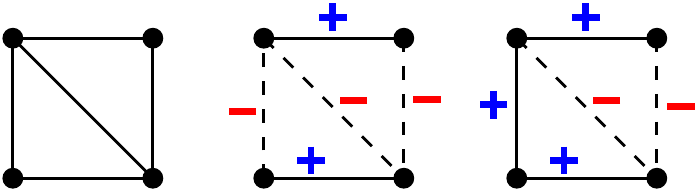}
    \caption{left: an underlying graph $G$; middle: an example balanced signing of $G$; right: an unbalanced signing of $G$, $\Sigma$. $G$ and $\Sigma$ are used as signed graph examples in the rest of the paper.}
    \label{fig:MainEx}
\end{figure}

A balancing set is \emph{minimal} if no proper subset is a balancing set. The \emph{frustration index} of a signed graph $\Sigma$, denoted $Fr(\Sigma)$, is the smallest number of edges whose change in sign can result in a balanced signed graph \cite{Har1}. All balanced signed graphs necessarily have a frustration of $0$. The frustration index has applications in various areas including physics \cite{Spin4, Spin5}, economics \cite{2011Yoshikawa}, negative feedback loops in Boolean networks \cite{SONTAG2008518}, and statistical mechanics \cite{StatisticalMechanics}. Each balanced state represents a consensus outcome, meaning all paths between two vertices have the same sign. These concepts are related by a Theorem of Harary and motivate our proposed probabilistic consensus model to examine all the nearest balanced states.

\begin{theorem}[\cite{Har0, Har1}] For a signed graph $\Sigma'$, the following are equivalent:
\begin{enumerate}
 \setlength{\leftmargin}{0pt}
    \item $\Sigma'$ is balanced. (All circles are positive.)
    \item For every vertex pair $(v_i,v_j)$ with $v_i,v_j \in V$ all $(v_i,v_j)$-paths have the same sign. (Agreement or consensus)
    \item There exists a bipartition of the vertex set into sets $U$ and $W$ such that an edge is negative if, and only if, it has one vertex in $U$ and one in $W$. The bipartition ($U$,$W$) is called the \emph{Harary-bipartition}.) 
    \item $Fr(\Sigma') = 0$. ($0$ frustration.)
\end{enumerate}
\label{t:HararyCut}
\end{theorem}

\begin{figure}[!ht]
    \centering
    \includegraphics{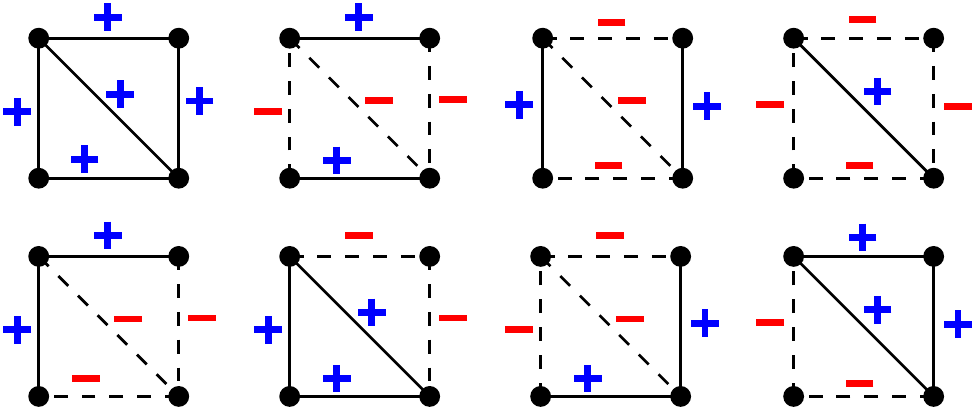}
    \caption{All $8$ balanced signed graphs of the underlying graph $G$ in Figure~\ref{fig:MainEx}. Harary-cut of negative edges is represented by dashed edges.}
    \label{fig:HCut}
\end{figure}

Figure~\ref{fig:HCut} shows all possible balanced states $\Sigma'$ on the given underlying graph $G$ from Figure~\ref{fig:MainEx} (left).  Each of the balanced states satisfies all conditions in Theorem \ref{t:HararyCut}. Once balanced, the set of negative edges whose deletion produces the Harary-bipartition is called the \emph{Harary-cut} of the balanced graph. The Harary-cuts are emphasized by representing the negative edges with dashed edges. 

The frustration index is the size of the smallest Harary-cut; for example in Figure~\ref{fig:MainEx}, $Fr(\Sigma) = 1$.  Computing the frustration index of a signed graph is an NP-hard problem. There exist scenarios that are solvable in polynomial time and for which exact large-scale solutions are possible. Wu and Chen proposed a branch-and-bound algorithm to balance signed graphs by editing edges and deleting vertices, demonstrating its efficiency over trivial and heuristic algorithms on inputs with up to $n=40$ vertices \cite{2013Wu}. In control of multi-agent systems, Altafini analyzed the convergence to a balanced state in the decision-making process and presented an effective way to compute the average consensus for a network with up to 100 vertices \cite{2019Altafini}.

Aref et al.~developed three binary linear programming models to compute the frustration index quickly and exactly as the solution to a global optimization problem. They demonstrated the efficiency of their techniques for inputs with up to 15,000 edges \cite{Aref2016,Aref2020} with an extension that allows for the incorporation of weights in the interval $[-1,+1]$ to determine weighted minimum frustration \cite{Aref2020}. The computational complexity of \cite{Aref2016,Aref2020} algorithms is bounded by a polynomial function of the size of the underlying graph.  

As balancing only requires the sign (sentiment) of the edge, and not the intensity (weight), we demonstrate that the balancing applications introduced in this paper produce a quantitative spectrum of vertex and edge metrics that drive balanced/consensus outcomes. By maintaining the separation of signs and weights, as suggested by Zaslavsky in \cite{BSG}, we are able to preserve the matroidal structure, and our approach immediately generalizes to any edge weight value by replacing each tree with the weight-product as in \cite{Tutte}, which is expected in future work. Another critical difference in our approach is the relaxation on determining the frustration index (or weighted frustration index). We are producing a set of balanced states with minimal (in containment) sign alterations, instead of the minimum value (in cardinality) for frustration. We propose an approach in Section \ref{sec:Balancing} that focuses on multiple nearest balanced states to avoid the NP-hardness of determining the frustration index while simultaneously analyzing multiple possible nearest consensus outcomes that scale with the size of social network. 

\subsection{Related Work in Social Network Analysis and Cybernetics}
\label{ssec:SNA}

Large virtual communities and decision networks of the $21^{st}$ century initiated the explosive growth of social network analysis and network science fields.  The analysis of largely digital traces of \emph{social networks at scale} expanded well-studied mathematical algorithms for reinforcement, information processing, social judgment, balance, and dissonance \cite{Hunter1984}. Wasserman et al. introduced social network analysis as algebraic graph representations and proposed a series of statistical tests \cite{wasserman1994social}. The domain research has focused on \emph{predicting} the existence and/or sentiments of edges in the graph, recommending content or a product, or identifying unusual trends. Baseline signed graph theory was used to explain the relative status that individuals hold within in a social network \cite{Leskovec2010a,Leskovec2010b} and focused on socially-conscious science to help understand bias, controversy, conflict, and trust \cite{Mishra2011,Guha2004}.  All mathematical models in network science that model intents and trends in online social networks have relied on aspects of well-established consensus-based models in  signed graph theory \cite{Garimella2017, Yuan2017, Chen2018} and balanced modeling \cite{Javed2018,Lu2011,Ruby2017,Tang2016,Zhao2018,Zhou2018}.

A multitude of measures have been proposed to access the rich information coded in signed graphs. Mishra et al. \cite{Mishra2011} introduces \emph{trustworthiness} and \emph{deserve} as local vertex-based measures of bias to reflect the expected weights of out- and in-edges. \emph{Controversy} was introduced by Garimella et al. \cite{Garimella2017} as the likelihood a random walk will return to the same side of the network. This method improved the examination of triangles in \cite{Leskovec2010b} by including pendant vertices and proposed to reduce controversy by bridging opposing viewpoints. \emph{Conflict} is defined in Chen et al. \cite{Chen2018} by examining the Laplacian matrix to produce a ``Conservation Law of Conflict'' reminiscent of Kirchhoff's laws --- we provide our own Conservation Law of Controversy in subsection \ref{ssec:controversy}. Kumar et al. \cite{Kumar2018} discuss group mobilization against other groups to describe conflict in intra-community interactions, and Guha et al. \cite{Guha2004} examines \emph{trust} through an iterative build of belief matrix. Yuan et al. \cite{Yuan2017} introduces a sign prediction model for sparse data edge prediction in which they convert the original graph into a edge-dual graph and apply machine learning to predict signs in sparse graphs. Established methods of network graph analysis focus on endorsement analysis through local topology analytics and strive for agreement by changing \cite{Leskovec2010b}, adding \cite{Garimella2017} or removing \cite{Guha2004} edges in the graph. We propose to analyze the signed graph in its entirety and characterize the vertices and edges through {\bf frustration cloud}-based attributes.

In cybernetics, a multi-agent network dynamic is often too complex for existing tools to analyze the entire network and collective dynamic reactions. It is known \cite{OH1} that our methodology for balancing a signed graph works on any signed graph and certain classes of hypergraphs. Altafini \cite{2013Altafini,2019Altafini} proposes controllability and consensus algorithms in networks by examining the effects of the bipartite consensus of Harary \cite{Har0}. Pan et al. \cite{7798380} examine the bipartite structure of Laplaican dynamics and node decomposition. Hu et al. show that the ideal state of the multi-agent system can modeled as a balanced graph, and that the system converges to the optimal state through the bipartite consensus iterations \cite{Hu2013}, while uncontrollability and stabilizability is examined by Alemzadeh et al. in \cite{8062440}. Algorithms for the characterization of the status quo have been examined in \cite{Li2005} for transitive graphs. Jiang et al. propose a sign-driven consensus as a control protocol measured via Laplacian dynamics \cite{Jiang2016}. She et al. \cite{She2020} examine consensus in terms of graphical characterizations of the controllability of signed networks and offers a heuristic algorithm for leader selection based on balance theory. This prior work focuses on producing a single balanced state.  In this paper, we propose a discrete alternative to find all nearest balanced states of the network via frustration cloud graph analysis.

\subsection{Related Work in Sign Graph Clustering}
\label{ssec:sign}

We compare the methods introduced in this paper to standard spectral clustering on only positive edges for community detection. Researchers have only recently started mining negative links in networks for community detection \cite{Esmailian2015}.  Spectral clustering for signed graphs was introduced by Kunegis et. al \cite{Kunegis2010} by the way of a positive semi-definite modified Laplacian matrix approach. The approach essentially counts positive edges between clusters and negative edges within clusters. Multiple signed graph clustering methods have been proposed since \cite{Tang2016}, normalizing Lapacian in different ways. A more recent survey is begin prepared that compares multiple signed spectral clustering methods in terms of effectiveness for the community finding and scalability for large network graphs.
%\cite{Tomasso2021}. 

\section{The Frustration Cloud Graph Analysis}
\label{sec:Balancing}

In this section, we expand the notion of the frustration index to frustration cloud analysis and propose a tree-based graph balancing algorithm  to discover the nearest balanced states of a signed graph. This methodology improves on the singular focus of the frustration index while avoiding the tedious calculations of finding all balanced states, some of which are only present by passing through another balanced state.

We define \emph{frustration cloud} as a set of all balanced states of an underlying graph that are achievable by a minimal number of edge sign changes. If a balanced state belongs to the frustration cloud, that means no subset of its edges can balance the underlying signed graph. While balanced states for an underlying unsigned graph $G$ are always the same, the nearest balanced states for a signed graph $\Sigma$ depend on the $\Sigma$ and the minimal number of edge signs that need to be changed to achieve a balanced state.  Balanced states that produce the frustration index are those with a minimal number of edge changes to reach a balance, and are always part of the frustration cloud. The nearness of these states for discovery of the frustration index are discussed in \cite{Aref2016}. We use spanning trees as matroidal bases to balance the signed graph.  A \emph{spanning tree} $T$ of a graph $G$ is a maximal acyclic subgraph that contains all the vertices of $G$. For a spanning tree $T$ in graph $G$ and an edge $e$ not in the spanning tree, $e \in E(G) \setminus E(T)$, the \emph{fundamental cycle of $e$ with respect to $T$ in $G$} is the unique cycle in $T \cup e$. The number of edges outside a spanning tree is a known constant called the \emph{cyclomatic number}. Spanning trees form a basis for the balanced signed-graphic matroid \cite{SG}. A spanning tree plus an additional edge whose fundamental cycle is negative is the base for the unbalanced signed graph.

\subsection{Balancing via Spanning Trees}
\label{ssec:Balancing}

For a connected graph $G$, let $\Sigma = (G,\sigma)$ be the signed graph of $G$, and  $\mathcal{T}_{G}$ be the set of spanning trees of $G$. We propose the graph-balancing algorithm that constructs the nearest balanced states of $\Sigma$ from spanning trees of the underlying graph $G$. The underlying graph $G$ is assumed to be connected.  If it is not, the algorithm is applied to connected components of $G$. Algorithm ~\ref{alg:Balance} produces one balanced state $\Sigma_T$ per spanning tree $T$. 

\begin{algorithm}[!ht]
  \caption{Signed Graph Tree-Balancing Algorithm:}
  \label{alg:Balance}
  \begin{algorithmic}
    \Require Input signed graph  $\Sigma = (G,\sigma)$. 
    \ForAll {$T \in \mathcal{T}$, $T$ is a spanning tree of $\Sigma$}
     \ForAll  {edges $e$,  $e \in \Sigma \setminus T$}
            \If {fundamental cycle $T \cup e$ is negative}
            \State       change edge sign: $e^- -> e^+; e^+ -> e^-$
            \EndIf
            \EndFor 
            \State   Construct new balanced signed graph $\Sigma_T$
        \EndFor
        \Ensure Set of nearest balanced states $\Sigma_T, T \in \mathcal{T}$
        \end{algorithmic}
\end{algorithm}

Algorithm ~\ref{alg:Balance} is illustrated in Figure~\ref{fig:BalStates0}, where the process is illustrated for the signed graph $\Sigma$ (left) and a single spanning tree (second left). Edges outside the spanning tree are dashed. As the fundamental cycles are found, the edges outside the spanning tree (grey) are examined. The edge sign is not changed if the fundamental cycle is positive (top), and it is changed if the fundamental cycle is negative (bottom) in Figure~\ref{fig:BalStates0} (second right).  The balanced signed graph is produced with these signing changes in Figure~\ref{fig:BalStates0} (right). 

\begin{figure}[!ht]
    \centering
    \includegraphics[width=4in]{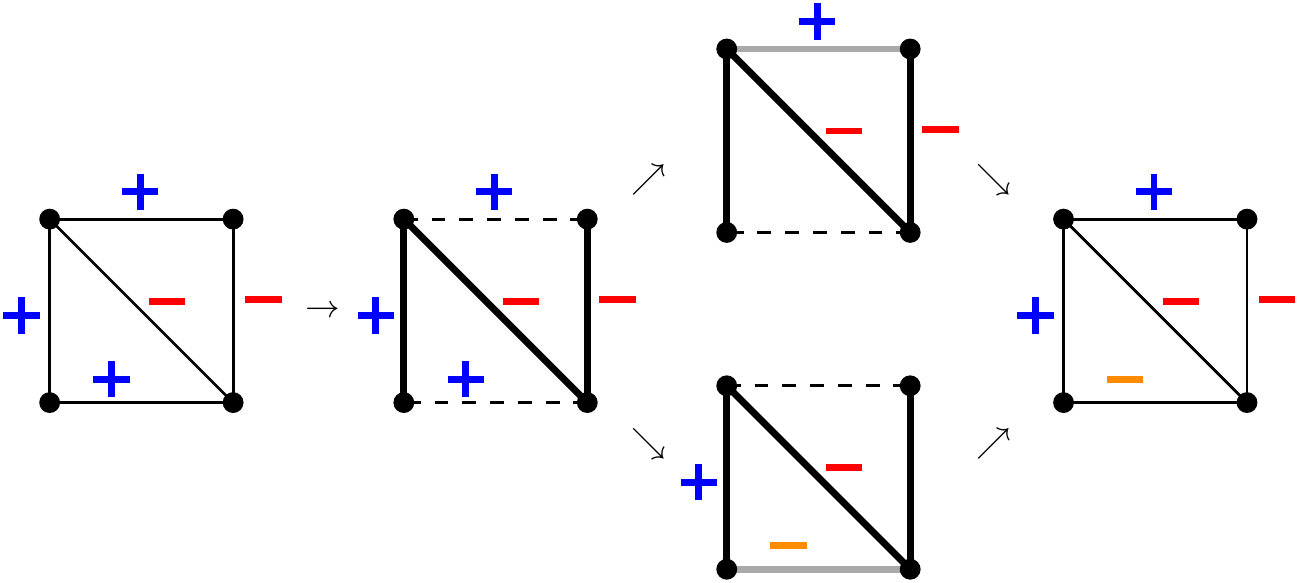}
    \caption{The spanning tree balancing process via fundamental cycles. Changed edges appear lighter.}
    \label{fig:BalStates0}
\end{figure}

The base signed graph $\Sigma$ in Figure~\ref{fig:BalStates0} (left) has a total of $8$ spanning trees, marked with darker edges in Figure~\ref{fig:BalStates1} (right). The edges outside each spanning tree are indicated by dashed edges. For any spanning tree $T$ and an edge $e$ outside of that tree, the sub-graph $T \cup e$ contains a unique fundamental cycle $C$. The sign of $e$ is chosen so that $C$ is \emph{positive}. The Algorithm~\ref{alg:Balance} result for $8$ spanning trees and signed graph $\Sigma$ is in Figure~\ref{fig:BalStates1} (right).

\begin{figure}[!ht]
    \centering
    \includegraphics[width=4.5in]{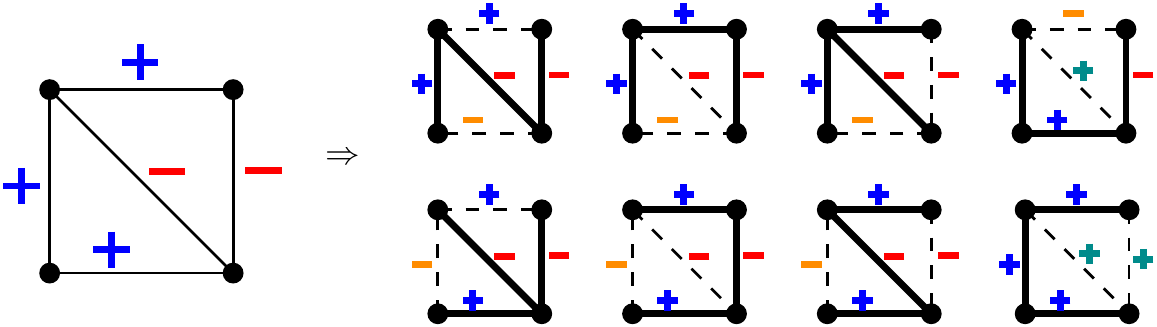}
    \caption{The spanning trees of a signed graph (bold) produce balanced signed graphs. Edges outside each spanning tree are dashed, and re-signed edges are labelled in orange and teal (lighter). The negative edges form a cut-set in each balanced graph.}
    \label{fig:BalStates1}
\end{figure}

For an underlying graph $G$, there are $8$ possible balanced graphs as shown in Figure \ref{fig:HCut}. However, only four of the eight balanced states are achievable by Algorithm~\ref{alg:Balance}, as shown in Figure~\ref{fig:BalStates2}.  Not every balanced signed graph is obtainable by a balancing algorithm that uses spanning trees, \emph{only the nearest balanced states} are. 

\begin{figure}[!ht]
    \centering
    \includegraphics[width=4in]{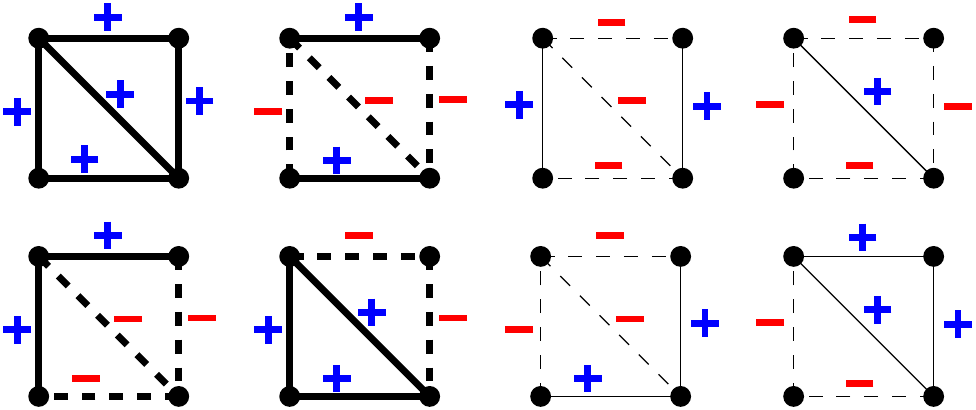} 
    \caption{Out of the eight balanced graphs for signed graph in Figure \ref{fig:BalStates1}. The four in bold can be reconstructed using the tree-balancing algorithm (Alg.~\ref{alg:Balance}).} 
    \label{fig:BalStates2}
\end{figure}

\begin{theorem}
If $\Sigma = (G,\sigma)$ is a signed graph of $G$, then the tree-balancing algorithm outlined in Alg.~\ref{alg:Balance} produces a minimal balancing set for $\Sigma$.
\label{t:BalSetMin}
\end{theorem}
\begin{proof}
Let $\Sigma$ be  the signed graph of graph $G$, $T \in \mathcal{T}$ a spanning tree of $\Sigma$, and $B_T$ be the balancing set produced by the tree-balancing algorithm (Alg.~\ref{alg:Balance}). If $B_T$ is not minimal, then there exists a smaller balancing set $S \subset B_T$ and an element $e \in B_T \setminus S$ whose reversal is not necessary to balance $\Sigma$. However, $T \cup e$ has a unique fundamental circle $C$, and the only edge of $C$ outside of $T$ is $e$, so $e$ is required to balance and $B_T \setminus S$ must be empty.
\end{proof}

We explain the notion of nearest balanced states and the construction of the frustration cloud in Section~\ref{ssec:aReal}.

\subsection{The Frustration Cloud and Consensus}
\label{ssec:aReal}

In this section we formalize the notion of the \emph{frustration cloud} as the set of nearest balanced states. See Definition~\ref{d:FrustrationCloud}. 

\begin{definition} The \emph{frustration cloud} of a signed graph $\Sigma$, denoted $\mathcal{F}_{\Sigma}$, is the set of all balanced signed graphs obtained by graph B, the tree-balancing algorithm~\ref{alg:Balance} on $\Sigma$.
\label{d:FrustrationCloud}
\end{definition}

All balanced states of the underlying signed graph have $0$ frustration, per Theorem \ref{t:HararyCut}. They all represent different views of graph consensus. The frustration index is the smallest number of edge sign switches so the signed graph achieves a balanced state. If the frustration index of signed graph $\Sigma$ is $Fr(\Sigma)$, that means that $\Sigma$ is $Fr(\Sigma)$ many sign changes from being balanced. 

Let us extend that notion to all balanced states. For a signed graph $\Sigma = (G,\sigma)$, the set of edge-signing functions $\{+,-\}^{E}$ form a Boolean lattice $\mathcal{L}$ ordered by negative edge subset containment. Thus, the all positive edge signing $(G,+)$ is the $\mathbf{0}$ element, the all negative edge signing $(G,-)$ is the $\mathbf{1}$ element, and it is graded by the number of negative edges. Let $\Sigma_1 = (G,\sigma_1)$ and $\Sigma_2 = (G,\sigma_2)$ be two signings of the same underlying graph $G$. The \emph{distance} between $\Sigma_1$ and $\Sigma_2$, $d(\Sigma_1 , \Sigma_2)$ is the Hamming distance between them in $\mathcal{L}$. This is equivalent to the length of the shortest path between $\Sigma_1$ and $\Sigma_2$ in $\mathcal{L}$ when regarded as a graph. The Boolean lattice for a signed triangle graph is illustrated in Figure~\ref{fig:LatticeK3}. 

\begin{theorem}
Let $\Sigma$ be a signed graph, and let $\Sigma'$ be a balanced state of $\Sigma$. $\Sigma' \in \mathcal{F}_{\Sigma}$ if, and only if, $\Sigma'$ can be obtained by the minimal balancing set whose size is less than or equal to the cyclomatic number.
\label{t:MinReal}
\end{theorem}

\begin{proof}
If $\Sigma' \in \mathcal{F}_{\Sigma}$, it is obtained by the tree-balancing algorithm, which cannot change more edge signs than the cyclomatic number. By Theorem \ref{t:BalSetMin}, this is a minimal balancing set.

Now, suppose $\Sigma'$ is obtained by the minimal balancing $B$ set whose size is less than or equal to the cyclomatic number. Observe that $G \setminus B$ is connected, and any spanning tree in $G \setminus B$ will also be spanning in $G$. Thus, $B$ is obtained by a spanning tree in $G \setminus B$. 
\end{proof}

The frustration cloud is the set of balanced states resulting from $\Sigma$ that have no more than the cyclomatic number of edge sign changes. It has a simple interpretation using the Boolean lattice of the signings of the underlying graph $G$. Consider the eight possible signings of the triangle graph in  Figure \ref{fig:LatticeK3} (left), ordered by negative edge set containment. Out of these eight signings, exactly four of them are balanced; these are marked with black boxes in Figure \ref{fig:LatticeK3} (center).

\begin{figure}[!ht]
    \centering
    \includegraphics{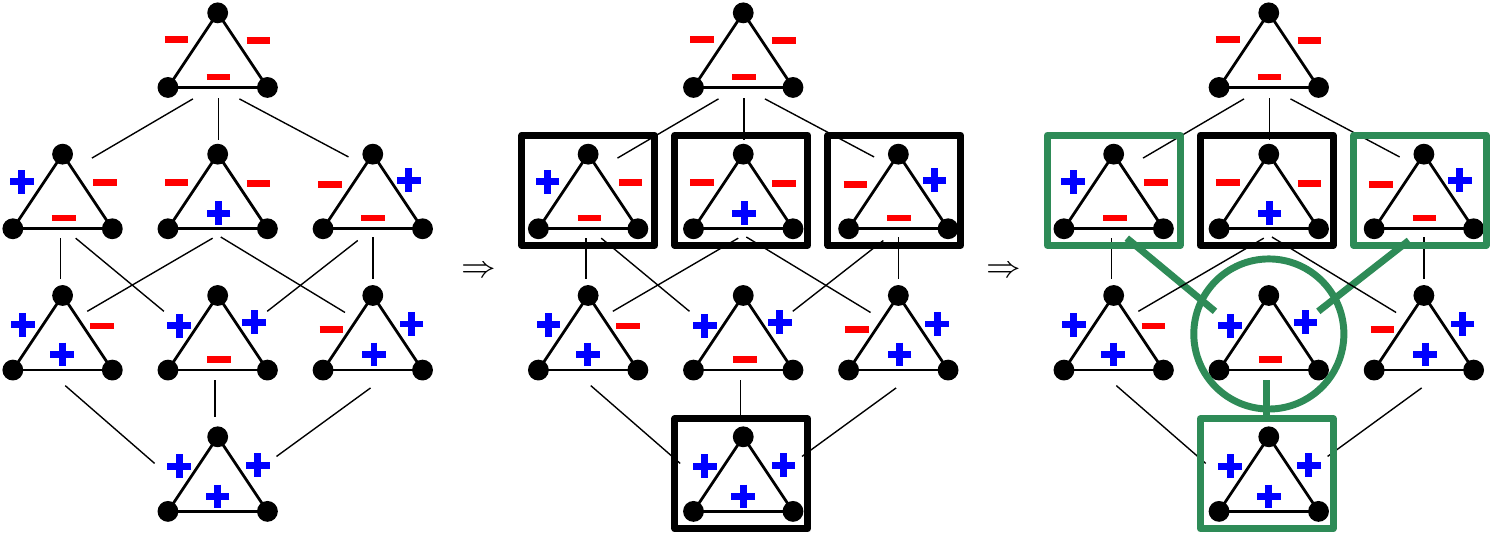}
    \caption{The Boolean lattice for a signed triangle graph (left); black boxes mark the balanced signed graphs (center); if the underlying signed graph $\Sigma$ is in green circle, then the green boxes mark a frustration cloud $\mathcal{F}_{\Sigma}$ (right).} 
    \label{fig:LatticeK3}
\end{figure}

Consider the signed graph $\Sigma$ to be the graph in Figure \ref{fig:LatticeK3} (right) with a green circle around it. Since the triangle graph has the cyclomatic number $1$, we search for all balanced states that are distance $1$ or less away from $\Sigma$; these are marked with green squares. Observe that the balanced state in the black box in Figure \ref{fig:LatticeK3} (right) is not in $\mathcal{F}_{\Sigma}$, as it requires a path that exceeds the cyclomatic number --- one must also travel through another balanced state to reach it.

More complicated graphs may produce balanced states of varying distance from the given signed graph. Consider the underlying graph given in Figure~\ref{fig:AggRelEx} (left) and its corresponding Boolean lattice of signings in Figure~\ref{fig:AggRelEx} (middle) where the negative edge sets are listed. Consider the signed graph $\Sigma$ where edges $e_2$ and $e_5$ are negative, and the rest are positive (Figure~\ref{fig:BalStates1}. This is marked with the open square in Figure~\ref{fig:AggRelEx} (middle) labeled $25$. All the balanced signings are marked with closed squares. Since the cyclomatic number of the underlying graph is $2$, we search for all balanced states of distance less than or equal to $2$ from the open square; these are indicated by the dark paths in Figure~\ref{fig:AggRelEx} (right).  

\begin{figure}[!ht]
    \centering
    \includegraphics[scale=1]{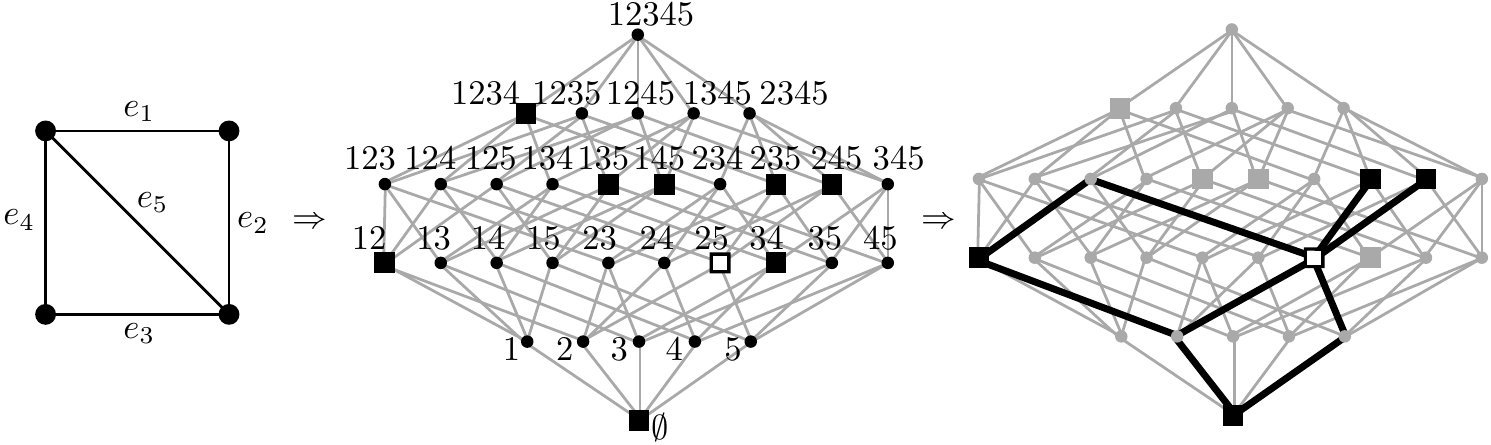}
    \caption{Left: The underlying graph $G$ from Figure \ref{fig:BalStates1}. Middle: The Hasse diagram of all signings of $G$, with balanced states as closed squares, and the given signed graph from Figure \ref{fig:BalStates1} as the open square. Right: The four elements of $\mathcal{F}_{\Sigma}$ and their shortest paths to $\Sigma$ from Figure \ref{fig:BalStates2} (bold).}
    \label{fig:AggRelEx}
\end{figure}

The example in Figure \ref{fig:AggRelEx} illustrates that the frustration index is obtainable by analyzing the frustration cloud. The signed graph $\Sigma$ from Figure \ref{fig:AggRelEx} has $Fr(\Sigma)=1$, as the balanced state of minimum distance is of distance $1$ away from $\Sigma$. It is trivial to verify that the frustration cloud of a balanced graph consists only of itself.

\section{Probabilistic Consensus Model}
\label{sec:consensus}

Consensus for social networks is community resolution when opposing parties set aside their differences and barely agree on a statement \cite{2011Hartnett}. State-of-art consensus modeling in social network analysis has focused on the locality of the agreement \cite{Guha2004, Garimella2017}, and it did not consider an entire graph. As illustrated in Section~\ref{sec:Balancing}, there can be multiple balanced states of the same graph, meaning there are multiple ways to achieve global consensus. In this section, we formalize the measures of vertices, edges, and the entire graph stemming from frustration cloud-based analysis. 

There are multiple ways in which a minimum set of sentiments can be changed to result in \emph{identical outcomes of consensus}. Different spanning trees in the tree-balancing algorithm (Alg.~\ref{alg:Balance}) can result in the same nearest balanced state, as illustrated in Figure~\ref{fig:BalStates1}. We weigh each element of the frustration cloud by the number of times it is produced by a basis, as spanning trees are bases for the balanced states of a signed graph \cite{SG}.

%The benefit for exchanging the frustration index calculation for fundamental cycle bases, is that we produce more balanced states and are able to quantify the likelihood a given balanced state will be achieved without needing iterative applications of Laplacian dynamics. Specifically, we build an approach that produces more usable data for the same cost, and probabilistic sampling can be used to increase performance.

\begin{definition}
For a signed graph $\Sigma = (G,\sigma)$ and balanced signed graph $\Sigma' \in \mathcal{L}$, let $w_{\Sigma'}$ be \emph{weight of $\Sigma'$ relative to $\Sigma$} and defined to be the number of spanning trees of $G$ that balance $\Sigma$ into $\Sigma'$.
\label{d:weights}
\end{definition}

An unbalanced signed graph is always assigned a weight of $0$. Figure~\ref{fig:Status} illustrates the balanced signed graphs in Figure~\ref{fig:BalStates1} grouped by identical balanced states. The weights of these balanced states are equal to $3$, $3$, $1$, and $1$, as indicated by the boxed groupings of the balanced states. The weight captures the frequency of \emph{appearance} of each balanced state using different underlying spanning trees. The weight of the balanced state is a measurement of the likelihood a given consensus will occur, as illustrated in Figure~\ref{fig:Status1}.

\begin{figure}[!ht]
    \centering
    \includegraphics[width=4.5in]{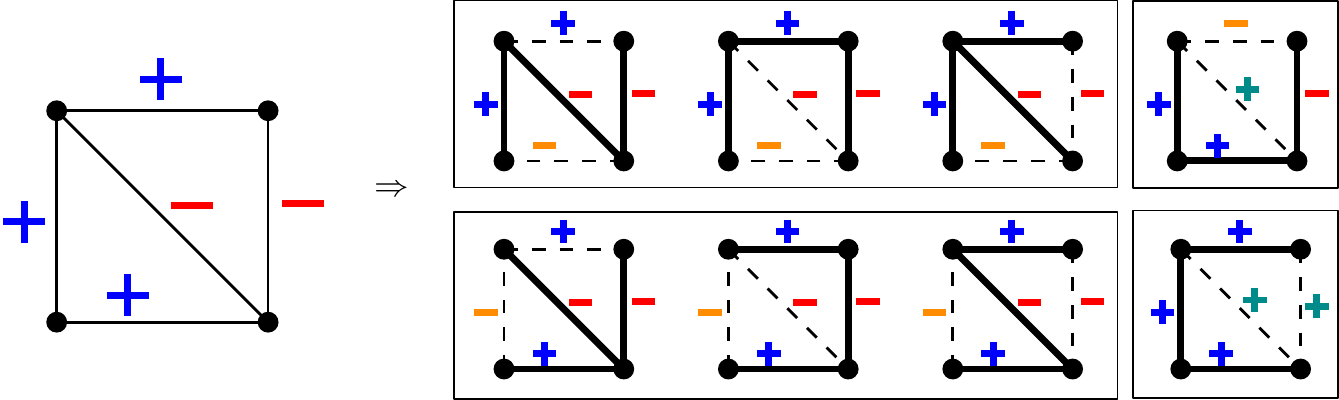}
    \caption{Tree-balancing algorithm (Alg.~\ref{alg:Balance}) on signed graph $\Sigma$ produces $4$ balanced states. Different spanning trees can produce same balanced state.  The edges outside each spanning tree are indicated as dashed lines.}
    \label{fig:Status}
\end{figure}

\begin{figure}[!ht]
    \centering
    \includegraphics[width=4.5in]{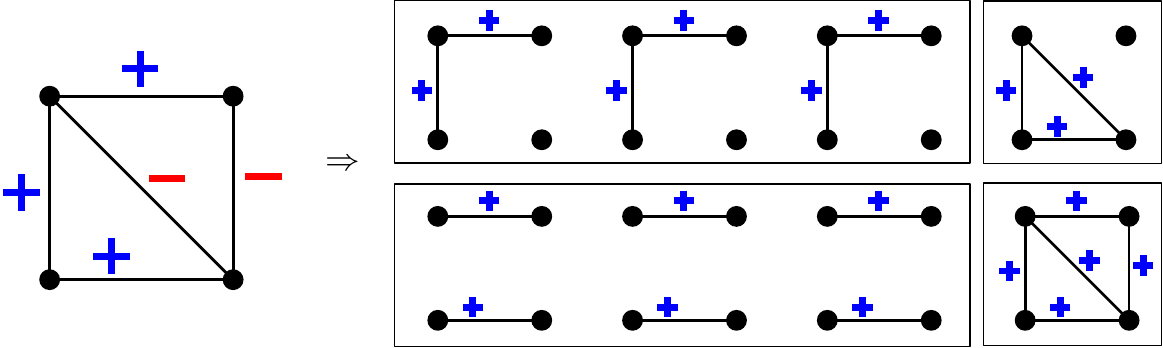}
    \caption{Harary cuts per balanced state: The deletion of the negative edges in each balanced state in the frustration cloud of the signed graph in Figure~\ref{fig:Status}.}
    \label{fig:Status1}
\end{figure}

\subsection{Global Vertex Status}
\label{ssec:status}

Let $G$ be a graph whose set of spanning trees is $\mathcal{T}_{G}$. Given a signed graph $\Sigma = (G,\sigma)$ and a spanning tree $T \in \mathcal{T}_{G}$, recall that $\Sigma'_T$ is a balanced signed graph obtained by the tree-balancing algorithm (Alg.~\ref{alg:Balance}). The bipartition $(U_T,W_T)$ from Theorem~\ref{t:HararyCut} results in two induced subgraphs, as illustrated in Figure \ref{fig:Status1},  $\Sigma'_{U_T}$ and $\Sigma'_{W_T}$.  The subgraphs are named so that $\vert U_T \rvert \leq \vert W_T \rvert$; thus, $\Sigma'_{W_T}$ always has the majority of vertices. \emph{Status} (Defn.~\ref{d:status}) is the likelihood that a majority of the vertices in a network can be convinced to agree with a specific node's position over all nearest balanced states, with multiplicity determined by the weight. 

\begin{definition}
The \emph{status} of a vertex $v$ in $\Sigma = (G,\sigma)$ is defined as the normalized sum of step functions if vertex $v$ is in the larger subgraph $\Sigma'_{W_T}$:
\begin{align*}
status(v)=\frac{1}{|\mathcal{T}_{G}|}\dsum_{T \in \mathcal{T}_{G}}\delta_{\Sigma'_{W_T}}(v), \text{ where }
 \delta_{\Sigma'_{W_T}}(v) =
    \begin{cases}
      1 & \text{if } v \in \Sigma'_{W_T}\\
      0.5 & \lvert W_T \rvert = \lvert U_T \rvert  \text{~~~tie-break}\\
      0 & \text{otherwise.}
    \end{cases} 
\end{align*}
\label{d:status}
\end{definition}

\begin{lemma}
Let $\Sigma = (G,\sigma)$ be a signed graph with frustration cloud $\mathcal{F}_{\Sigma}$. Status can be defined as a sum of step function if vertex $v$ is in larger balanced sub-graph $\Sigma'_{W_T}$, weighted by $w_{\Sigma'}$, the number of spanning trees of G that balance $\Sigma$ into $\Sigma'$.
\begin{align*}
status(v)=\frac{1}{|\mathcal{T}_{G}|}\dsum_{\Sigma' \in \mathcal{F}_{\Sigma}}w_{\Sigma'}\delta_{\Sigma'_{W_T}}(v).
\end{align*}
\end{lemma}
\begin{proof}
From Defn.~\ref{d:weights}, $w_{\Sigma'}$ counts the number of spanning trees that contribute to a given balanced state. Separating these into individual spanning trees and using Defn.~\ref{d:status} gives the result.
\end{proof}

Figure \ref{fig:Status0} uses the components of Figure \ref{fig:Status1} to determine the status.

\begin{figure}[!ht]
    \centering
    \includegraphics[width=4.5in]{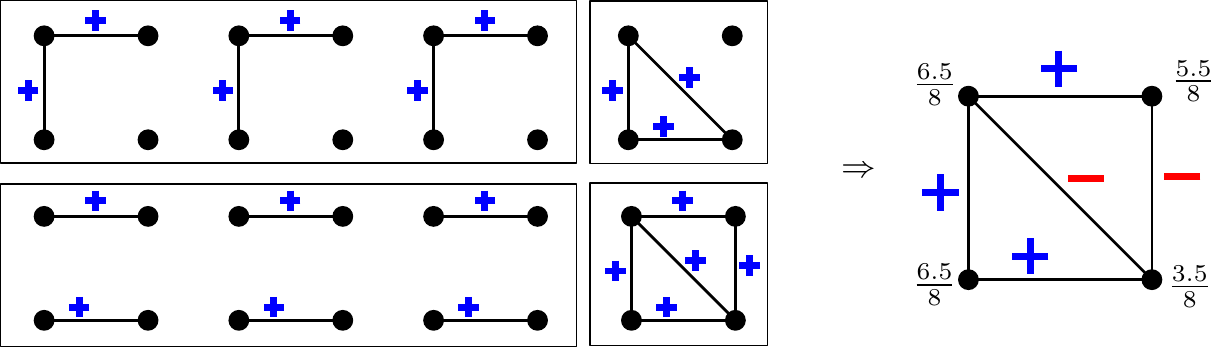}
    \caption{Harary-cut for the nearest balanced states weighted by their occurrence (left) and the calculated status values for signed graph $\Sigma$.} 
    \label{fig:Status0}
\end{figure}

The top-left vertex of Figure \ref{fig:Status0} has status $[3(1) + 1(1) + 3(0.5)  + 1(1.0)]/8 = 6.5/8$; the bottom-left is the same at $6.5/8$; the top right is $[3(1) + 3(0.5) + 1(1)]/8 = 5.5./8$; and the bottom-right is $[1(1) + 3(0.5)  + 1(1.0)]/8 = 3.5/8$. Next, let us consider a \emph{sum} of all statuses of all vertices in a graph and Definition~\ref{d:status}.  

\begin{lemma}
For signed graph $\Sigma = (G,\sigma)$, vertex set $V$, and tree-spanning set $\mathcal{T}_{G}$, the sum of statuses of all vertices in $\Sigma$ equals the normalized sum of cardinality of the larger component of the Harary-cut over all spanning trees $T \in \mathcal{T}_{G}$.
\begin{align*}
\lvert \mathcal{T}_{G} \rvert \dsum_{v \in V} status(v) = \dsum_{T \in \mathcal{T}_{G}} \lvert V(\Sigma'_{W_T}) \rvert.
\end{align*}
\label{c:statusSum}
\end{lemma}
\begin{proof}
Summing both sides of the definition of status over all vertices gives
\begin{align*}
\dsum_{v \in V} status(v) &= \dsum_{v \in V}\frac{1}{\lvert \mathcal{T}_{G} \rvert}\dsum_{T \in \mathcal{T}_{G}} \delta_{\Sigma'_{W_T}}(v) = \frac{1}{\lvert \mathcal{T}_{G} \rvert}\dsum_{v \in V}\dsum_{T \in \mathcal{T}_{G}} \delta_{\Sigma'_{W_T}}(v) \\
&= \frac{1}{\lvert \mathcal{T}_{G} \rvert}\dsum_{T \in \mathcal{T}_{G}}\dsum_{v \in V} \delta_{\Sigma'_{W_T}}(v) = \frac{1}{\lvert \mathcal{T}_{G} \rvert}\dsum_{T \in \mathcal{T}_{G}}\lvert V(\Sigma'_{W_T})) \rvert.
\end{align*}

The last equality holds even in the event of having two components of equal size, as we have defined status. Since $\delta_{\Sigma'_{W_T}}(v)$ treats them as an equal split ($0.5$), there are the same number of vertices in both the new majority as well as the minority --- which is equivalent to counting the size of the tied majority.  The proof is completed by multiplying by $\lvert \mathcal{T}_{G} \rvert$. 
\end{proof}

\subsection{Global Vertex Influence} 
\label{ssec:influence} 

We now consider the variation of attitudinal strength captured by edge signs in attitudinal networks.  When students assign a strong rating score for an instructor evaluation, it is hard to separate affective, behavioral, and cognitive components of the attitude expressed in that one sentiment.  Did the student take all of their other instructors into consideration? What is the subjective evaluation range? How much of the rating is based on students' own subjective performance in the class? How likely is the student to change his mind when he talks to his peers? We examine the strength of the beliefs held within these edges. We propose another new measure for edges similar to status, termed \emph{agreement}, to measure how strongly held a given edge-sentiment is. It is likely that an edge will be positive and contribute to the consensus decision in all near balanced states produced by the tree- balancing algorithm. See Definition~\ref{d:agreement}.

\begin{definition}
The \emph{agreement} of an edge $e$ in a signed graph is a normalized sum of all occurrences of an edge in the largest component of a Harary-cut over all spanning trees. 
\begin{align*}
agreement(e)=\frac{1}{|\mathcal{T}_{G}|}\dsum_{T \in \mathcal{T}_{G}}\delta_{\Sigma'_{W_T}}(e), \text{ where }
 \delta_{\Sigma'_{W_T}}(e) =
    \begin{cases}
      1 & \text{if } e \in \Sigma'_{W_T}\\
      0.5 & \lvert W_T \rvert = \lvert U_T \rvert  \text{~~~tie-break}\\
      0 & \text{otherwise.}
    \end{cases} 
\end{align*}
\label{d:agreement}
\end{definition}

Figure \ref{fig:Agreement} shows agreement for the given example. The larger agreement an edge has, the more likely it will appear in the final, majority decision. 

\begin{figure}[!ht]
    \centering
    \includegraphics[scale=1]{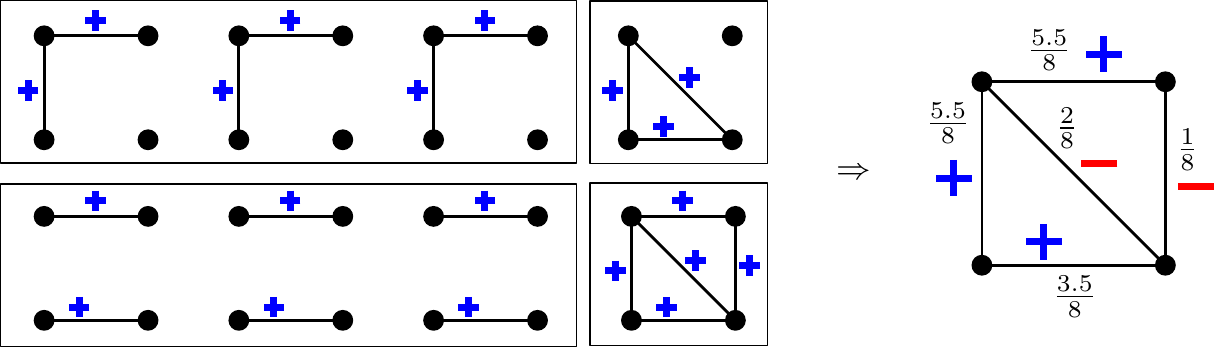}
    \caption{Left: Harary-cut; Right: Edge Agreement values.}
    \label{fig:Agreement}
\end{figure}

Agreement is calculated dually to status. Figure \ref{fig:Status0} demonstrated that the bottom-left and top-left vertices both have the largest status values. However, the agreement in Figure \ref{fig:Agreement} helps to quantify the differences between them. 

Parallel to Lemma \ref{c:statusSum} we immediately have:

\begin{lemma}
For signed graph $\Sigma = (G,\sigma)$, edge set $E$, and tree-spanning set $\mathcal{T}_{G}$, the sum of the agreement of all edges $e, e\in E$, in $\Sigma$ equals the normalized sum of edge cardinality of the larger component of the Harary-cut over all spanning trees $T \in \mathcal{T}_{G}$.
\begin{align*}
\lvert \mathcal{T}_{G} \rvert \dsum_{e \in E} agreement(e) = \dsum_{T \in \mathcal{T}_{G}} \lvert E(\Sigma'_{W_T}) \rvert.
\end{align*}
\label{c:agreeSum}
\end{lemma}

Edge agreement is now averaged around each vertex to compare to status. This metric is called \emph{influence}.

\begin{definition}
The \emph{influence} of a vertex $v$ in a signed graph is the average agreement of all edges incidental to the vertex $v$,
\begin{align*}
influence(v)=\frac{1}{deg(v)}\dsum_{e \sim v}agreement(e).
\end{align*}
\label{d:influence}
\end{definition}

Comparing the influence to the status in our examples, we see that the influence is always less than or equal to status.

\begin{figure}[!ht]
    \centering
    \includegraphics[width=5in]{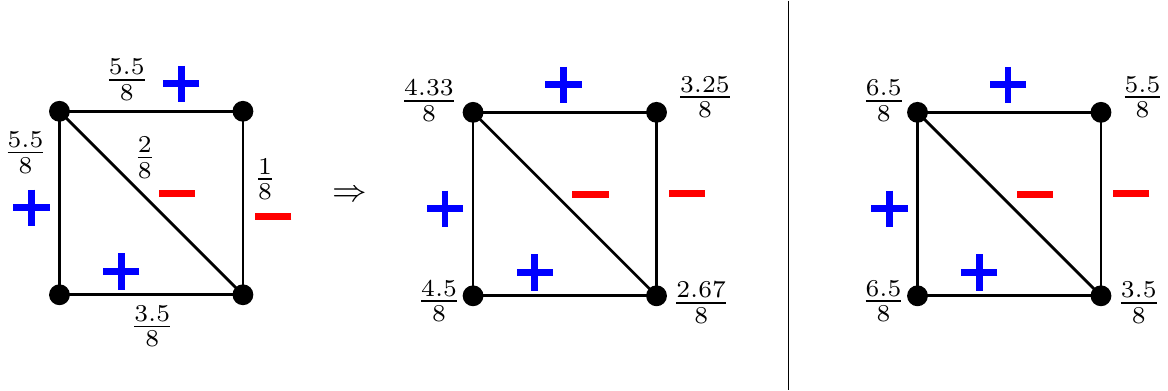}
    \caption{For the signed graph $\Sigma$ in Figure~\ref{fig:BalStates1}, its edge agreement values (left) and vertex influence (center) are compared to its vertex status (right).}
    \label{fig:AAS}
\end{figure}

Status and influence provide two measures of vertex influence in the attitudinal network graph, as illustrated in Figure~\ref{fig:AAS}. The relationship of \emph{status} and \emph{influence} measures for vertex $v$ as outlined in Lemma~\ref{l:StatusCone}. Their relation stems from Defn.~\ref{d:status}, Defn.~\ref{d:influence}, and from comparing the totality of edge counts around each vertex.

\begin{lemma}
For an unbalanced signed graph $influence(v) \leq status(v)$. Moreover, equality holds when $v$ is a pendant vertex whose edge is positive; influence is $0$ when $v$ is a pendant vertex whose edge is negative.
\label{l:StatusCone}
\end{lemma}

\subsection{Conservation of Controversy}
\label{ssec:controversy}

Consensus is a general agreement that can be achieved without unanimous voting. If consensus in the signed graph is unanimous, then the Harary-cut produces one partition consisting of the entire connected graph; the nearest balanced state has all positive edges. On the other end of the spectrum, the nearest balanced state can result in a Harary-cut that has bipartitions of equal size, and the entire graph is deadlocked in indecision. Controversy (from Latin \emph{controversia} meaning "turn in opposite direction") occurs anytime there are conflicting opinions in the group. Controversy in balanced graph states occurs when consensus is achieved but the voting was not unanimous. Every balanced state but one (the all positive signed graph) has a certain level of controversy associated with it. The measure of the average status of the nearest balanced states can quantify controversy for the underlining signed graph, and the graph status definition is Definition~\ref{d:graphStatus}. 

\begin{definition}
Let $status(\Sigma)$ denote the graph status measure, and $\lvert V(G) \rvert$ is the number of vertices in the graph.  Then, an average status of a signed graph $\Sigma$ is defined as
\begin{align*}
status(\Sigma) = \frac{1}{\lvert V(G) \rvert}\dsum_{v \in V} status(v).
\end{align*}
\label{d:graphStatus}
\end{definition}

\begin{lemma}
Let $\Sigma = (G,\sigma)$ be a signed graph, then $
0.5  \leq status(\Sigma) \leq 1$.
\label{t:status2}
\end{lemma}
\begin{proof}
This lemma sets the bounds of status sum in Lemma \ref{c:statusSum}. From the definition of majority, we have that for every spanning tree $T$ we have
\begin{align*}
\frac{\lvert V(G) \rvert}{2} \leq \lvert V(\Sigma'_{W_T}) \rvert \leq  \lvert V(G) \rvert.
\end{align*}
Summing over all spanning trees and normalizing the sum using Lemma \ref{c:statusSum}, we get the result. \end{proof}

From Theorem \ref{c:statusSum} we know $\lvert \mathcal{T}_{G} \rvert \dsum_{v \in V} status(v)$ is a sum of the sizes of the majority, so it must be an integer. The bounds are from Theorem \ref{t:status2} and multiplying by $\lvert \mathcal{T}_{G} \rvert$.  Combining Lemmas \ref{c:statusSum} and \ref{t:status2} we have:

\begin{theorem}
For a signed graph $\Sigma = (G,\sigma)$  and for all spanning trees $T$ of $\Sigma$:
\begin{enumerate}
    \item $status(\Sigma)$ is minimal ($=0.5$) if, and only if, $\lvert V(\Sigma'_{W_T}) \rvert = \lvert V(\Sigma'_{U_T}) \rvert, \forall T$,
    \item $status(\Sigma)$ is maximal ($=1$) if, and only if, $\lvert V(\Sigma'_{W_T}) \rvert = \lvert V(G) \rvert, \forall T$.
\end{enumerate}
\label{t:controversy}
\end{theorem}
%%%%%%%%%%%%%%%%%%%%%%%%%%%%%%%%%%%%%%%%%%%%%%%%%%%%%%%%%%%

We define average status over all vertices in the graph as a measurement of \textbf{controversy} (Theorem \ref{t:controversy}). The maximum value of $status(\Sigma)$ is $1.0$; this is the case when all nearest balanced states have all positive edges, and everyone agrees all the time. The minimum value of  $status(\Sigma)$ is $0.5$, and the balancing consistently splits the set in two equally sized subsets. In between, if $status(\Sigma)$ is closer to $1$, the entire graph has low controversy, and if it is trending to $0.5$, the entire graph has higher controversy.  One way to resolve a tie-break in Section~\ref{ssec:status} is to assign status and agreement values of 0.5 if the Harary-cut bipartitions are equal size.  In the Human Resource Scenario, consider that a ``reliable'' or ``reputable'' vertex exists, and have that person (vertex in signed graph) break all ties in its own favor when the Harary-cut bipartitions are of equal size.  We define \emph{vertical status} in  Definition \ref{d:verticalStatus}.

\begin{definition}
The \emph{vertical status} of a vertex $v$ in $\Sigma = (G,\sigma)$ with respect to designated vertex $t$ is
\begin{align*}
status_{t}(v)=\frac{1}{|\mathcal{T}_{G}|}\dsum_{T \in \mathcal{T}_{G}}\delta^t_{\Sigma'_{W_T}}(v),~~~
 \delta^t_{\Sigma'_{W_T}}(v) =
    \begin{cases}
      1 & \text{if } v \in \Sigma'_{W_T},\\
      1 & \lvert W_T \rvert = \lvert U_T \rvert, \text{~and~}  v, t \text{~in the same partition,}\\
      0 & \text{otherwise.}
    \end{cases}    
\end{align*}
\label{d:verticalStatus}
\end{definition}

Definition~\ref{d:verticalStatus} states that all tie-breaks increase the status of vertex $t$ and vertices in the same subset as $t$. Figure \ref{fig:StatusEx} illustrates how vertical status compares to status for the signed graph from Figure \ref{fig:BalStates1}.  Let us consider the top-left vertex (now a closed box) as the tie-breaker in Figure \ref{fig:StatusEx} in case 1, and the bottom-right vertex (now an open box) as the tie-breaker in case 2. \emph{Case 1}: the top-left vertex is used to break any ties, and the vertical status values are now $(8/8, 7/8, 2/8,5/8)$, as illustrated in Figure \ref{fig:StatusEx}(middle). \emph{Case 2}: the bottom-right vertex is used to break ties, and the vertical status values are now $(5/8,4/8,5/8,8/8)$. In both cases, the status of the chosen vertex increased over the original status in Figure~\ref{fig:Status0}.

\begin{figure}[!ht]
    \centering
    \includegraphics[scale=1]{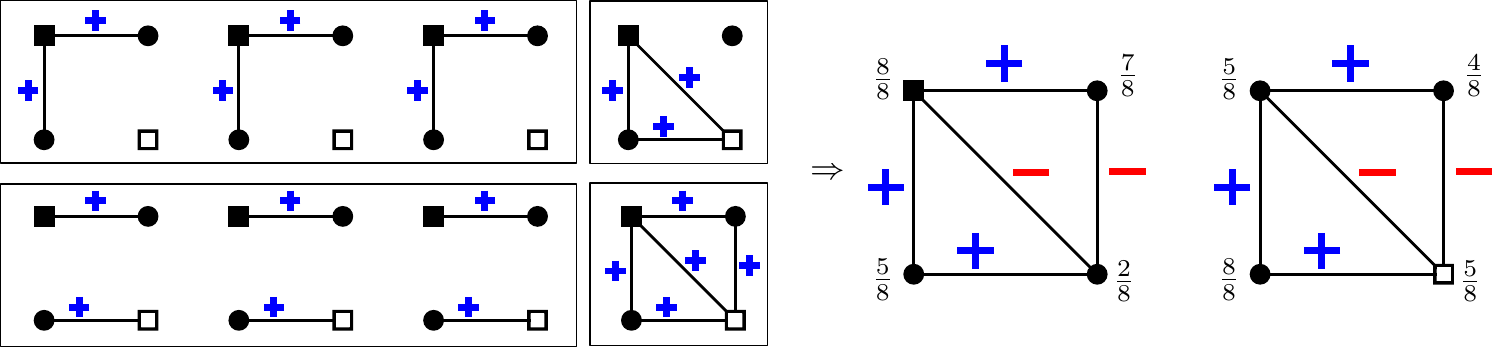}
    \caption{Calculating the vertical status using the top-left vertex (middle); and bottom-right vertex (right).}
    \label{fig:StatusEx}
\end{figure}

\begin{lemma} For a signed graph $\Sigma = (G,\sigma)$ and vertex $v \in V$, the $status_t(v)$ is maximized when $t=v$.
\label{l:maxStatus}
\end{lemma}
\begin{proof}
If $t=v$, vertex determines its own tie-breakers, and every $0.5$ in $\delta_{\Sigma'_{W_T}}(v)$ is replaced with a $1$ in $\delta^t_{\Sigma'_{W_T}}(v)$.
\end{proof}

\begin{definition}
Let $status_t(\Sigma)$ denote the average vertical status of a signed graph $\Sigma$ with distinguished vertex $t \in V$.
\begin{align*}
status_t(\Sigma) = \frac{1}{\lvert V(G) \rvert}\dsum_{v \in V} status_t(v)
\end{align*}
\label{d:tStatus}
\end{definition}

The average status and the average vertical status over the whole signed graph is a constant, called the \emph{controversy} of the signed graph. This constant provides the following Conservation of Controversy Law.

\begin{theorem}[Conservation of Controversy Law:]  For a signed graph  $\Sigma = (G,\sigma)$, graph controversy is equal to its status and any vertical status: 
\begin{align*}
    controversy(\Sigma) = status(\Sigma) = status_t(\Sigma), \forall t \in V(G).
\end{align*}
\label{t:LawOfCC}
\end{theorem}

\begin{proof}
As in Theorem \ref{c:statusSum}, 
\begin{align*}
\dsum_{v \in V} status_t(v) &= \dsum_{v \in V}\frac{1}{\lvert \mathcal{T}_{G} \rvert}\dsum_{T \in \mathcal{T}_{G}} \delta^t_{\Sigma'_{W_T}}(v) = \frac{1}{\lvert \mathcal{T}_{G} \rvert}\dsum_{v \in V}\dsum_{T \in \mathcal{T}_{G}} \delta^t_{\Sigma'_{W_T}}(v) \\
&= \frac{1}{\lvert \mathcal{T}_{G} \rvert}\dsum_{T \in \mathcal{T}_{G}}\dsum_{v \in V} \delta^t_{\Sigma'_{W_T}}(v) = \frac{1}{\lvert \mathcal{T}_{G} \rvert}\dsum_{T \in \mathcal{T}_{G}}\lvert V(\Sigma'_{W_T}) \rvert = \dsum_{v \in V} status(v).
\end{align*} 
The second to last equality is a strict count of the size of the majority, while the last equality is from Lemma \ref{c:statusSum}.
The proof is completed by dividing by $\lvert V \rvert$.
\end{proof}

The \textbf{Conservation of Controversy Law} in Theorem~\ref{t:LawOfCC} states that the average status is equal to any average vertical status. Interpreting average status as controversy, we can conclude that \emph{the level of controversy is independent of vertex preference}, but the individual status values may change. The controversy in Figure~\ref{fig:Status0} (right) is $0.6875$, in Figure~\ref{fig:StatusEx} (middle) is $0.6875$, and in Figure~\ref{fig:StatusEx} (right) is $0.6875$.  The vertical status increase of the chosen vertex in Figure \ref{fig:StatusEx} is \emph{at the cost of status values of other vertices}, so the overall controversy stays the same. Controversy is one of the most important concepts in this paper, as it quantifies the level of controversy in a graph as a whole and does not depend on the tie-breaker decision, as proven by the Conservation of Controversy Law. 

\section{graphB: Spanning Tree-Sampling Balancing Algorithm}
\label{ssec:practical}

The proposed measures of status \ref{d:status}, agreement \ref{d:agreement}, and controversy \ref{t:controversy} of a signed graph require the computation of \emph{all} spanning trees for the underlying unsigned graph $G$. For a real life socio-technical network, this is computationally prohibitive.  In this section, we propose to utilize existing spanning tree graph discovery and sampling methods to accurately model measures derived from {\emph all} spanning trees as outlined in Section~\ref{sec:consensus} with a {\emph sample} of spanning trees \cite{AI2009}. Note that the Conservation of Controversy Law (Theorem \ref{t:LawOfCC}) holds for any fixed subset of spanning trees; while different subsets of trees may produce different controversy values, the conservation is in tie-break scenarios within the sample. 

The upper limit on the number of spanning trees is computed by Cayley's theorem: the complete graph with $v$ vertices has $v^{v-2}$ spanning trees, and a complete bipartite graph with $v,q$ vertices has  $v^{q-1} \cdot q^{v-1}$ spanning trees \cite{OHMTT,BUEKENHOUT1998}. For a small social network such as the Highland Tribes \cite{1954Read} in Figure~\ref{fig:Highland} with $16$ vertices and $29$ positive and $29$ negative edges, that number is quite high $16^{14} = 7.21e+16$. The exact number of spanning trees for any graph $G$ can be calculated in polynomial time as the determinant of a matrix derived from the graph, using Kirchhoff's matrix-tree theorem \cite{Tutte}. The Tutte polynomial of a graph can be defined as a sum, over the spanning trees of the graph, of terms computed from the "internal activity" and "external activity" of the tree. Its value at the arguments (1,1) is the number of spanning trees \cite{Tutte}, and the computed number of spanning trees for the Highland Tribes is $402,506,278,163$. Computing probabilistic consensus measures as outlined in Section~\ref{sec:consensus} will require running Algorithm~\ref{alg:Balance} over 400 billion trees, and that is computationally prohibitive. We examine statistical samples of spanning trees to approximate modeling of balanced state coverage with $k$ sampled spanning trees. We propose scaled adjustment of Alg.~\ref{alg:Balance} termed {\bf graphB}: Spanning Tree-Sampling Balancing Algorithm, and we implement the proof-of-concept \cite{graphB}. The efficiency improvements on the tree-based balancing algorithm are outlined in Alg.~\ref{alg:Sampling}, and its implementation complexity is addressed in Sec.~\ref{ssec:complexity}. We model the frustration cloud and balanced state weights using $n$ spanning trees to compute status, influence, and controversy.

\begin{algorithm}[!ht]
  \caption{graphB: Spanning Tree-Sampling Balancing Algorithm:}
  \label{alg:Sampling}
  \begin{algorithmic}
    \Require Input signed graph  $\Sigma = (G,\sigma)$. 
    \Require Sample $n$ spanning trees $T$ to $\mathcal{T_k}$ set using BFS, random or DFS. 
    \ForAll {$i \in [1,n]$,  $\Sigma_i, T_i \in \mathcal{T_k}$}
     \ForAll  {edges $e$,  $e \in \Sigma \setminus T$}
            \If {fundamental cycle $T_i \cup e$ is negative}
            \State       change edge sign: $e^- -> e^+; e^+ -> e^-$
            \EndIf
            \EndFor 
            \State   Construct new balanced signed graph $\Sigma_i$
            \State Create Harary cutset for $T_i$, $(U_i, W_i)$
        \EndFor
        \Ensure Compute status for each vertex and agreement for each edge
        \end{algorithmic}
\end{algorithm}

Algorithm~\ref{alg:Sampling} includes sampling a tree step; instead of looping over all spanning trees, we select a subset of spanning trees $\mathcal{T_n}$  that contains $n$ spanning trees of signed graph $\Sigma$.  Given a fixed number of vertices and edges, path-like trees have minimal eigenvalues while star-like trees have maximal eigenvalues, per Lovasz' eigenvalue characterization for trees \cite{Lovasz1}. Next, we analyze three strategies for sampling spanning trees in the graphB algorithm (Alg. \ref{alg:Sampling}) w.r.t. eigenvalues and the frustration cloud. A breadth-first sampling search favors star-like spanning trees (max eigenvalues),  a depth-first spanning tree search favors path-like spanning trees (min eigenvalues), and a random tree selection is used as a baseline as it is the most efficient (random). We examine the sensitivity of our proposed method using random, breadth-first, and depth-first spanning tree samplings and demonstrate how status, influence, and controversy can be perceived in different paradigms. 

{\bf Random Sampling Baseline:} A \emph{uniform spanning tree} is a tree chosen randomly from among all the spanning trees with equal probability, and there are multiple known implementation algorithms. We assumed that a random sample will best represent the frustration cloud and multiplicity. The fastest implementation is a \emph{random minimal spanning} tree algorithm; it generates random trees, but cannot guarantee sampling uniformity.  The edges of the graph are assigned random weights, and then the minimum spanning tree of the weighted graph is constructed. We compared the implementation of three different algorithms in NetworkX \cite{NetworkX} for discovery of the \emph{minimal spanning tree}, namely DJP, Boruvka, and Kruskal's algorithms \cite{AI2009,NetworkX}.  They are all greedy algorithms that run in polynomial time, and we have not observed much difference in speed or randomization when selecting a specific one.\\
{\bf Breadth-first search (BFS):} BFS searches for spanning trees by progressively exploring all the neighborhood vertices from the starting vertex (search key) at present depth before moving to next depth level.  As a result, breadth-first spanning trees are the most star-like trees in the graph, and they represent the classes of trees that have maximal eigenvalues \cite{Lovasz1}. The BFS algorithm has been used to find the shortest path between two vertices in a graph measured by number of edges, as it allows for the discovery of the shortest fundamental cycles in a graph. Thus, spanning trees resulting from a breadth-first search have the maximal number of pendant vertices and fundamental cycles of minimal length \cite{AI2009}. We use NetworkX \cite{NetworkX} implementation of BFS for proof-of-concept. \\
{\bf Depth-first search (DFS):} DFS searches for spanning trees by exploring the branch to the highest depth level possible before backtracking and expanding \cite{AI2009}, effectively delaying cycle feedback as long as possible. Trees determined by a depth-first approach are the most path-like trees with the minimal number of pendant vertices. The process maximizes the length of the fundamental cycle. The DFS algorithm has been used in determining the number of connected components in a graph \cite{AI2009}. Also, the DFS spanning tree search algorithm maximizes the sampling of path-like trees, the classes of trees that have minimal eigenvalues \cite{Lovasz1}, and the fundamental cycle length. 

Random spanning tree sampling provides a straightforward way to analyze the frustration cloud, as context-driven algorithms such as breadth- or depth-first alter the resolution of the data. We conjecture that breadth-first tree sub-sampling will provide the greatest data resolution for computing measures in Section~\ref{sec:consensus}, while depth-first tree sub-sampling will produce a noisy interpretation of the same data. We demonstrate the validity of this assumption on a larger Wikipedia administrator dataset in Section~\ref{ssec:Wiki}. 

%Larger eigenvalues are used for power iteration, while smaller eigenvalues are used in spectral clustering. 
 % frustration cloud technique can be adapted to \emph{simultaneously study} both types of problems, and even indicate whether a dataset is prone to clustering or power iteration techniques. 
 
\subsection{graphB Implementation and Complexity Analysis}
\label{ssec:complexity}

The graphB algorithm implementation in Python is released as open source \cite{graphB}. The algorithm implementation has the overall complexity of $O(n \cdot v \cdot e)$ for run time and $O(v^2v)$ for memory consumption, where $kn$ is the number of sampled trees, $v$ is the number of vertices in the signed graph, and $e$ is the number of edges. We keep the adjacency matrix in memory for the entire process ($O(v^2)$). In the pre-process step, we symmetrize the adjacency matrix ($O(v^2)$), find connected components ($O(v+e)$), sort by the number of neighbors in the largest connected component, and write it to a file ($O(v^2)$). In the process step, we find $n$ trees $(O(n \cdot e))$. Then for each tree, we find all fundamental cycles in the graph, as there exists a cycle containing  \emph{e} if, and only if, there exists a fundamental cycle with respect to the selected spanning tree  \emph{T} that contains  \emph{e}. We have adopted a linear time algorithm for finding articulation points \cite{Farina2015} in  $(O(v+e))$ time for a single spanning tree. The complexity of the entire processing step is $O(n \cdot v \cdot e)$. Note that algorithm implementation does not make any assumptions on the signed graph (sparsity, planarity). The only assumption graphB implementation makes is that the graph edge weights are either +1 or -1. 

\section{Proof of Concept} 
\label{ssec:setup}

The computation of frustration cloud-based measures for signed graphs is implemented using Python and Python libraries, and the code used for proof-of-concept is released on GitHub \cite{graphB}.  The analysis of the largest connected component, spanning tree search methods, and statistics is computed using the  NetworkX \cite{NetworkX} package. Experiments are run on the Texas State University LEAP system \cite{LEAP}: Dell PowerEdge C6320 cluster node consists of two (14-core) 2.4 GHz E5-2680v4 processors 128 GB of memory each, and two 1.5TB memory vertices with four (18-core) 2.4 GHz E7-8867v4 Intel Xeon processors \cite{LEAP}. The LEAP system allows  us to scale the data analysis to support the sampling of $n=1000$ spanning tree computations and to demonstrate the feasibility of computation on larger graph datasets \cite{snapnets}.  In the \emph{graphB} Alg.\ref{alg:Sampling} pipeline, spanning trees are generated over the dataset and saved in h5 format; we use the NetworkX \cite{NetworkX} implementation of random minimal tree, breadth-first, and depth-first tree discovery. Next, for each generated spanning tree, the balancing algorithm is executed on the edges not in the spanning tree (Alg.~\ref{alg:Sampling}) for more details. We obtained a list of unique paths encompassing the spanning tree and given edge and checked for fundamental cycles. If the product of the cycles is $-1$, then the given edge completing the cycle by changes signs, resulting in a balanced cycle.  We repeat the process for each edge. Once all edges outside the spanning tree were visited and edge signs persisted or flipped, the resulting state was a balanced state of the graph. Next, we take a Harary-cut and split the graph into two components. We repeat the process for $k=1000$ trees (Alg.~\ref{alg:Sampling}). The final step is computing the status for each vertex and the agreement for each edge as the normalized sum over  \emph{k} sampled trees, per Defn.~\ref{d:status} and Defn.~\ref{d:agreement}. Vertex influence is then computed as the normalized sum over the sampled $1000$ trees per Defn~\ref{d:influence} and the controversy of the entire graph per Theorem~\ref{t:controversy}. 

\subsection{Wikipedia Administratorship Election Data} 
\label{ssec:Wiki} 

Stanford Network Analysis Project's (SNAP) repository of network data provides good proof-of-concept access to attitudinal network graphs \cite{snapnets}. \emph{Wikipedia} administrator election data represents votes by Wikipedia users in elections for promoting individuals to the role of administrators from July 2004 to January 2008. Wikipedia administrators are editors who have been granted the ability to perform special tasks. The dataset contains 7118 users (vertices) and 103,747 votes (edges) over 2794 elections with one election per candidate, and the outcome of the elections. Out of  2794 elections, 1235 resulted in the promotion to administrator ($44.2\%$), and 1,559 elections did not result in the promotion of the candidate. On the editorial side, $3464$ editors cast zero or $1$ vote over all elections, $5506$ editors cast under 10 votes, and $1612$ editors voted 10 times or more.  Administrators are chosen through a \emph{community review process} that seeks \emph{consensus} is not a majority rule, as the \emph{editor in charge} reviews editors' votes and rationale.

The signed graph is constructed  from Wikipedia administrator election data so that each vertex represents an editor or nominee; if both are running for multiple administrator positions (this is possible as there are different Wikipedia sections), we represent one user with multiple vertices, where each vertex has a final outcome (winner, loser, editor). The edges in the graph model the vote of support or initial nomination ($+1$) or a vote of opposition ($-1$).  We ignore the neutral votes as they are equivalent to no-vote or ambivalent votes per \cite{Abelson1958}.  For $k$ spanning trees and balanced states, $k$ Harary cutsets are found from balanced states as in  Alg.~\ref{alg:Sampling}. Sampled status and influence are computed as follows: 

\begin{definition}
For a signed graph $\Sigma = (G,\sigma)$ and subset of size $k$ of spanning trees $\mathcal{T_k} \subseteq \mathcal{T}_{G}$, the sampled status, agreement, influence, and controversy are computed as:
\begin{enumerate}
    \item $status(v)=\frac{1}{k}\dsum_{T \in \mathcal{T_k}\delta_{\Sigma'_{W_i}}(v)}, \text{ where } \delta_{\Sigma'_{W_i}}(v)$ is defined in Defn.~\ref{d:status},
    \item $agreement(e) = \frac{1}{k}\dsum_{T \in \mathcal{T_k}\delta_{\Sigma'_{W_i}}(e)} \text{ where }
 \delta_{\Sigma'_{W_i}}(e)$ is defined in Defn.~\ref{d:agreement},
    \item $influence(v)=\frac{1}{deg(v)}\dsum_{e \sim v}agreement(e)$, 
    \item $controversy(\Sigma) = \frac{1}{\lvert V \rvert}\dsum_{v \in V}status(v)$.
\end{enumerate}
\label{d:computation}
\end{definition}

\paragraph{Connected components: }  $7066$ users (vertices), or  $99.3\%$ of all vertices, and $103663$ of all votes (edges), or $99.97\%$ of all edges, belong to the largest connected component of the constructed attitudinal graph.  Only $52$ users and $26$ votes are not in largest connected component, and they represent unsuccessful nominations that fail to gather significant votes. We examine the Wikipedia dataset using a sample size of $k$ spanning trees, where $k \in \{10,100,1000\}$. 

\subsubsection{Experiment: Spanning Tree Discovery}
\label{sssec:Exp1}

First, we measure the status and influence for three different spanning tree discovery techniques on the Wikipedia dataset, for $k=1000$ spanning trees are compared: \emph{random} trees as determined by minimal spanning trees with random edge weights, \emph{breadth-first} trees with a random initial vertex, and \emph{depth-first} trees with a random initial vertex. We compute status and influence scores for all participants (vertices) in the Wikipedia election data.  First, we analyze the data in id-status and id-influence space, and we color editors as black triangles and nominees as yellow circles. Whether a vertex is an editor or nominee \emph{is not used} to determine status and influence scores; this information is only used in data analysis and the visualization step. 

Results for status and influence for Wikipedia data editors and nominees for three spanning tree discovery techniques are shown in Figures~\ref{fig:WikiAllTree} and \ref{fig:WikiAllTree1}. Both the editor and nominee means coupled with a $1$ standard deviation band are shown as solid and dashed lines, respectively. The computed {\emph status} of the editors (black triangles) and nominees (yellow circles) appear in Figure~\ref{fig:WikiAllTree}. 

\begin{figure}[!ht]
    \centering
    \includegraphics[width=5.5in]{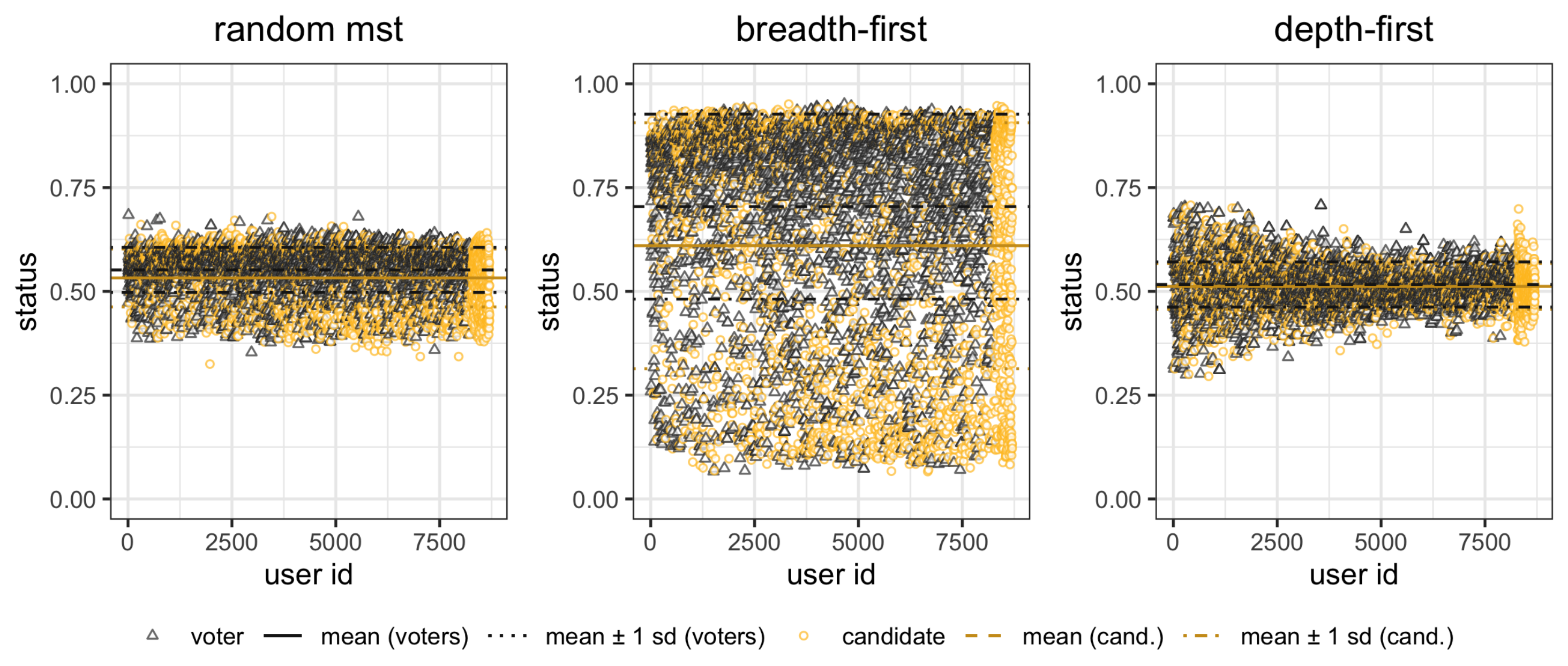}
    \caption{Status of editors (black triangles) and nominees (yellow circles) in Wikipedia administrative election dataset resulting from different tree sampling methods: Random minimal spanning tree (left), breath-first (center), and depth-first (right).}
    \label{fig:WikiAllTree}
\end{figure}

Status score distribution for each sampling methodology in Figure~\ref{fig:WikiAllTree} produces similar score distribution for editors and for nominees, as $1$-SD bands for editors and nominees are close for each sampling method. By observing status measure only, one may conclude the Wikipedia election processes are fair based on the likelihood of landing in a majority (evenly spread out between voters and votees).  Status does not discriminate between voters and nominees in Figure~\ref{fig:WikiAllTree}.  Next, let us examine the influence measure of the editors (grey) and nominees (yellow) in Figure~\ref{fig:WikiAllTree1}. The editors display a much larger influence value than the nominees. The influence score clearly separates high influence individuals (voters) from low influence individuals (votees) in the network.  Note that the conjecture from Section~\ref{ssec:practical} is shown valid in practice, as the breadth-first search provides the highest resolution of status and influence metric analysis in both figures. The depth-first experiment consistently produces a biased sample in terms of balanced states and the frustration cloud, and our experiments on other datasets consistently confirm the conjecture.

\begin{table}[!ht]

\caption{Spanning three discovery methods comparison on Wiki data:}
\label{t:Summary}       % Give a unique label
\begin{tabular}{lllll}
\hline\noalign{\smallskip}
Type & Mean Status (Controversy) & St.Dev status & Mean Influence & St. Dev influence   \\
\noalign{\smallskip}\hline\noalign{\smallskip}
breadth-first & \textbf{0.6693} & 0.2564 & 0.5178 & 0.2823\\
random & 0.54446 & 0.06117 & 0.3465239 & 0.15985\\
depth-first & \textbf{0.5149} & 0.0547 & 0.3067 & 0.1528 \\
\noalign{\smallskip}\hline
\end{tabular}
\end{table}

Overall data statistics are summarized in Table~\ref{t:Summary}. Note that by design, depth-first (DPS) and breadth-first (BFS) represent lower and upper bounds of controversy (mean status) that the experiment confirms. Controversy is computed for all three tree discovery strategies, and the breadth-first value is \textbf{0.6693} while the depth-first value is \textbf{0.5159}. Controversy, as defined in Theorem~\ref{t:controversy}, is a constant value when all spanning trees are accounted for, and breadth-first and depth-first give us the range estimate of the actual value if all spanning trees are considered.  Controversy is a relative measure of the attitudinal network graph, and if compared to another graph, the same spanning tree sampling method must be used for valid comparison.

\begin{figure}[!ht]
    \centering
    \includegraphics[width=5.5in]{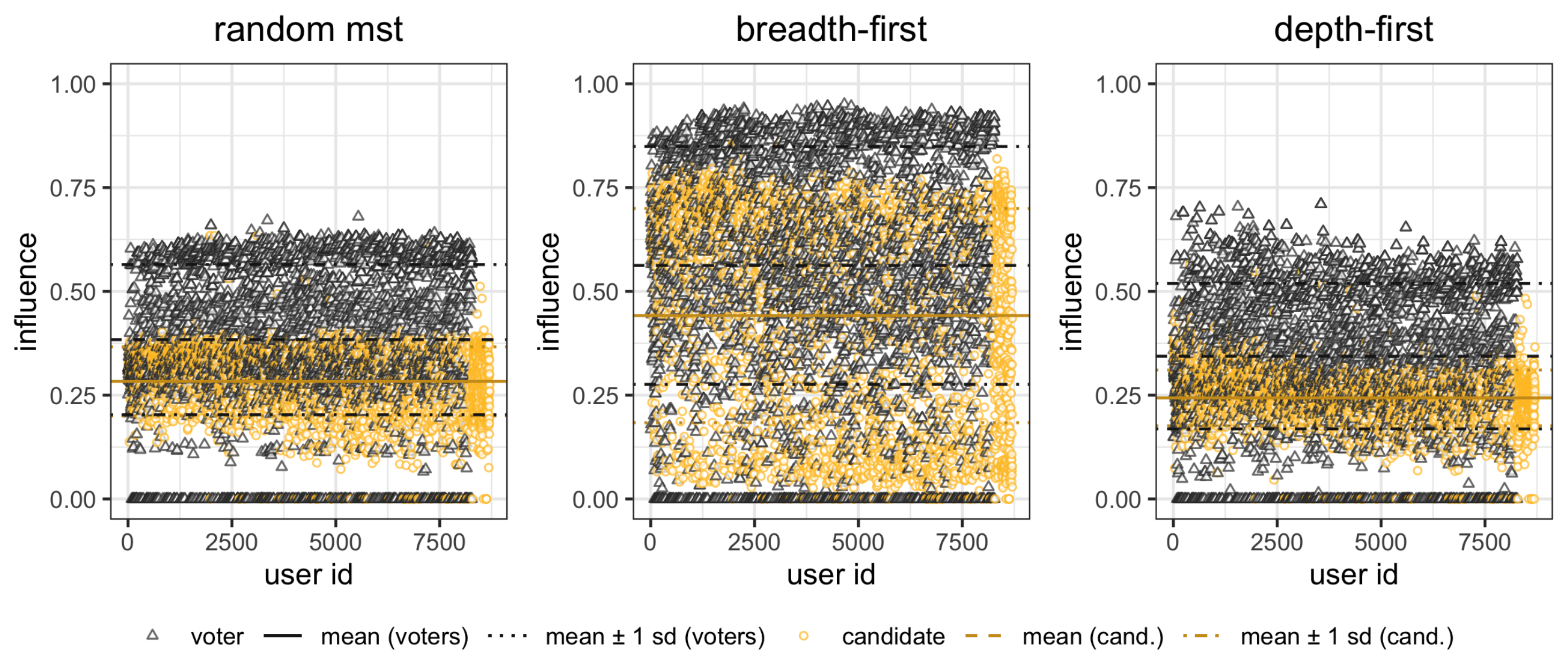}
    \caption{Influence of editors (black triangles) and nominees (yellow circles) in Wikipedia administrative election dataset resulting from different tree sampling methods: Random minimal spanning tree (left), breadth-first (center), and depth-first (right).}
    \label{fig:WikiAllTree1}
\end{figure}

Next, we provide analysis of only the \emph{nominees} (yellow circles). While restricted to the nominees and using the known outcomes of the votes, we examine the efficacy of our method. We re-color the id-status and id-influence graphs with nominee outcomes to illustrate the effectiveness and difference in the metrics as follows: blue triangles represent a positive outcome (promoted to administrator) and red circles represent a negative outcome (not promoted or withdrew its nomination).  This is depicted in Figure~\ref{fig:WikiNoTree} and captures the different measures of influence and status present. High status individuals are mostly the nominees that won the administrator elections, and low status individuals are mostly the nominees that did not win the elections for random and BFS spanning tree sampling. The sensitivity of status measure for nominees and its strong relation to outcome makes status a good predictor of promotability in random and breadth-first spanning tree discoveries. For DFS, the resolution of the status is too low to make any conclusion. 

\begin{figure}[!ht]
    \centering
    \includegraphics[width=5.5in]{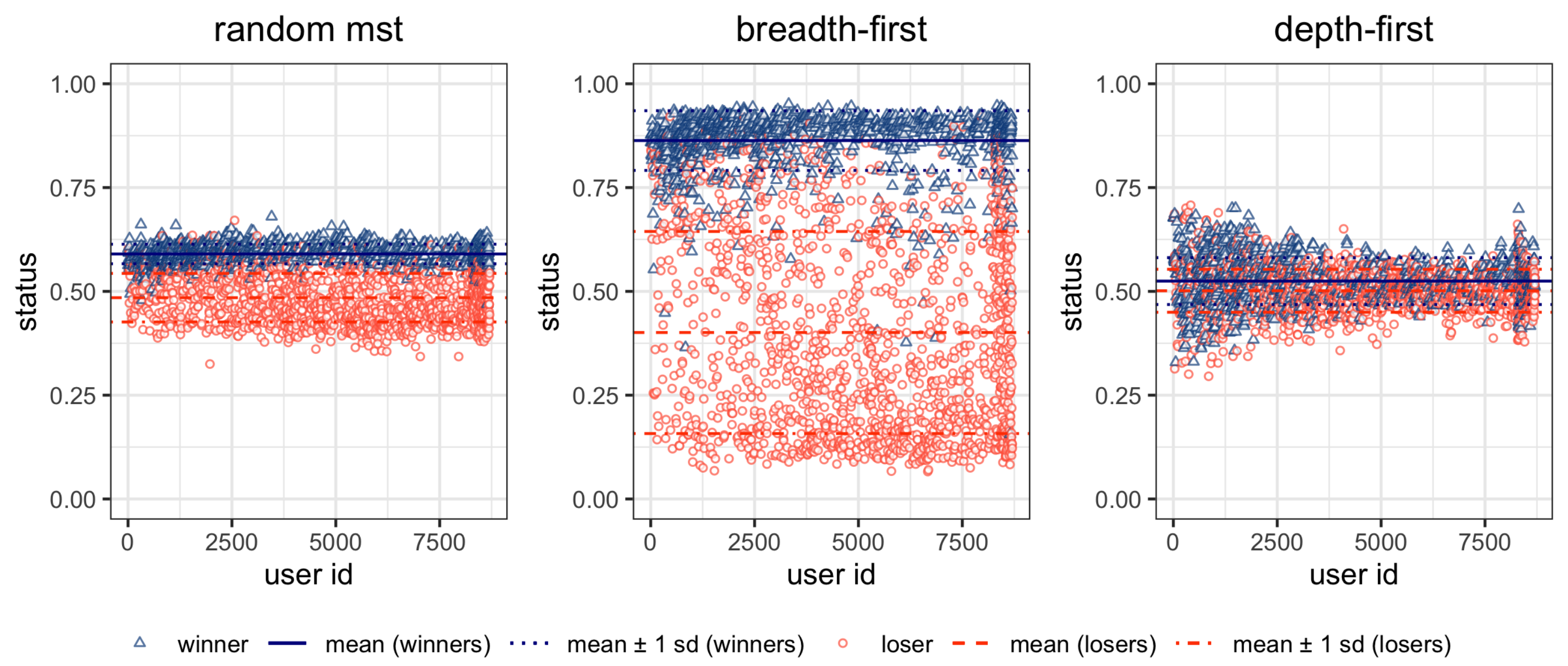}
    \caption{Status of winners (blue triangles) and losers (red circles) in Wikipedia administrative election dataset resulting from different spanning tree discovery methods: Random minimal spanning tree (left), breadth-first (center), and depth-first (right).}
    \label{fig:WikiNoTree}
\end{figure}

Influence score distribution colored by outcome in Figure~\ref{fig:WikiNoTree1} provides a better separation of winners and losers, even for the DFS sampling method. In a promotional network such as a Wikipedia election, status does not distinguish between editors and nominees, but influence does. Nominee-only analysis (people that have been voted on) illustrates a high correlation between status and influence scores prediction. Upon further analysis of Figure~\ref{fig:WikiNoTree} and Figure~\ref{fig:WikiNoTree1}, we identify clear promotional outliers: nominees with high status/influence that did not win the nomination. This deserves further study on anomalous promotion and case-by-case analysis. 

\begin{figure}[!ht]
    \centering
    \includegraphics[width=5.5in]{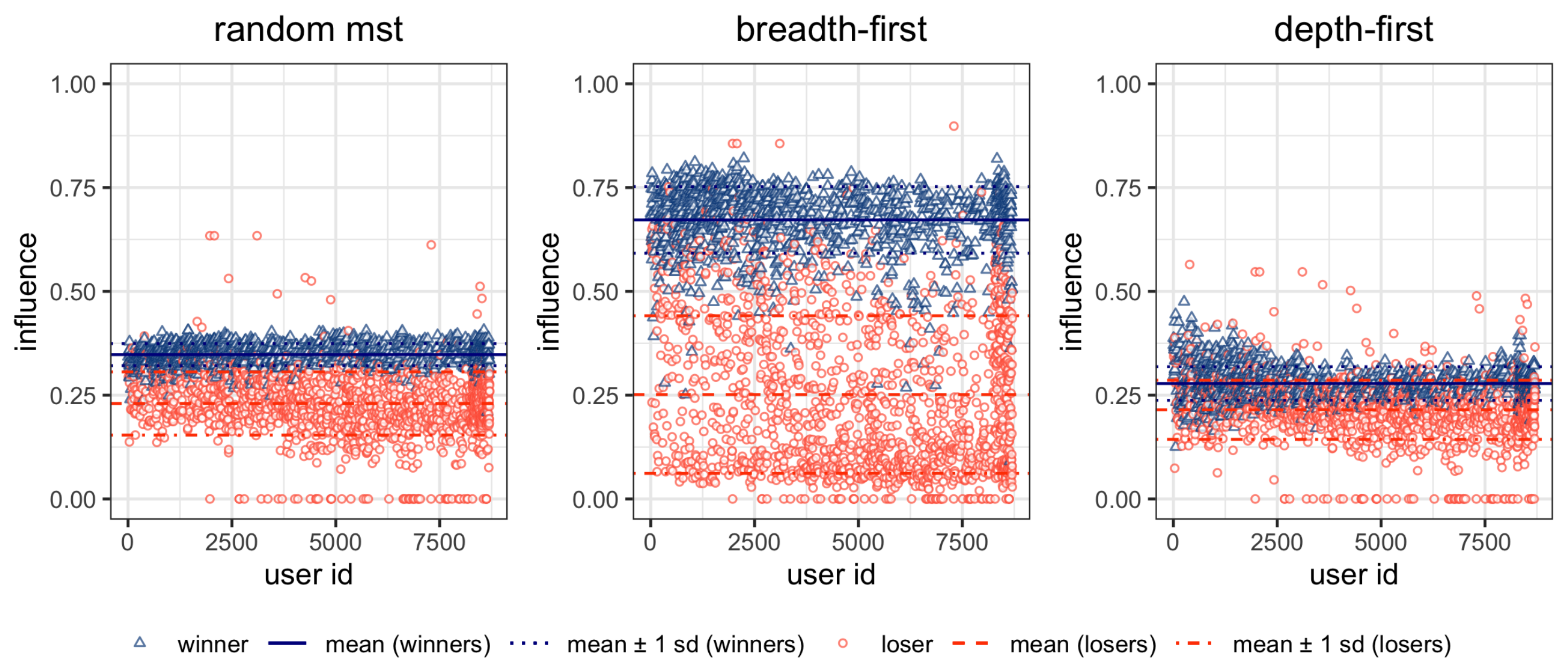}
    \caption{Influence of winners (blue triangles) and losers (red circles) in Wikipedia administrative election dataset resulting from different spanning tree discovery methods: Random minimal spanning tree (left), breadth-first (center), and depth-first (right).}
    \label{fig:WikiNoTree1}
\end{figure}

Depth-first spanning tree discovery strategy does not allow for a reliable representation of the statistical significance of balanced states. We exclude the depth-first spanning tree search in subsequent experiments for status and influence measure and focus on random and breadth-first search spanning tree sampling. 

\subsubsection{Experiment: Sufficient Number of Spanning Trees }
\label{sssec:Exp2}

\begin{figure}[!ht]
    \centering
   \includegraphics[width=5.4in]{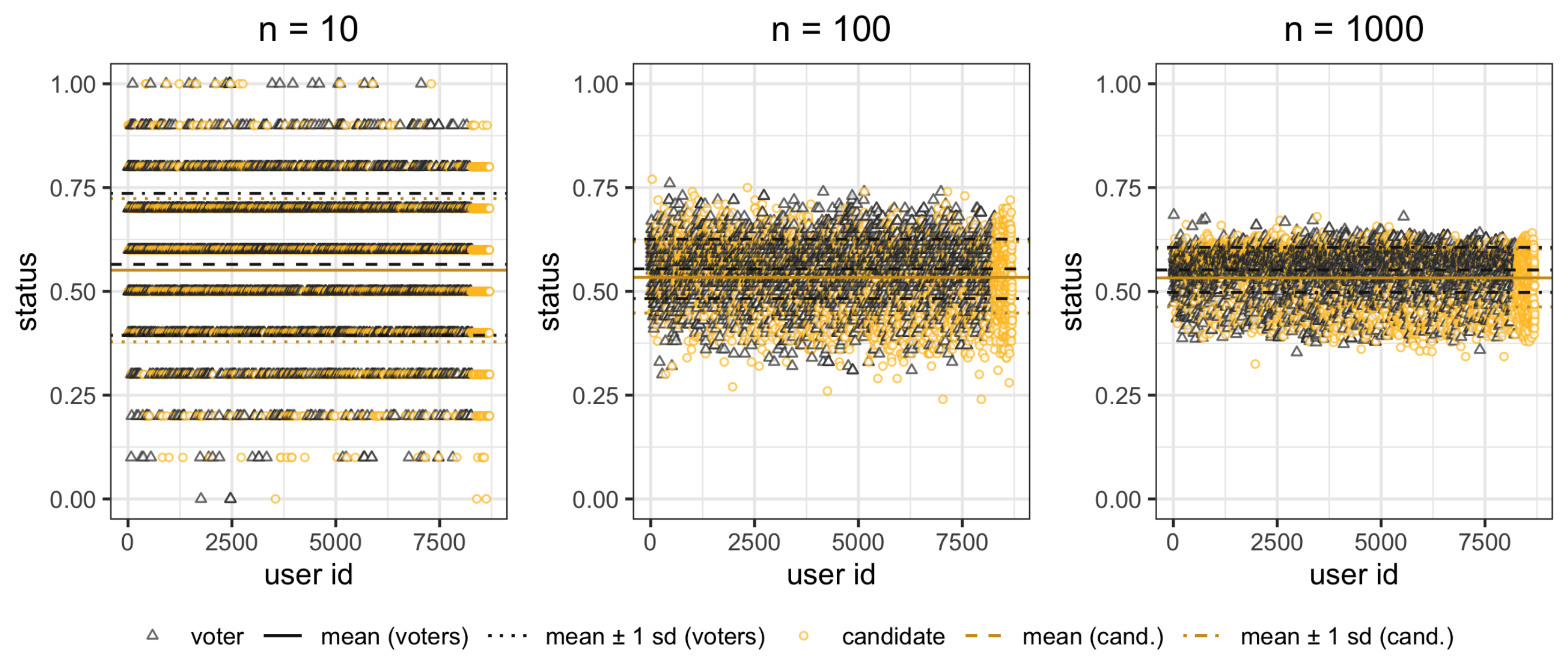} \\
 \includegraphics[width=5.4in]{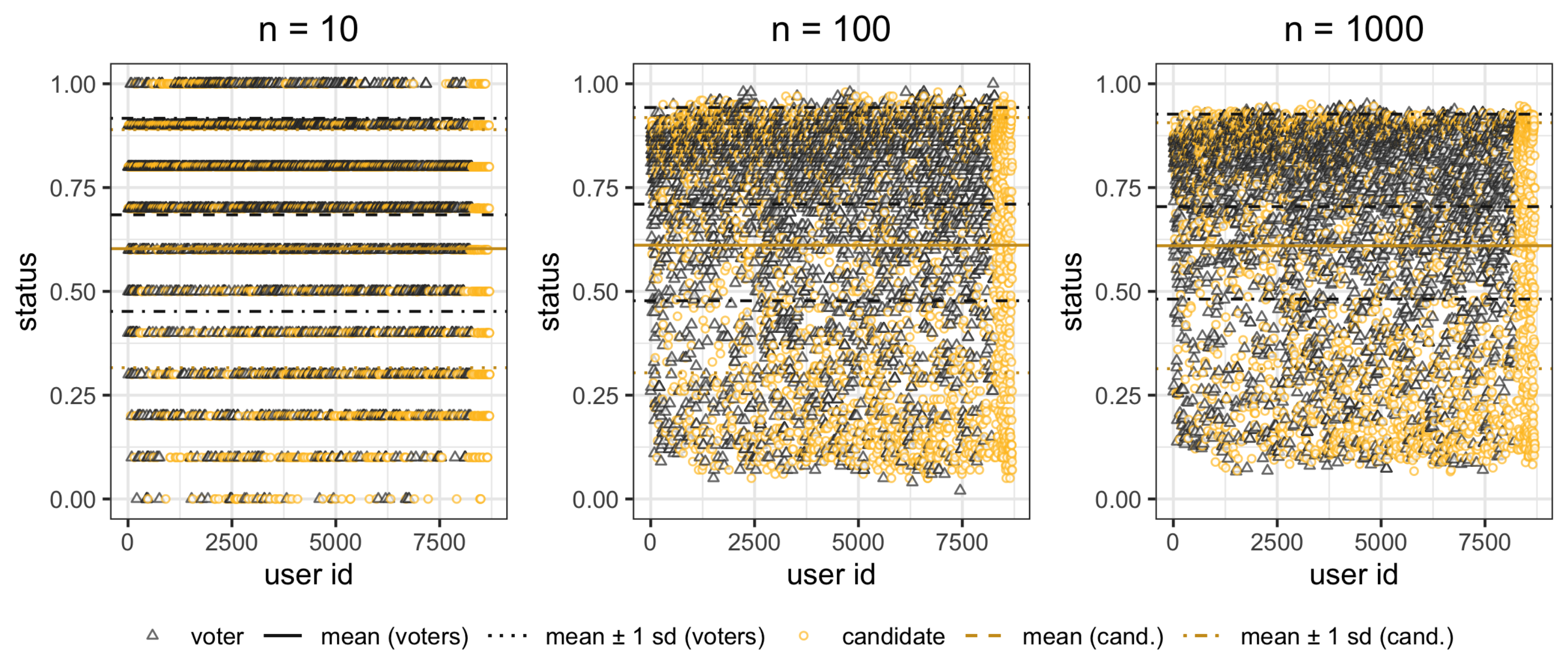}
    \caption{Status score distribution for random (top row) and breadth-first (bottom row) search tree sampling to number of $n$ sampled trees, $n=10$, $n=100$ and $n=1000$ for editors and nominees for status.}
    \label{fig:WikiSamplingStatus}
\end{figure}

In this experiment, we offer a heuristic answer to the complex estimation of the sufficient number of spanning trees ($n$ in Defn.~\ref{d:computation}) that will produce reliable modeling of balanced state representation. We evaluate the sensitivity of status and influence scores to the number of spanning trees sampled for random minimal tree and breadth-first search spanning tree for $n=10$, $n=100$, and $n=1000$ trees using random minimal sampling and breadth-first search techniques. The results for random spanning tree sampling for status is illustrated in Figure~\ref{fig:WikiSamplingStatus} and for influence in Figure~\ref{fig:WikiSamplingInfluence}. For both figures, random samples are on top and the breadth-first samples are on the bottom.

Status for $n=10$ can only take one of 11 different values (as the vertex in majority or not for each of the 10 sampled balanced states), as illustrated in Figure~\ref{fig:WikiSamplingStatus}. A shelving effect is visible due to so few samples. Influence (Defn.~\ref{d:computation}) has a higher resolution, as it is based on the average of agreement for each vertex, and vertex measure differs. A shelving effect is visible in Figure~\ref{fig:WikiSamplingInfluence} for a tight group of editors that behave in a similar fashion. Higher $n$ allows for more diverse samples to contribute to status, and while the overall resolution of both discovery strategies is smaller, it provides better results.  Figures~\ref{fig:WikiSamplingStatus} and \ref{fig:WikiSamplingInfluence} show that the status values have the higher resolution. They also show that $n=100$ of breadth-first and random sampling spanning trees achieves similar separation results in nominee outcome status as $n=1000$ trees. 
\begin{figure}[!ht]
    \centering
    \includegraphics[width=5.5in]{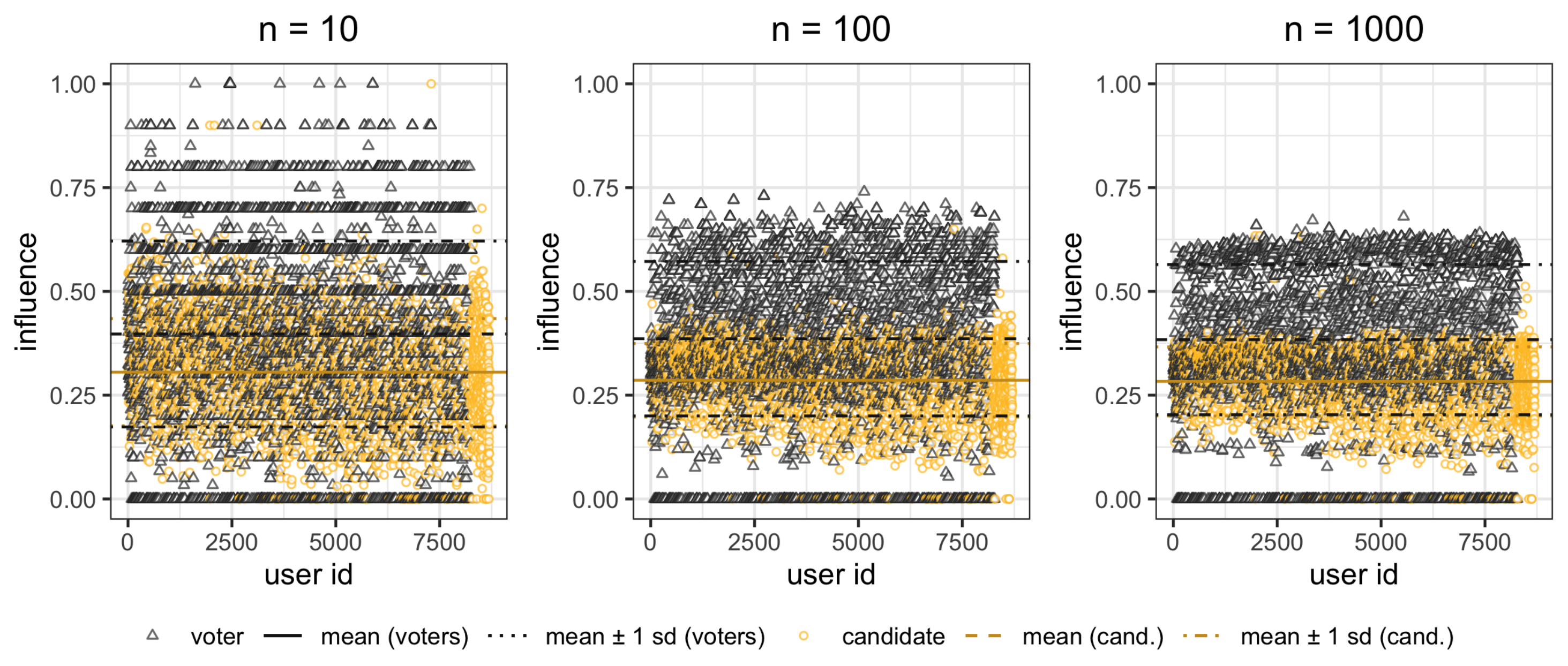} \\
\includegraphics[width=5.5in]{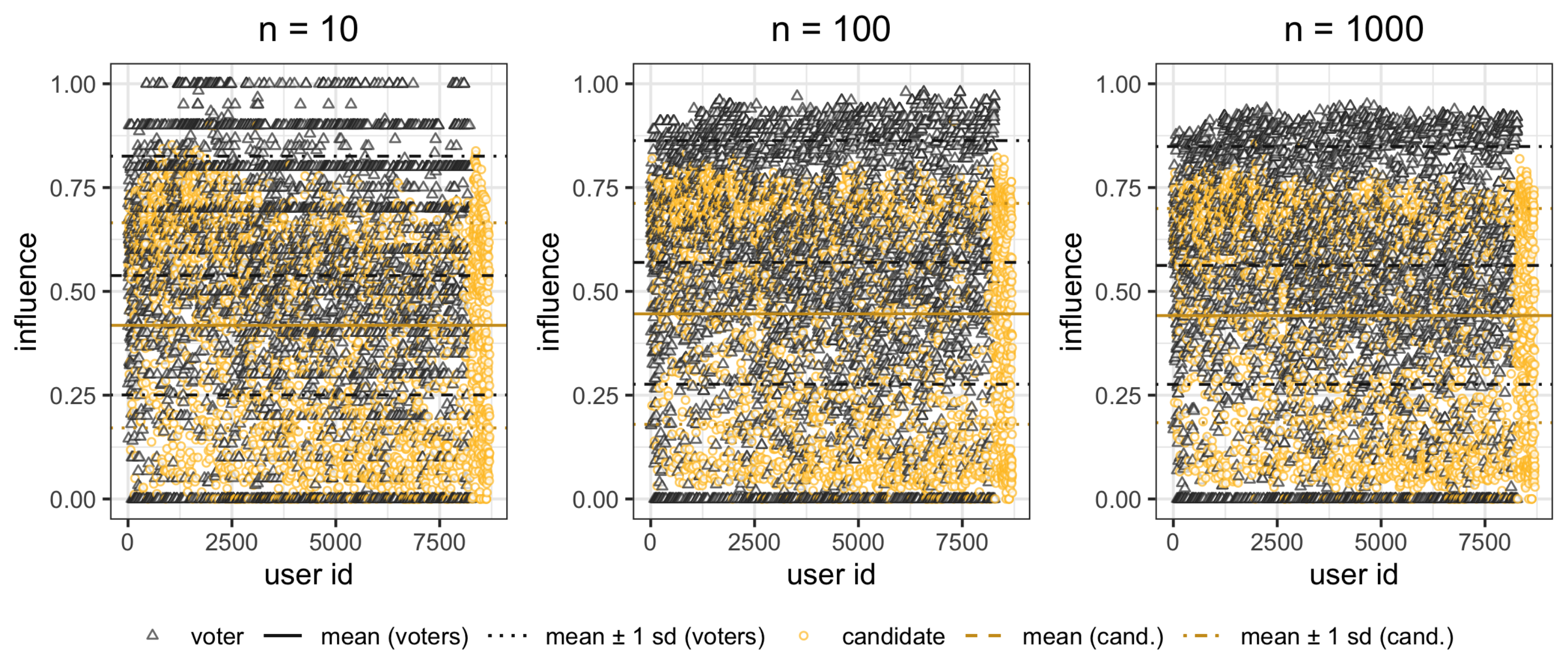}
    \caption{Influence score distribution for random (top row) and breadth-first (bottom row) search tree sampling to number of $N$ sampled trees, $N=10$, $N=100$ and $N=1000$ for editors and nominees for influence.}
    \label{fig:WikiSamplingInfluence}
\end{figure}
 
What is the guiding principle for larger graphs? We have tested the method on a larger signed graph for Slashdot, and $n=1000$ seems to be a good sampling rate for breadth-first spanning tree sampling, as illustrated in Figure~\ref{f:Slashdot}.

\subsubsection{Experiment: Outcome Analysis}
\label{sssec:Exp3}

Measures of status and influence can be used to access the outcome for a vertex in an attitudinal graph.  Requests for adminship (RfA) is the process by which the Wikipedia community decides to promote nominees into administrators \cite{RfA}.  Here, we use \emph{RfA} as the ratio of total votes for the nominee overall votes. Election outcomes are colored by blue triangles representing candidates that won the election and red circles representing candidates that lost the election.  Nominee status is obtained by users submitting their own requests for adminship or being nominated by editors. 

\begin{figure}[!ht]
    \centering
    \includegraphics[width=3.5in]{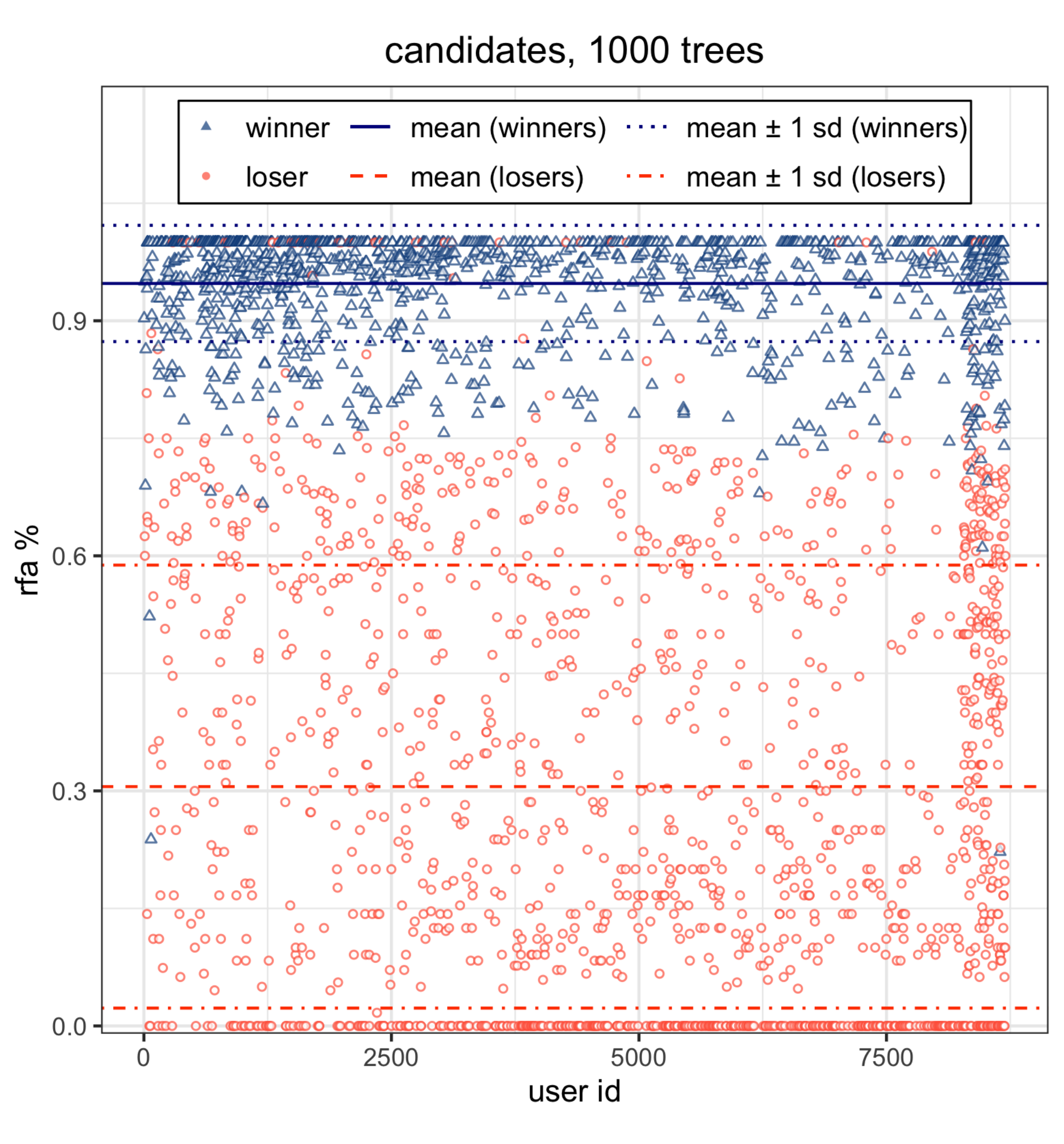}
    \caption{Wikipedia data analysis of the request for adminship (RfA) \cite{RfA}.  Blue triangles are users that won the election, and red circles are users that lost the election.}
    \label{fig:RfAnalysis}
\end{figure}

The final outcome for Wikipedia is a complex process that involves majority voting (RfA in Figure~\ref{fig:RfAnalysis} and a vetting process. In general, if the number of positive votes is under 65\%, the nominee is rejected; if the number of votes is over 75\%, the nominee is selected. A vetting process and discussion determine the final outcome.  Burke proposed a model of the behavior of candidates for promotion to administrator status in Wikipedia \cite{Burke2008}. He analyzes multiple measurable features of the nominee (strong edit history, varied experience, user interaction, helping with scores) and highlights similarities and differences in the community’s stated criteria for promotion decisions to those criteria that actually correlated with promotional success.  
In this experiment, we examine the use of status and influence scores per vertex as vertex features. Status and influence do not consider any of Burke's candidates' features, only their position in the signed graph. The relationship of status and influence to RfA is illustrated in Figure~\ref{fig:RfAnalysis}.

\begin{figure}[!ht]
    \centering
    \includegraphics[width=3in]{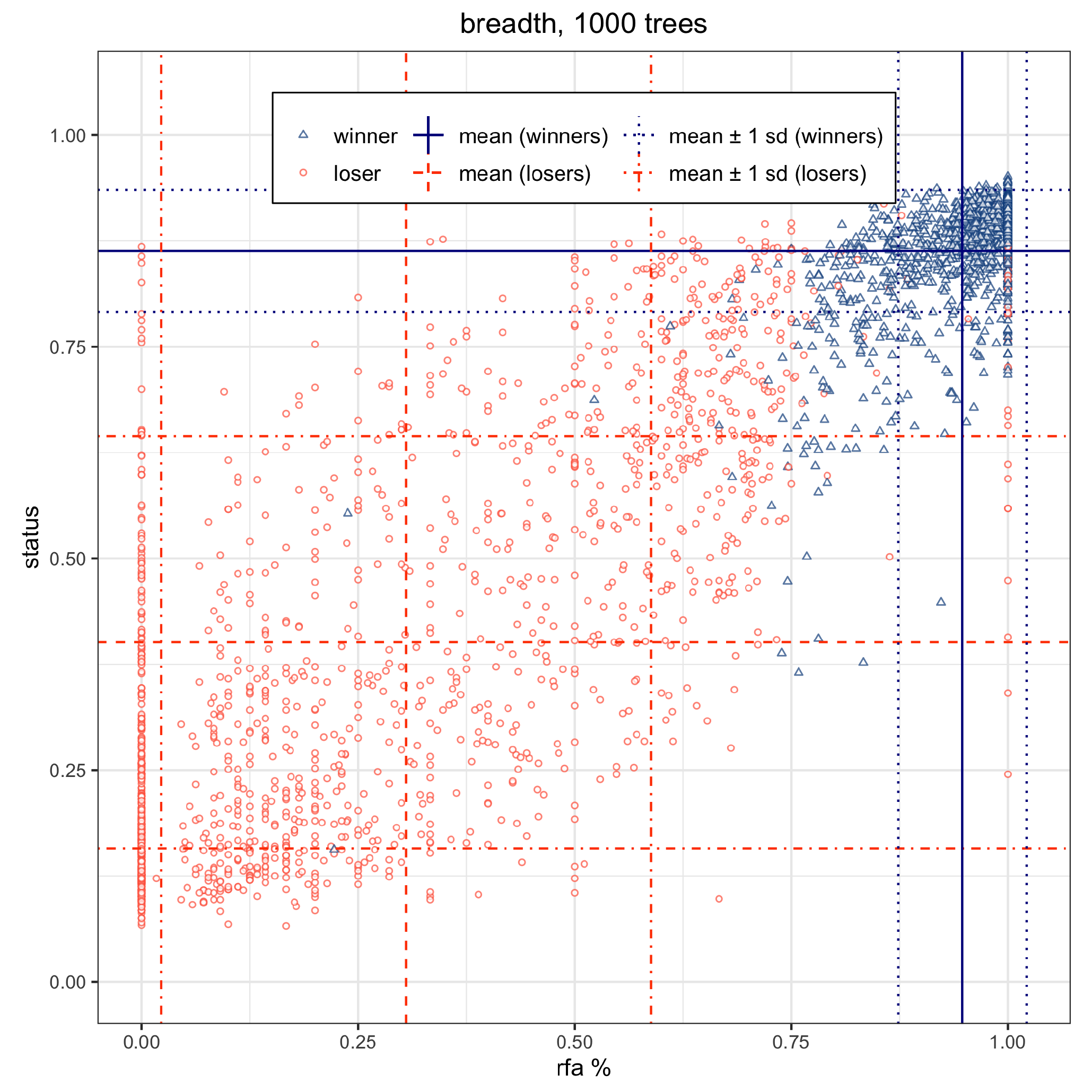}
    \includegraphics[width=3in]{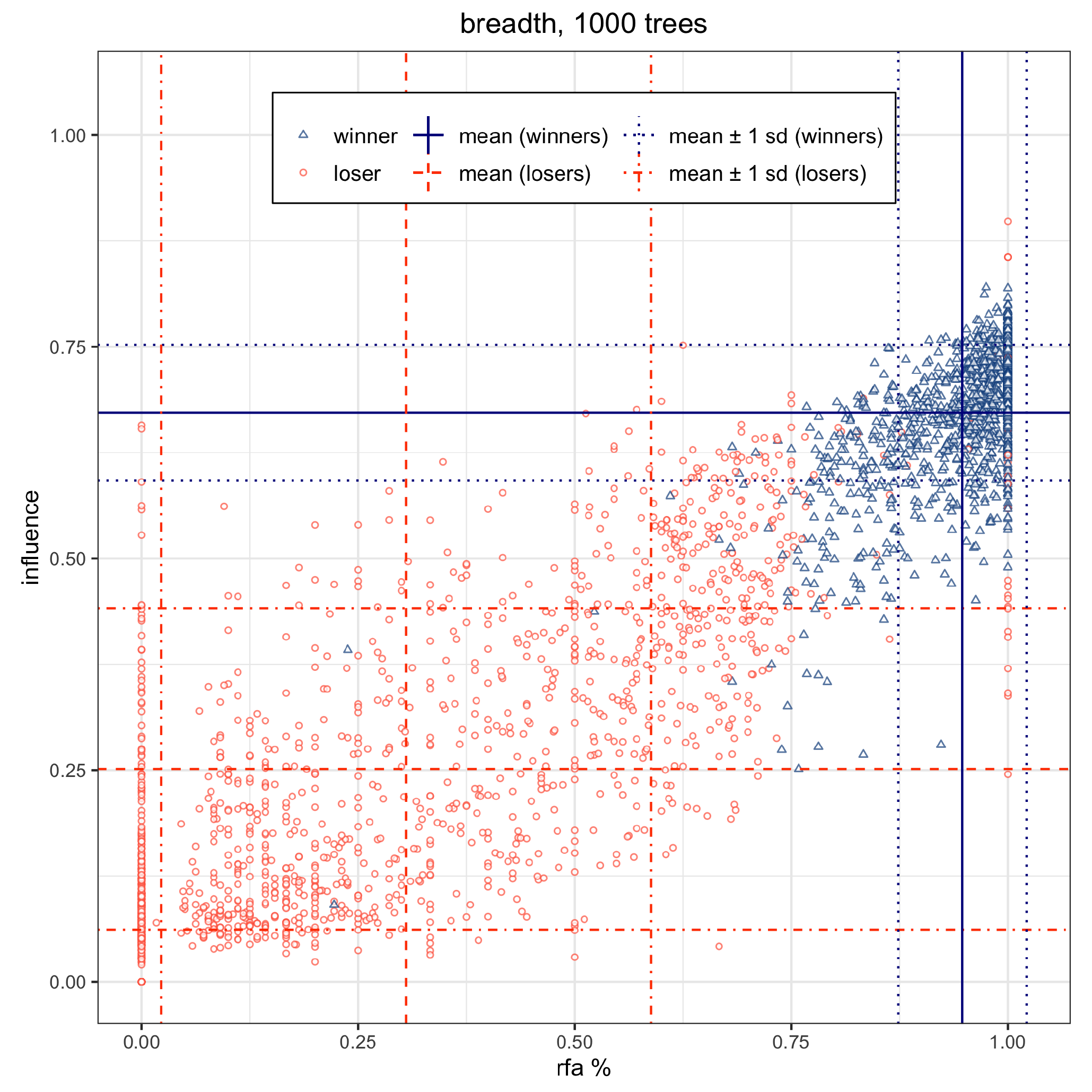}
    \caption{Wikipedia data analysis of (left) status vs RfA and (right) influence vs RfA}
    \label{fig:RfAnalysisStatus}
\end{figure}

Figure~\ref{fig:RfAnalysis} shows that both status and influence are highly positively correlated with RfA, and all measures are highly correlated to the outcome as well. A summary of aggregate findings is in Table \ref{tab:2}. Next, we study the nominees whose status and influence scores are different.  

\begin{table}[!ht]
\centering
\caption{Measure distribution in Wiki adminship dataset for $N=1000$ breadth-first spanning tree discovery:}
\label{t:Election}       % Give a unique label
\begin{tabular}{l|lll}
\hline\noalign{\smallskip}
Mean (St.Dev) & \textcolor{olive}{Nominees} & \textcolor{blue}{Promoted} & \textcolor{red}{Not Promoted} \\
\noalign{\smallskip}\hline\noalign{\smallskip}
RfA &  0.9476 (0.0742) & 0.3055 (0.2826)\\
Status &  0.6097 (0.2962) & 0.8632 (0.0719) & 0.4009 (0.2433)\\
Influence  & 0.4414 (0.2579) & 0.6721 (0.0803) & 0.2514 (0.1898)\\
\noalign{\smallskip}\hline\noalign{\smallskip}
Mean (St.Dev) & All & \textcolor{gray}{Editors} & \\
\noalign{\smallskip}\hline\noalign{\smallskip}
RfA & 0.5958 (0.3852) & N/A & \\
Status & 0.6693 (0.2564)  & 0.7041 (0.2226) & \\
Influence & 0.5178 (0.2823) & 0.5624 (0.2864) & \\
\end{tabular}
\label{tab:2}
\end{table}

%(3) status-versus-influence comparisons are well founded for community detection while retaining the ability to detect outliers across status, influence, promotion, and voting.

\subsubsection{Experiment: Status/Influence Cone}
\label{sssec:ExpCone}

We now analyze outcomes and RfA in a status-influence feature space. The vertices only appear under $x=y$ line in the graph, as shown in Figure~\ref{fig:WikiStatusInfluence} and Lemma \ref{l:StatusCone}. Editors, the black triangles in Figure~\ref{fig:WikiStatusInfluence} (left) have the higher status and influence as leaders in swaying opinions. In Figure~\ref{fig:WikiStatusInfluence} (right), we restrict the analysis to nominees and color by outcome only: blue triangles are elected and red circles are rejected. Significant separation between the elected and rejected nominees is apparent. Note they are further away from the $45^{\circ}$ line representing the most influential users, which are almost exclusively editors.

\begin{figure}[!ht]
    \centering
    \includegraphics[width=3in]{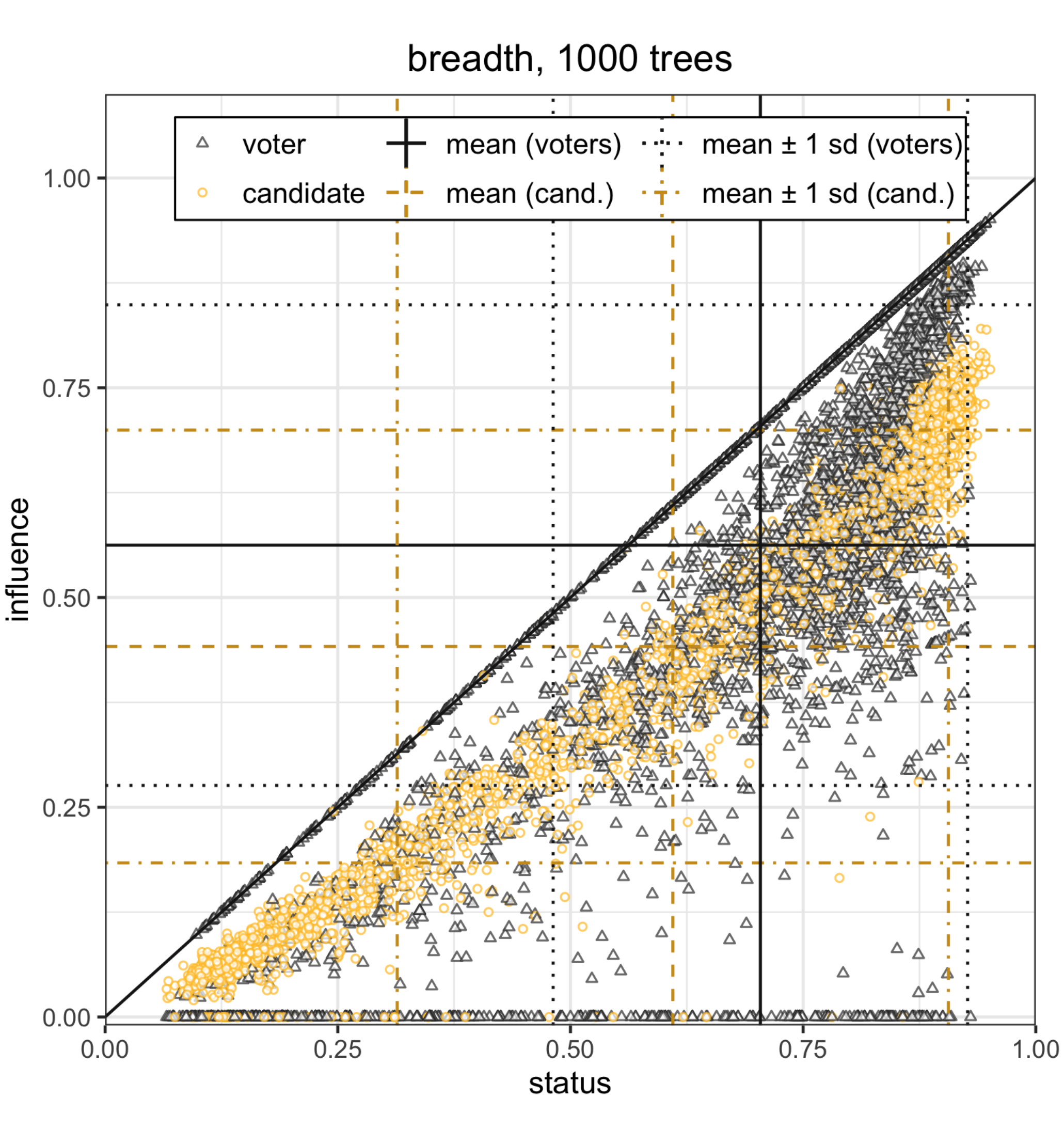}
     \includegraphics[width=3in]{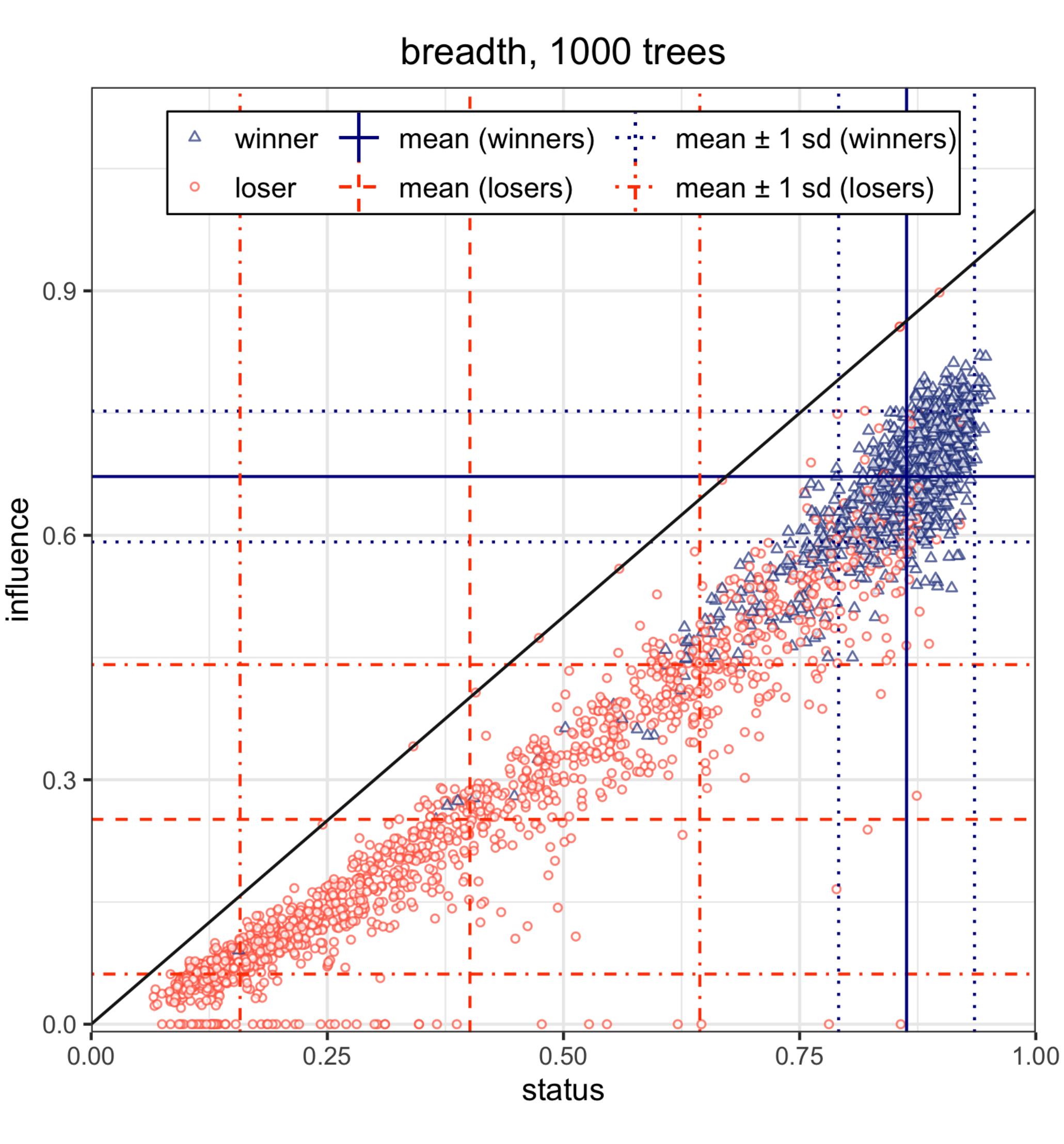}
    \caption{Wikipedia data analysis of the voting process from status vs. influence perspective: (left) editors (black triangles) vs. nominees (yellow circles); (right) by Wikipedia nominee outcome: blue triangles are elected, red circles are rejected.}
    \label{fig:WikiStatusInfluence}
\end{figure}
   
Next, we analyze RfA values in the status-influence space in Figure~\ref{fig:WikiStatusInfluence2}. The left graph shows continuous scores, while the right one uses the Wikipedia admin score scale. Here, we flag spam users, privileged users, narrow domain users and all anomalies by examining red circle distribution, the RfA in [65,75)\% over status-influence graph in Figure~\ref{fig:WikiStatusInfluence2} (right).

\begin{figure}[!ht]
    \centering
    \includegraphics[width=3in]{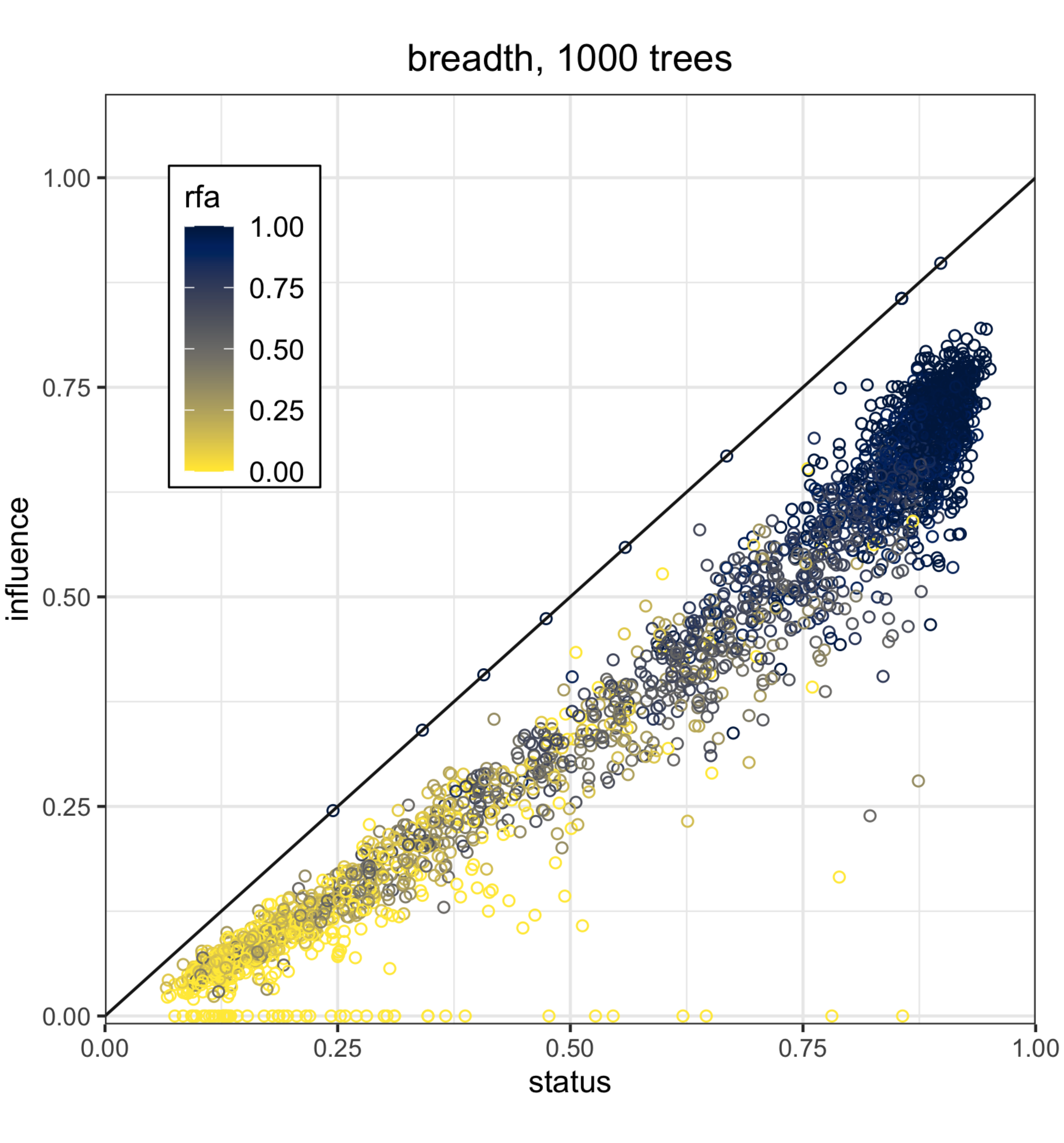}
    \includegraphics[width=3in]{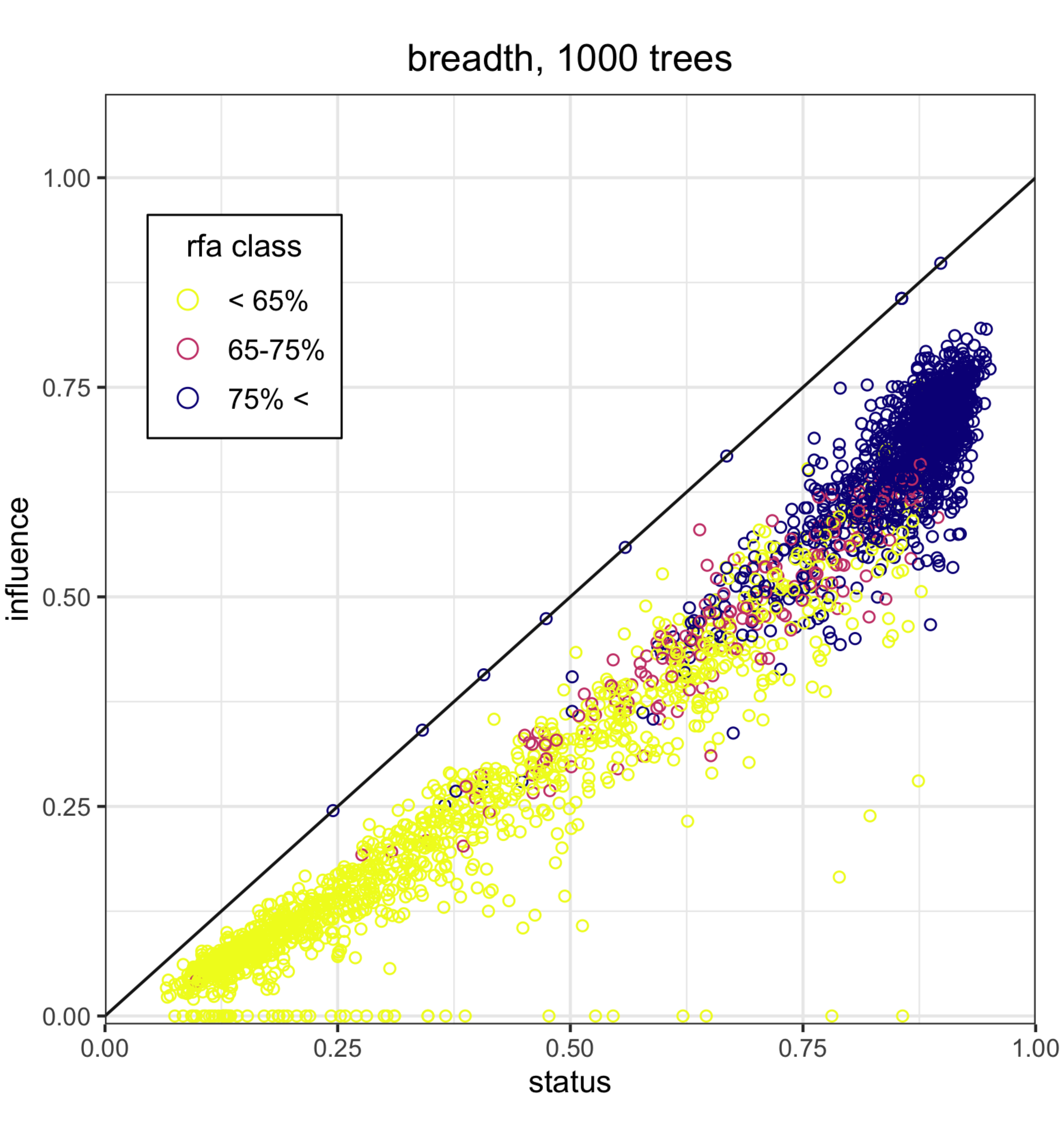}
    \caption{Wikipedia data analysis of the voting process from status vs. influence perspective. Left: RfA score for nominees ramped from yellow to blue (light-to-dark). Right: Binned by $<.65\%$ in yellow, $[.65,.75)$ red, and $[.75,1]$ blue. }
    \label{fig:WikiStatusInfluence2}
\end{figure}

Our algorithm for Wikipedia adminship uncovered several interesting cases. Wiki ID \emph{80-man} had a status of $0.919$, an influence of $0.622$, and lost the elections. The user primarily edited a lot of music pages, and was deemed ``not ready'' for adminship yet based on other criteria. Wiki ID \emph{bozmo} had a status of $0.405$, an influence of $0.277362637$, and won the elections. He self-nominated after around 3 years of contributions, and he received a fair amount of opposition due to being less active around the time of nomination. It is unclear why he won. Wiki ID \emph{tjstrf} had a status of $0.905$, an influence of $0.649$, and lost. Due to some discussion of sensitive topics, he rejected the promotion. Wiki ID \emph{dmn} had a status of $0.448$, an influence of $0.279753623$, and won the election. Further research showed they have been on Wikipedia for over 16 years, made regular edit contributions for a year, and have a history of conflicts and controversial comments. All four cases show that our algorithm uncovered atypical promotion or lack thereof, even if RfA was within guiding limits.  The Wikipedia administrator election outcome analysis using the graphB approach allows for a fast, objective snapshot of the outcome, as it allows users to flag nominees whose outcome is not in balance with the rest of the attitudinal network. If the outcome is known, as it is in Wikipedia adminship, graphB is used to flag unexpected outcomes for editor review. 

\subsection{Slashdot Zoo}
\label{ssec:Slashdot}

\begin{figure}[!ht]
    \centering
     \includegraphics[width=3in]{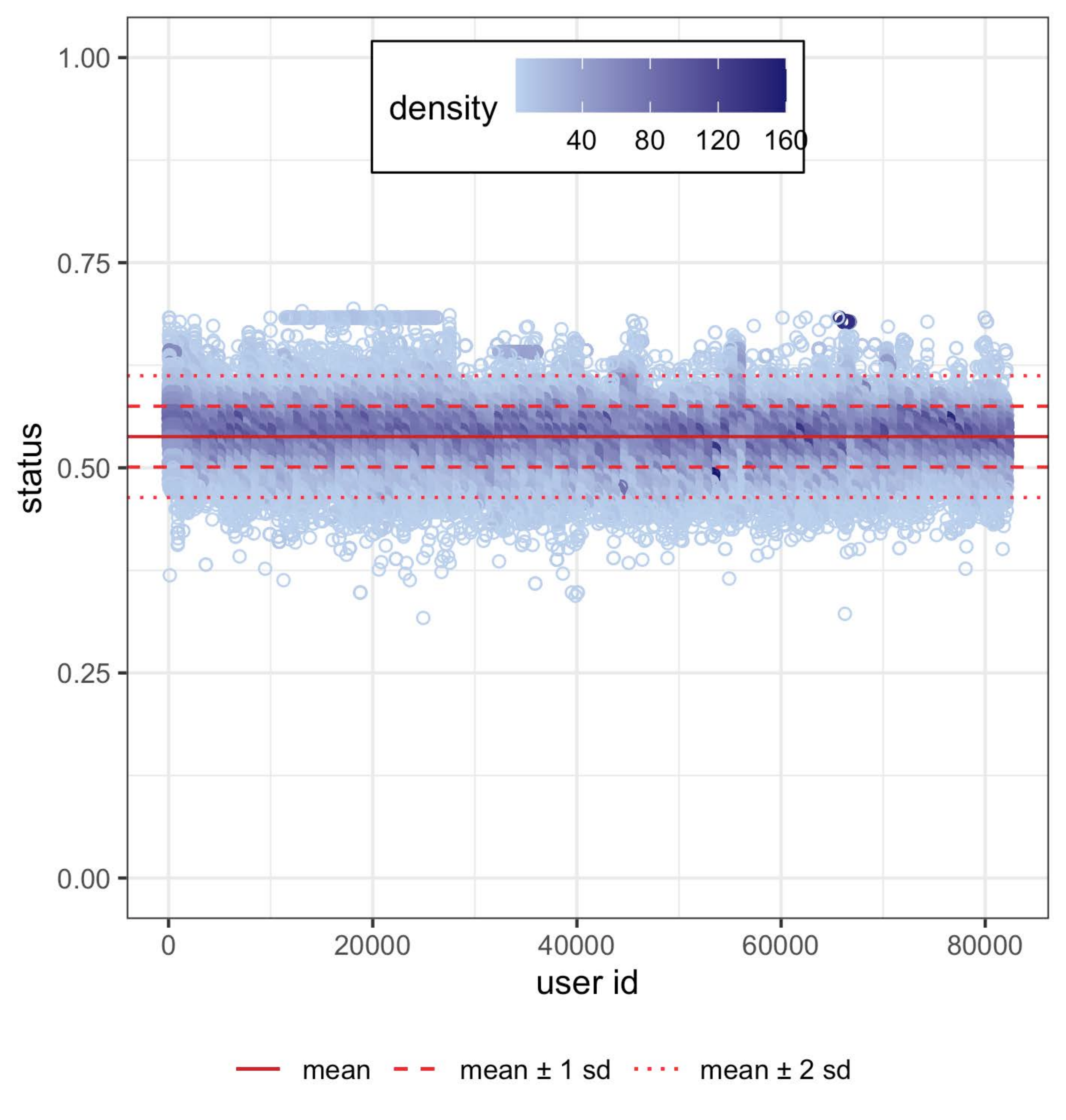}
     \includegraphics[width=3in]{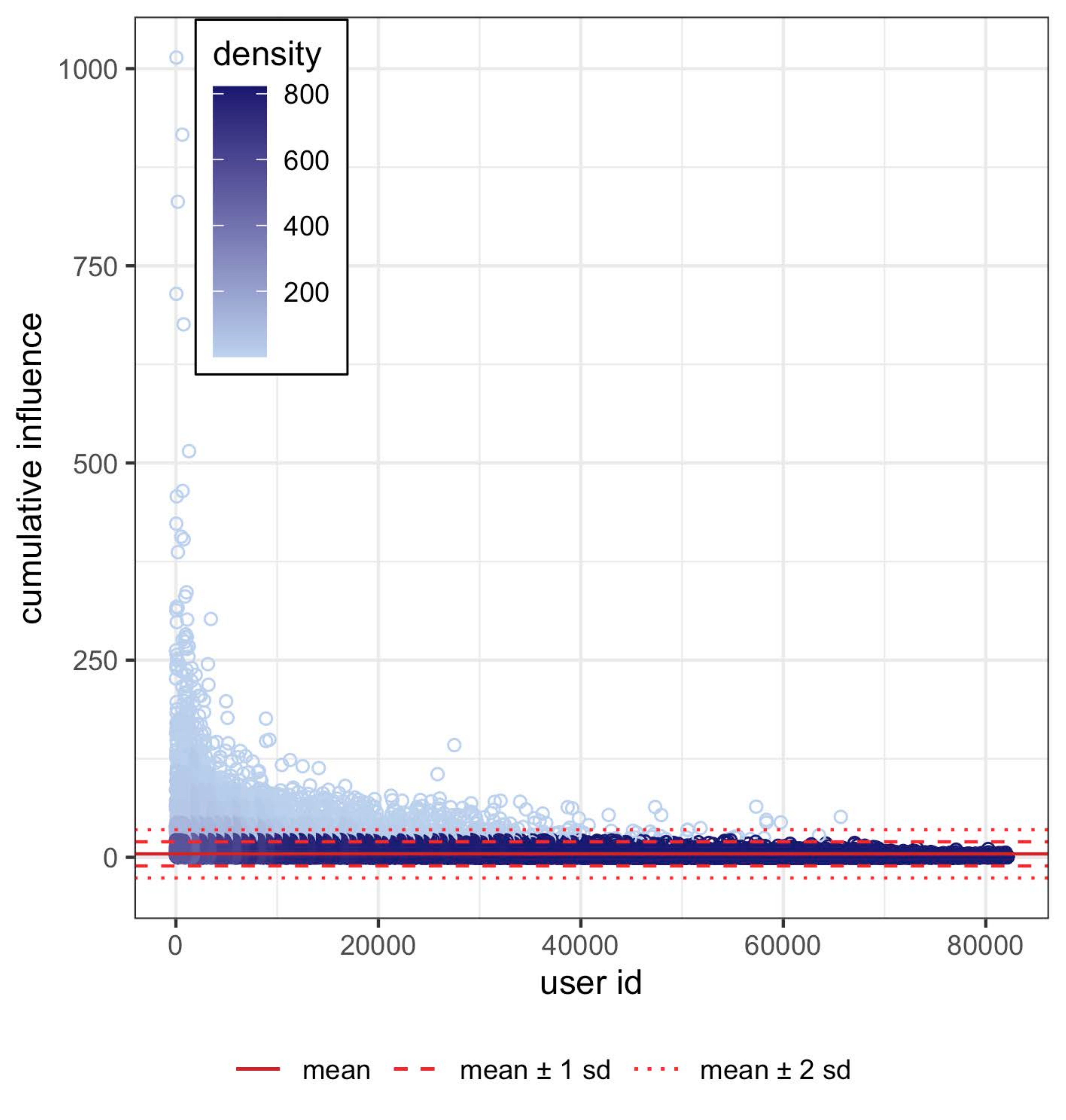}
    \caption{Slashdot friend-foe network analysis using frustration cloud approach and $N=1000$ breadth-first spanning trees: status density and influence density.}
    \label{f:Slashdot}
\end{figure}

Slashdot Zoo is a signed social network with $82,144$ users (vertices) and 549,202 edges; $77.4\%$ edges are positive \cite{snapnets}.  Edge direction and weight annotates that the origin user tagged the target user as a friend (weight +1) or foe (target -1) \cite{snapnets}.  The largest connected component of this data contains $82,052$ ($99.89\%$) users. The type of analysis presented in Section~\ref{ssec:Wiki} can be expanded to any attitudinal dataset in light of status and influence, with or without the outcome. There is no outcome for Slashdot Zoo data.

We construct the attitudinal graph from the friend-foe relationships and analyze the status and influence of users.  Results are presented in Figures~\ref{f:Slashdot} and \ref{f:Slashdot2}, with the frustration cloud and $n=1000$ breadth-first balance tree discovery. In the Slashdot analysis, we consider vertex degree and remove normalization from Defn.~\ref{d:influence} to analyze cumulative influence. Figure~\ref{f:Slashdot} shows the density distribution of the status and cumulative influence over the entire network. The second image in Figure~\ref{f:Slashdot} clearly shows higher influence for early adopters of the network.

Figure~\ref{f:Slashdot2} illustrates Lemma \ref{l:StatusCone} for the influence and status relation. The vertices on the slope status = influence line are a single pendant vertex whose edge is positive (RfA is one, degree is 1); the slope influence = $0$ line is a single pendant vertex whose edge is negative (RfA is 0, degree is 1); the influence is always smaller than the status value by Defn.~\ref{d:status} and Defn.~\ref{d:influence}. The angular outliers corresponding to single-decision outcomes (positive slope $1$, and negative slope $0$) and the radial outliers are the most/least influential nodes in the network. This influence-status cone analysis allows us to analyze the measurements as a function of node degree. See Figure~\ref{f:Slashdot2} colored by node degree. The users with high node degree, overwhelmingly positive votes, mid-range influence, and high status are excellent moderator candidates in this set. 

\begin{figure}[!ht]
    \centering
     \includegraphics[width=3in]{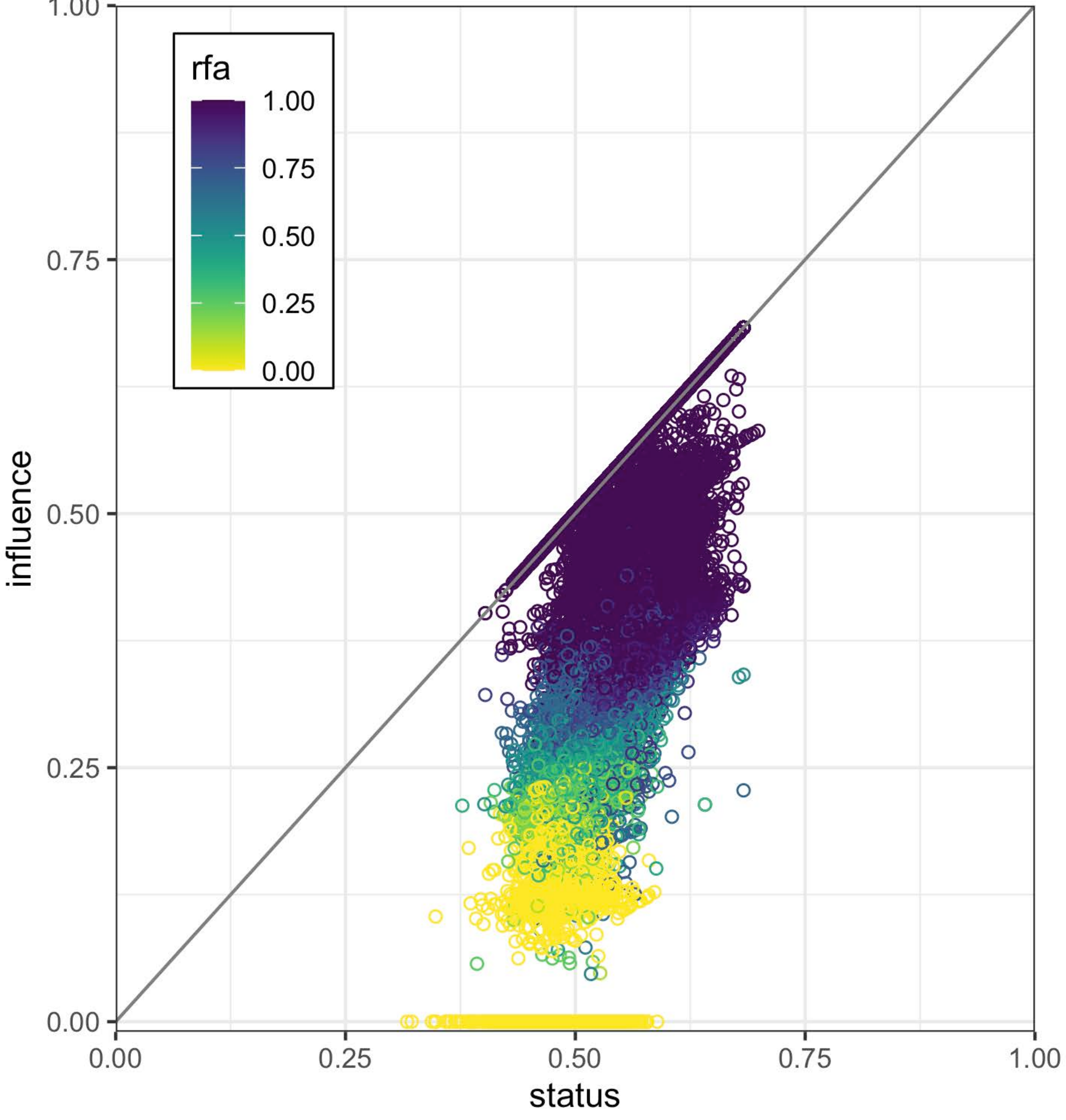}
     \includegraphics[width=3in]{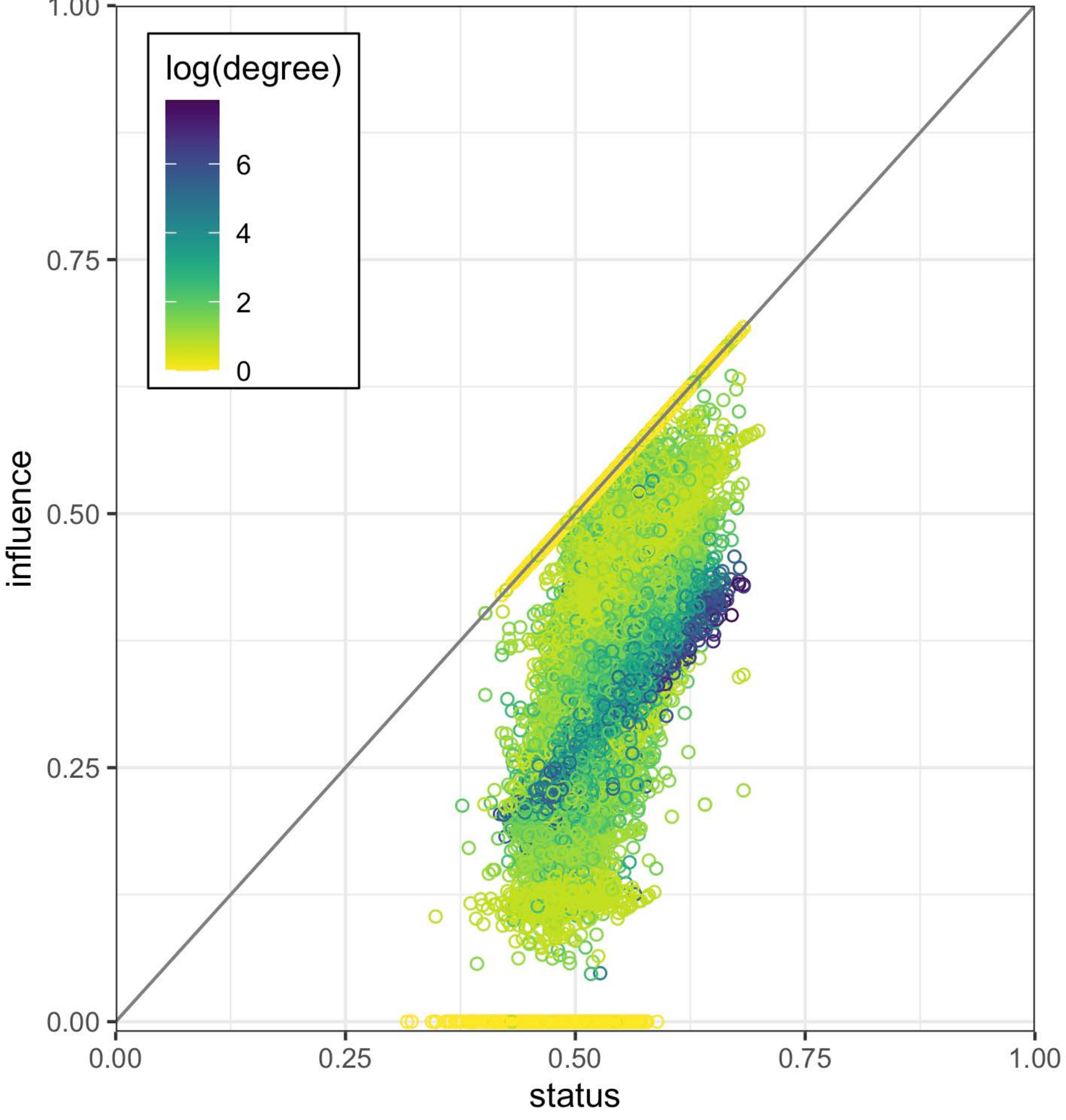}
    \caption{Slashdot friend-foe network analysis using the frustration cloud approach and $N=1000$ breadth-first spanning trees: influence vs. status by RfA, and influence vs. status by log vertex degree.}
    \label{f:Slashdot2}
\end{figure}

\subsubsection{Scaling graphB Implementation to Large Signed Graphs}
\label{ssec:Slashdot2}

The overall complexity of the implemented graphB algorithm ($https://github.com/DataLab12/graphB$) is $O(n \cdot v \cdot e)$ for run time and $O(v^2)$ for memory consumption, where $n$ is number of spanning trees, $v$ is number of vertices, and $e$ is number of edges. The Slashdot dataset \cite{snapnets} is the largest dataset we have processed to date with over 82,000 vertices. The memory requirement to keep and process $O(v^2)$ matrices required us to upgrade to high memory nodes on HPC \cite{LEAP} to run the code for Slashdot data.  

Scaling bottleneck in our graphB implementation \cite{graphB} was the tree discovering and tree balancing step with $O(n \cdot v \cdot e)$ complexity. The number of discovered cycles is linear with the number of edges and vertices, and the time to balance a graph per spanning tree became prohibitively high. We have implemented Apache Spark parallelization for finding spanning trees and fundamental cycles (as the process is independent for each spanning tree) to utilize the computing cluster and overcome computing issues. The released graphB code \cite{graphB} allows the user to measure the timing of each step, and with Apache Spark parallelization we have achieved speedup of 22.3 times, as shown in Figure~\ref{fig:SparkSpeedup}.  

\begin{figure}[!ht]
    \centering
     \includegraphics[width=4in]{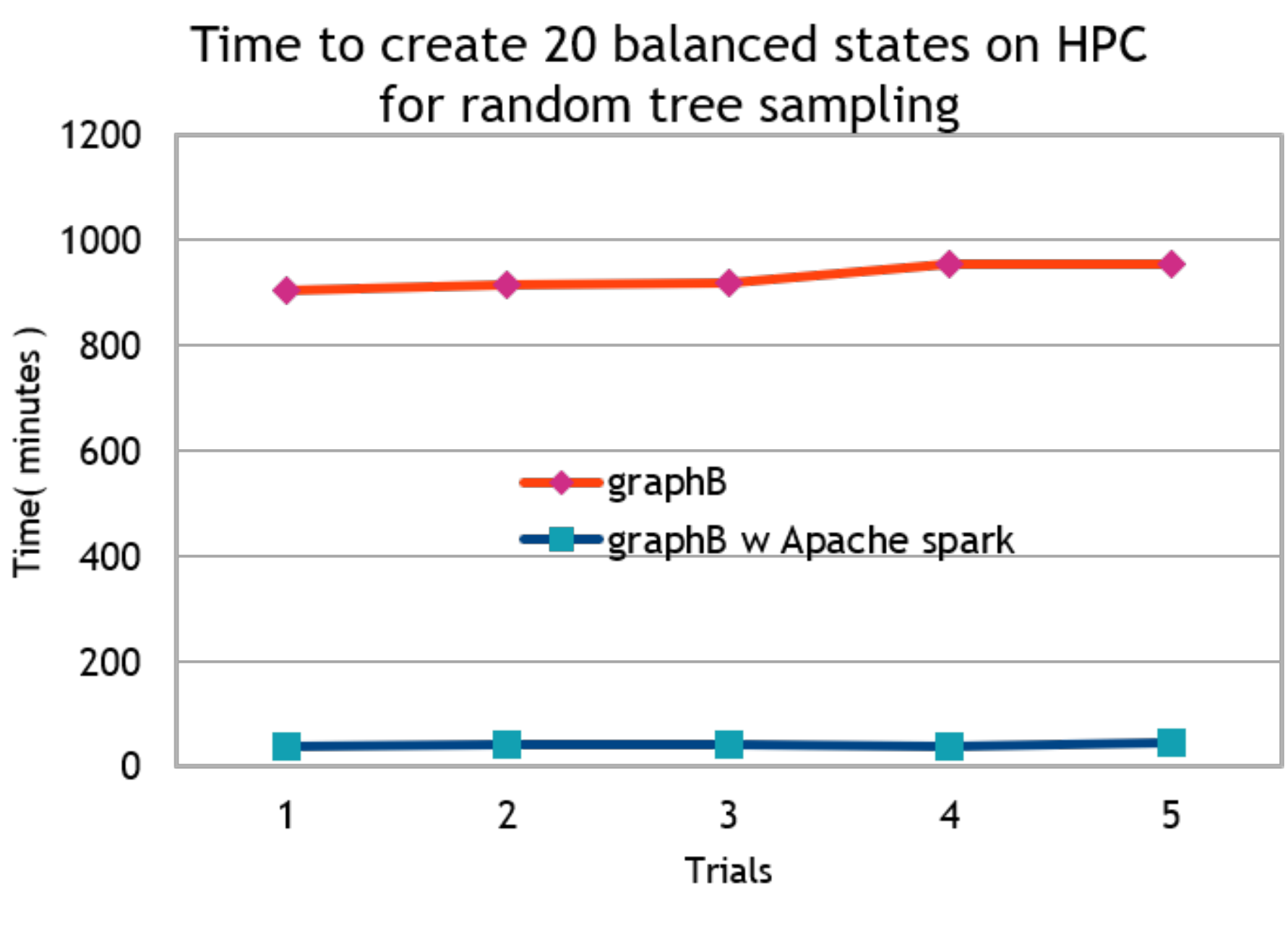}
    \caption{Timing of balancing graphs using 20 spanning trees and Slashdot data on LEAP cluster.}
    \label{fig:SparkSpeedup}
\end{figure}

\subsection{Highland Tribes}
\label{ssec:Highland}

The frustration cloud-based approach allows for a more robust way to analyze the perceived outcome in an attitudinal network graph, as it is based on the mathematical sociology model for a balanced system. The Highland Tribes datasets captures the alliance structure of a network of tribes in the Eastern Central Highlands of New Guinea \cite{1954Read}. The network contains sixteen tribes (vertices), and the edges represent agreement (``rova'') or animosity (``hina'') between two tribes, as illustrated in Figure~\ref{fig:HTiebreak} with solid lines for agreement and dashed lines for animosity.  Read's ethnography portrayed an alliance structure among three tribal groups containing balance as a special case, as the enemy of an enemy can be either a friend or an enemy \cite{1983Hage}. There are 16 vertices (tribes) and 58 signed edges (tribe relations): 29 positive (sign +1) and 29 negative (sign -1). The signed graph $\Sigma$ for the Highland tribes dataset is constructed by adding the two provided matrices.  

\subsubsection{Experiment: Vertical Status}
\label{sssec:Exp4}

\begin{figure}[!ht]
    \centering
  \includegraphics[width=4in]{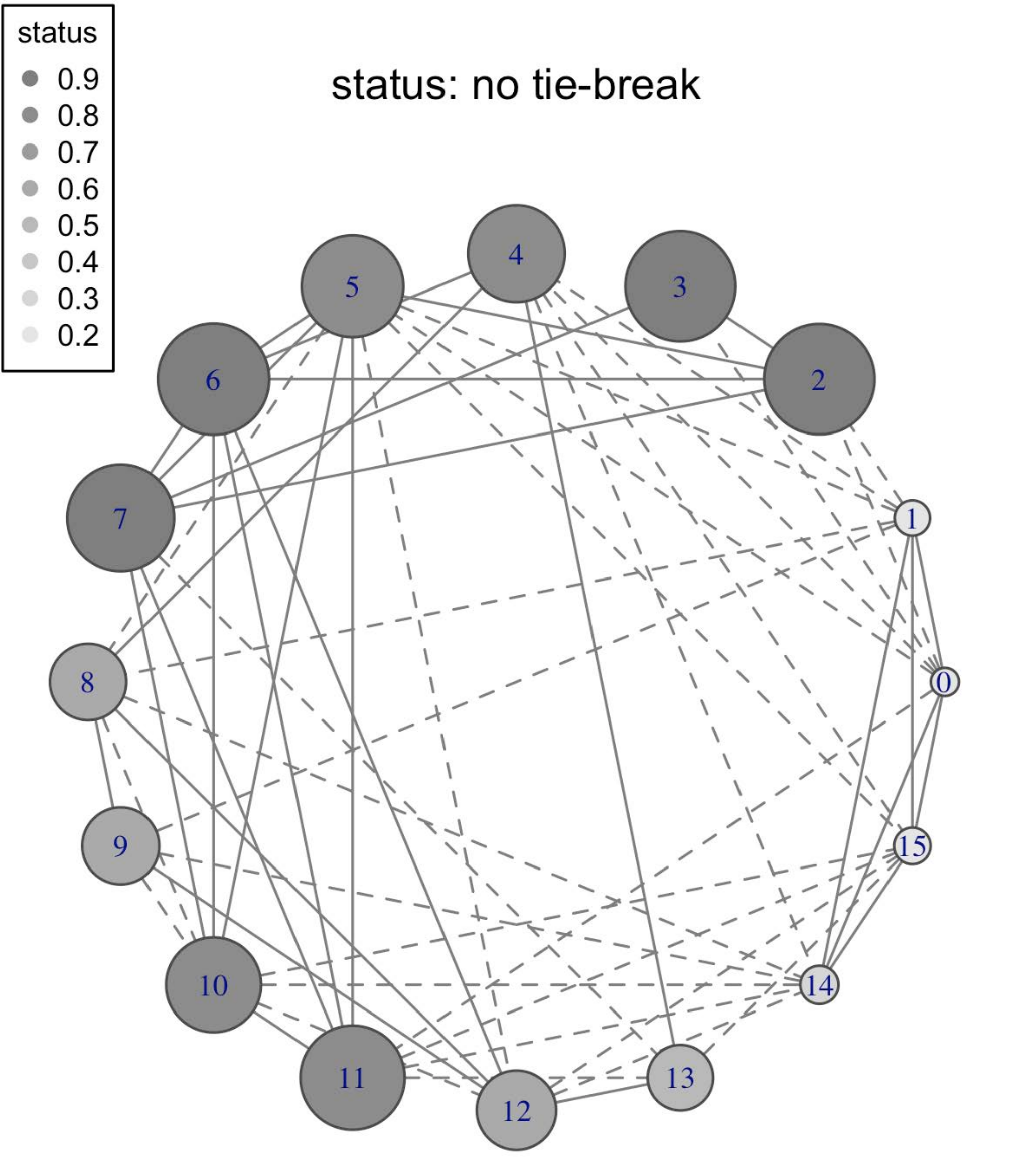}
    \caption{Highland Tribe Status computation for breadth-first sampled $1000$ trees: solid circles are tribes, solid lines are agreeable relations, dashed lines are antagonistic relations between two tribes, the size of the vertex circle illustrates computed status for that circle.}
    \label{fig:HTiebreak}
\end{figure}

The Highland Tribes graph has $402,506,278,163$ spanning trees, and we sample $1000$ spanning trees using the breadth-first approach for this experiment. Highland Tribes relations in Figure~\ref{fig:HTiebreak} separate two groups of tribes. The gray shade and size of the vertex circle correspond to the computed status per Defn.~\ref{d:status} and Defn.~\ref{d:verticalStatus} vertices $0$, $1$, $14$, and $15$ form a smaller group, and it is reflected in the lowest status scores for those $4$ vertices and lower overall status. This \emph{agrees} with the spectral clustering analysis in Section \ref{sssec:Exp5}. We also examine the Conservation Law of Controversy from Section~\ref{ssec:controversy} by examining tie-break rule changes and calculating the status of the vertices: one for vertex $0$, which has minimum status, and one for vertex $6$, which has maximum status. These new status values represent the hypothetical maximum that vertex $0$ or vertex $6$ may achieve. The corresponding temperature graph shows how the vertical status maximizes status for the selected vertex and connected vertices in Figure~\ref{fig:HTiebreak1}.  Figure~\ref{fig:HTiebreak1} (left) illustrates the change in status when the tie-break node is from the smaller cluster. While the status of all 4 nodes in that community grows significantly at the expense of the reduced status of vertices in the majority group, Figure~\ref{fig:HTiebreak} (right) illustrates the maximization of the status in the majority group if vertex 6 is selected as a tie breaker.

\begin{figure}[!ht]
    \centering
  \includegraphics[width=3in]{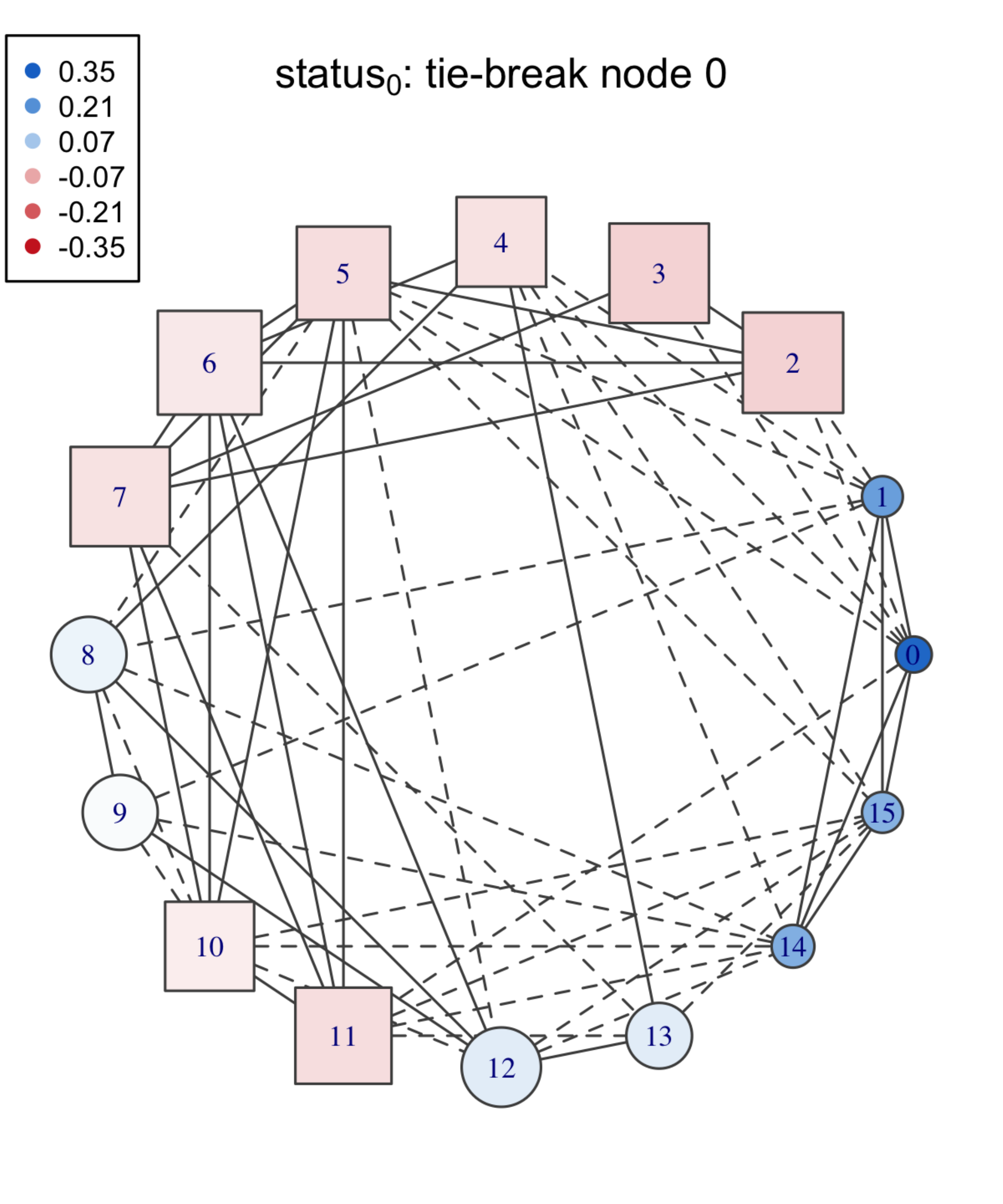}
  \includegraphics[width=3in]{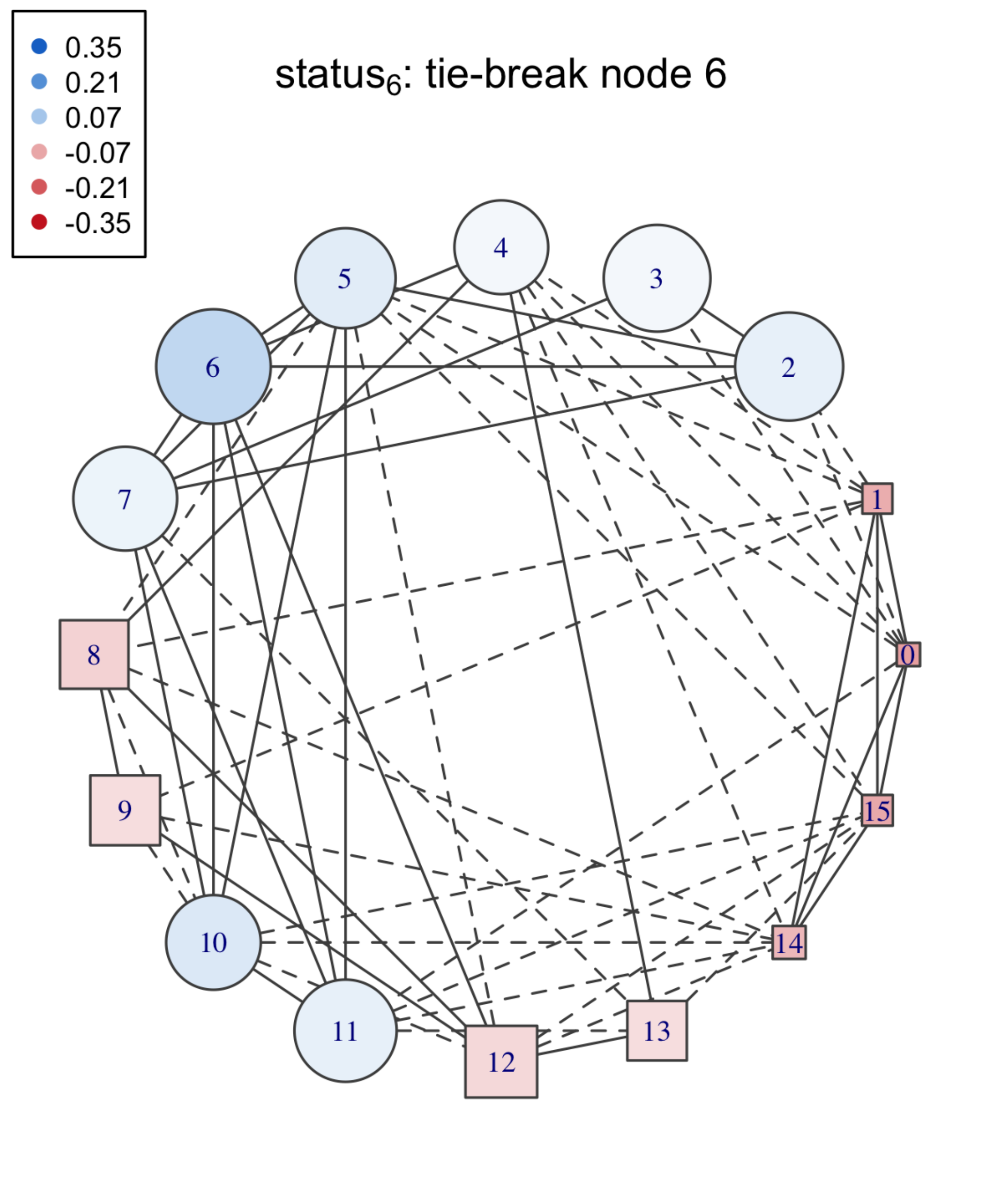}
    \caption{Highland Tribe Verticial Status computation for breadth-first sampled $1000$ trees. Left: Vertex $0$ breaks ties; Right: Vertex $6$ breaks ties. Blue circles represent an increased status and red squares represent a decreased status relative to the original status.}
    \label{fig:HTiebreak1}
\end{figure}

Figure~\ref{fig:Highland} illustrates the status difference per vertex ID for each of $status_6$, $status_0$, and $status$.  The solid black line demonstrates the Conservation Law, as \emph{controversy} for the Highland graph is constant at $10/16$. Also, the status/influence correlation for the Highland dataset has an $R^2 = 0.81$. This may indicate an isolated system, as demonstrated in Figures \ref{fig:WikiAllTree} and \ref{fig:WikiNoTree} on the Wikipedia dataset --- the correlation between status and influence is high when restricted to just the nominees. This means there is no tribe outside of the system acting in a supervisory role.

\begin{figure}[!ht]
  \centering
  \includegraphics[width=3in]{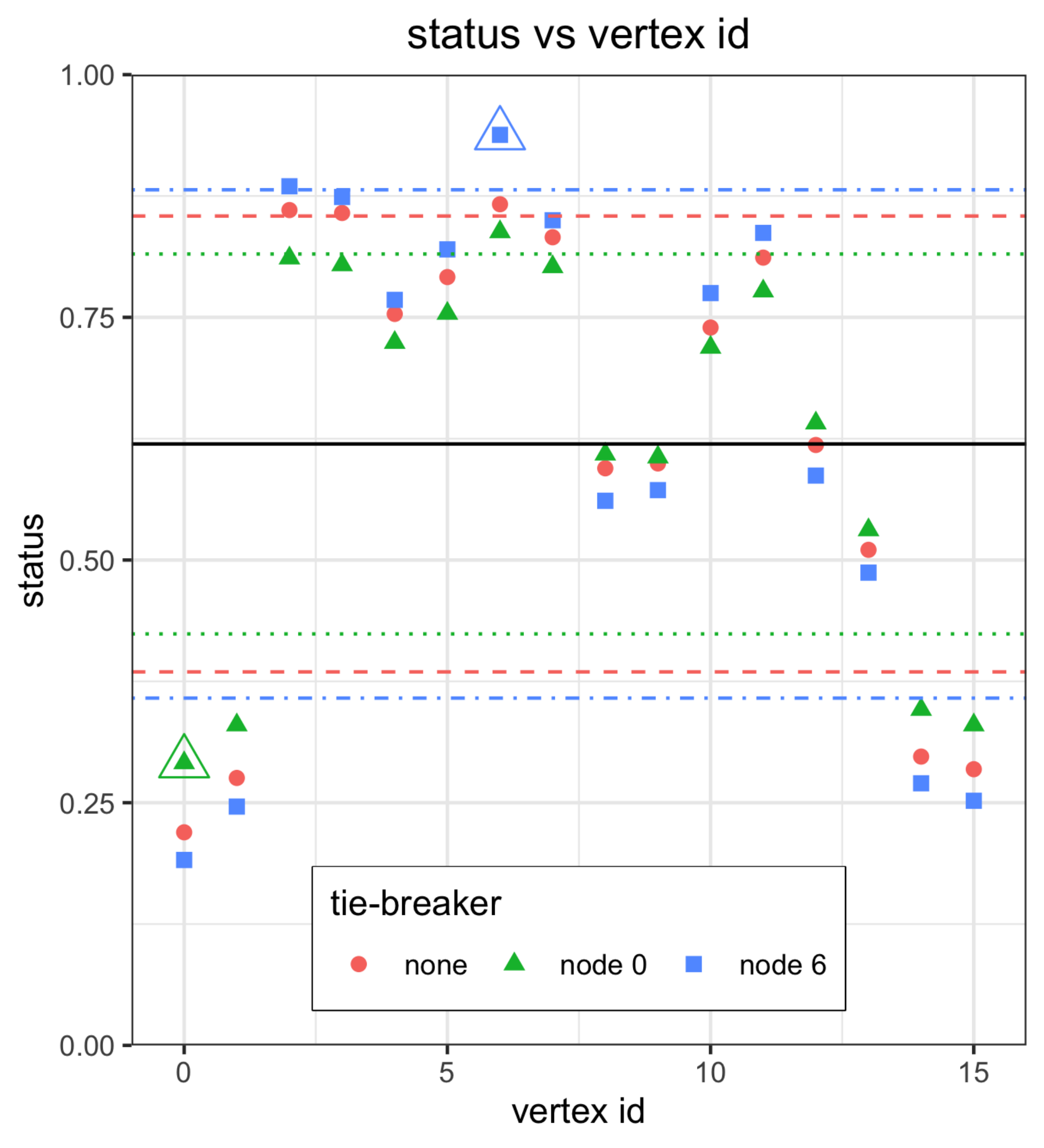}
  \includegraphics[width=3in]{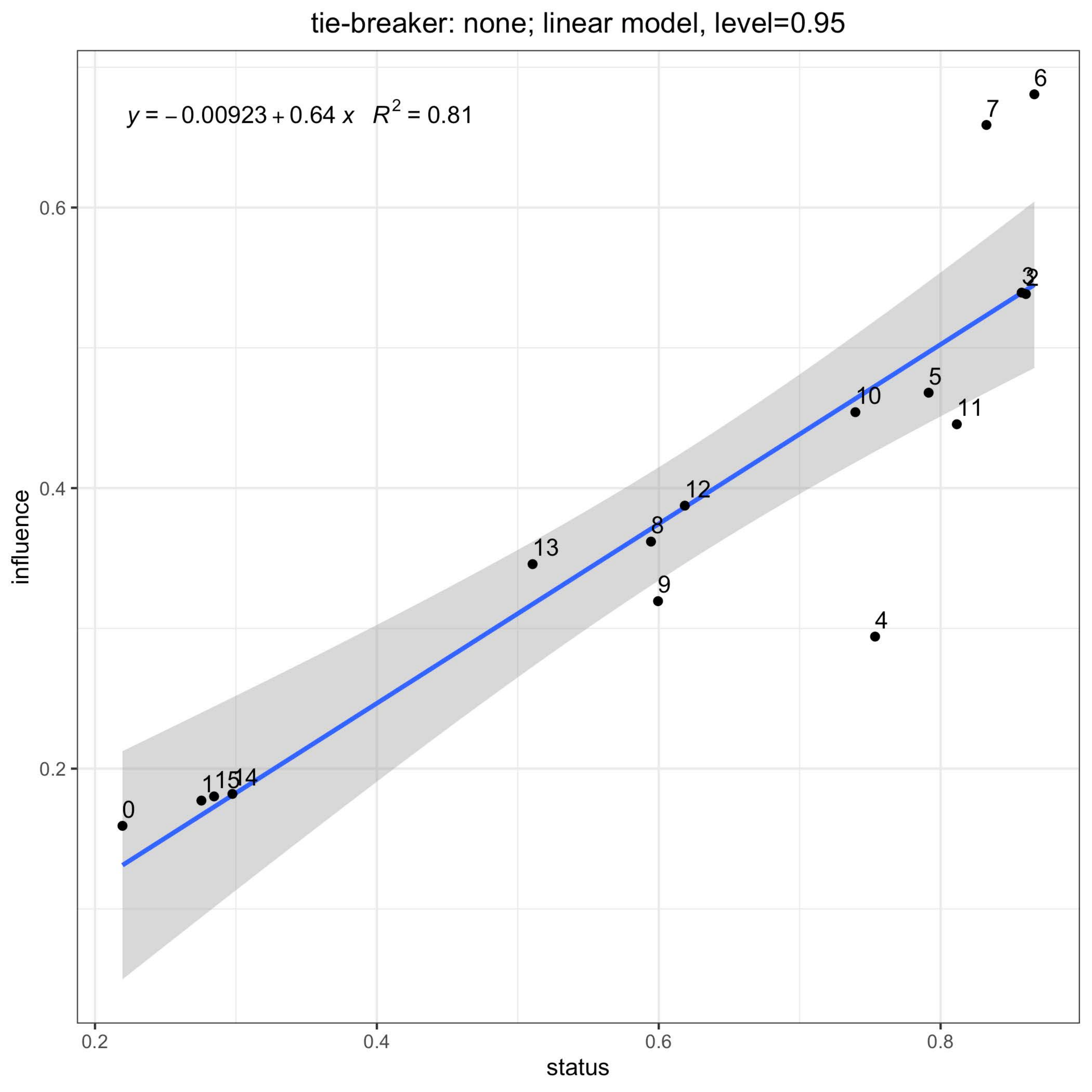}
  \caption{Left: The status of each vertex for no-tie-break (red circle), vertex $0$ breaks ties (green triangle), and vertex $6$ breaks ties (blue square). Tie-break vertices are outlined with an open triangle. Right: Highland Tribes status/influence correlation for no 0 tiebreak.}
  \label{fig:Highland}
\end{figure}

\subsubsection{Experiment: Signed Spectral Clustering vs. Frustration Cloud} 
\label{sssec:Exp5}

Figure~\ref{fig:Highland} shows influence as a function of vertex ID, and the same 4 nodes have the lowest computed influence under each tie-break scenario. An interesting observation on nodes  4, 6, and 7 is that they all have high status; the influence of vertex  4 (\emph{Nagam} tribe) is less than its status in the network, and the influence of vertices 6 and 7 (\emph{Masil} and \emph{Ukudz} tribes) is higher than their status in the network. graphB analysis provides a simple, unbiased view into Highland data and flags 3 out of 16 tribes to be re-examined more deeply by anthropology experts. The clusters in Figure~\ref{f:Spectral} are calculated only using positive edge spectral clustering to detect nearly-connected non-adversarial relationships. We demonstrate that status is a spectrum of spectral clustering, as anticipated in Section \ref{ssec:practical}.  The circled blue cluster is the same as the low status group we originally identified. The inclusion of the negative edge information for meaningful signed spectral clustering must be examined \cite{Kunegis2010}, especially in highly adversarial networks.  For a dataset like Highland, signed spectral clustering  \cite{Kunegis2010} produces the same result for $k=2$ and $k=3$, as illustrated in Figure~\ref{f:Spectral1}.

\begin{figure}[!ht]
  \centering
  \includegraphics[width=2.75in]{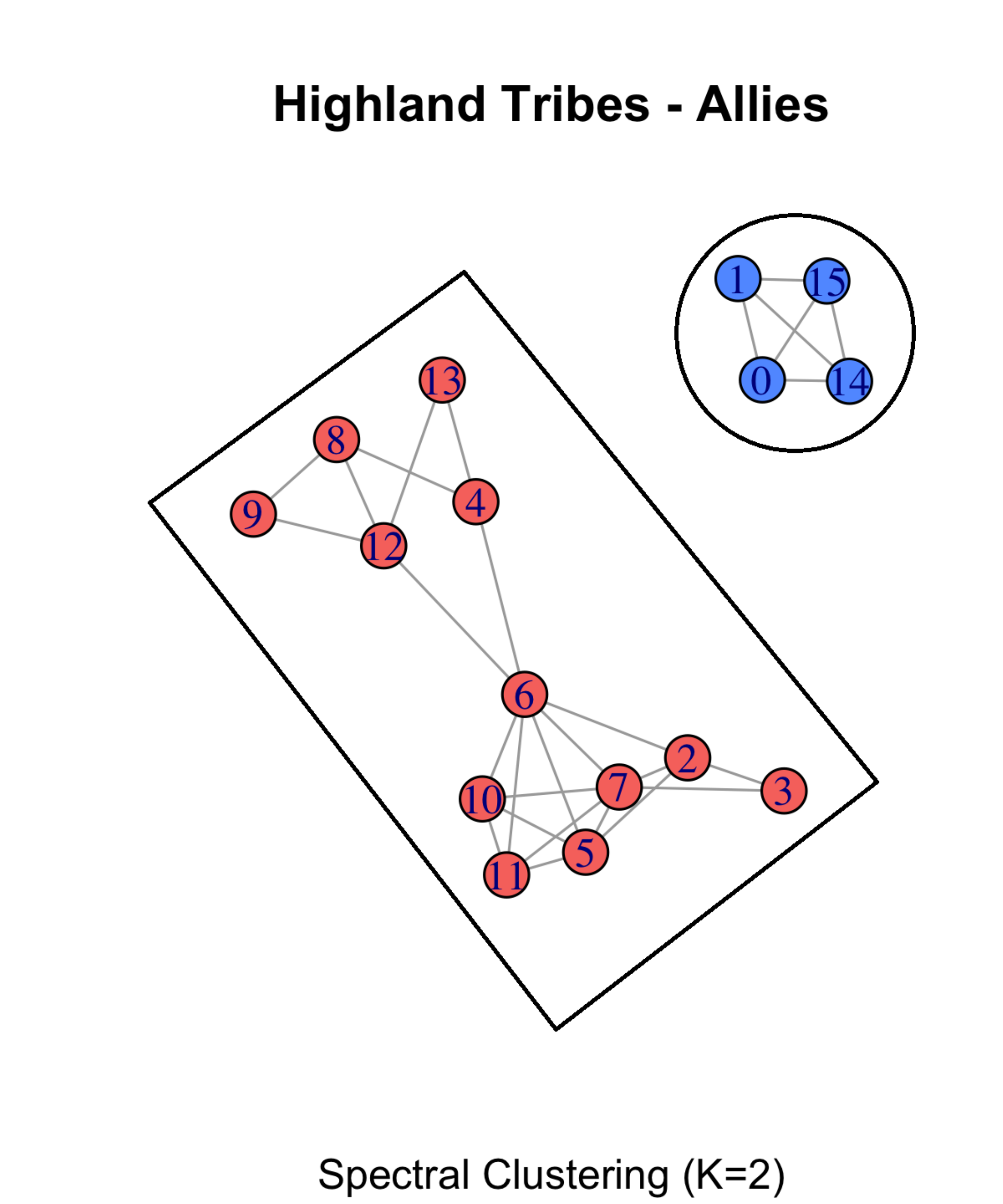}
  \includegraphics[width=2.75in]{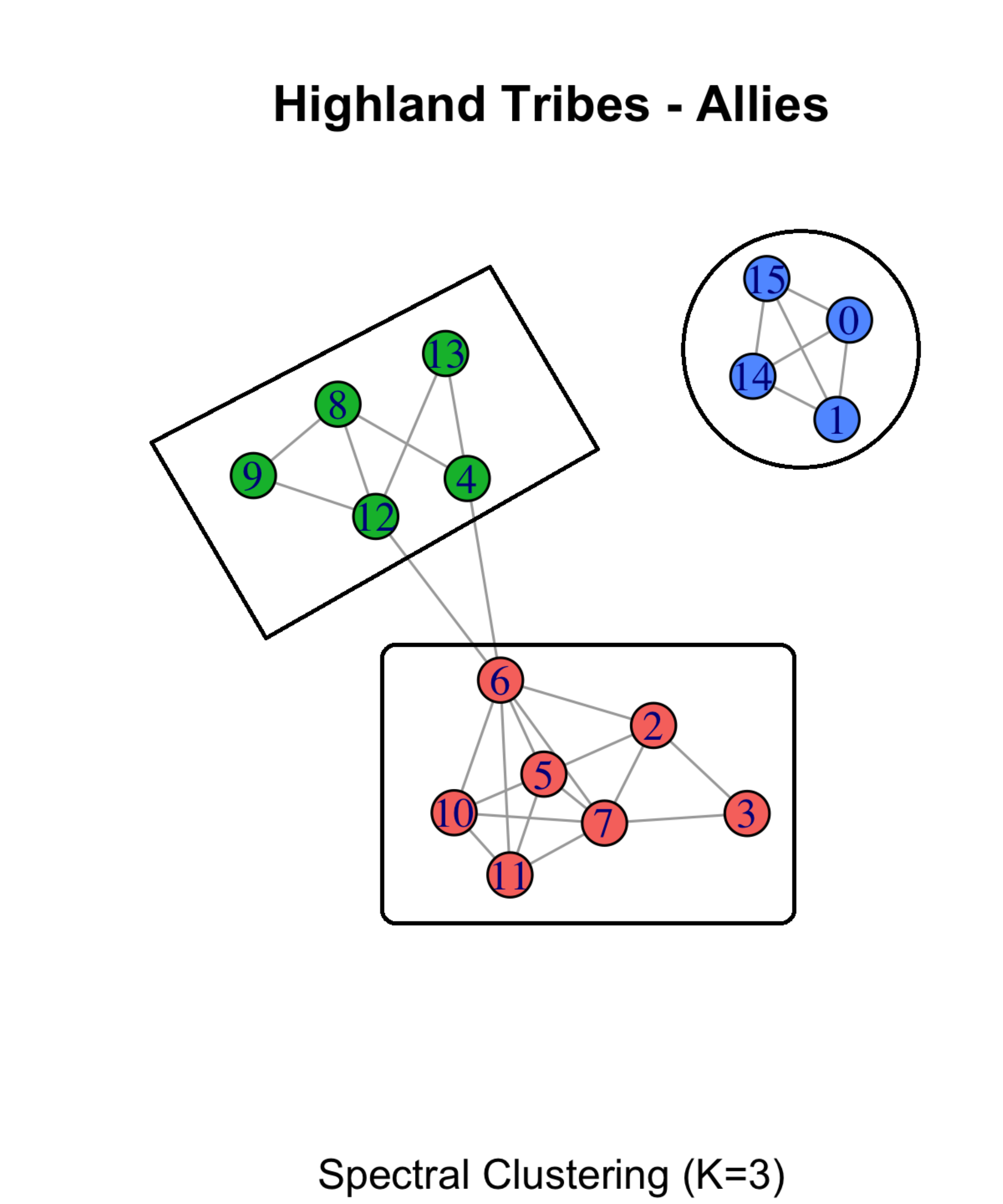}
  \caption{Spectral clustering of Highland Tribes data using positive edges for $k=2$ and $k=3$;}
  \label{f:Spectral}
\end{figure}

Figure~\ref{f:Spectral1} marks vertices by cluster belongings (color and shape) in status/influence space. Clusters are computed using signed spectral clustering implementation, and it is clear that status and influence capture spectrum of spectral clustering, as illustrated in Figure~\ref{f:Spectral1}. This experiment indicates that the robustness of our status/influence model means we do not need the matrix or the eigenvalues to cluster vertices. Moreover, here is no need to specify $k$ for spectral clustering, as the status/influence cone groups nodes in 2-D space. A study on more degenerate, adversarial networks is necessary to determine if status can provide insight into networks where spectral clustering fails.

\begin{figure}[!ht]
  \centering
  \includegraphics[width=2.75in]{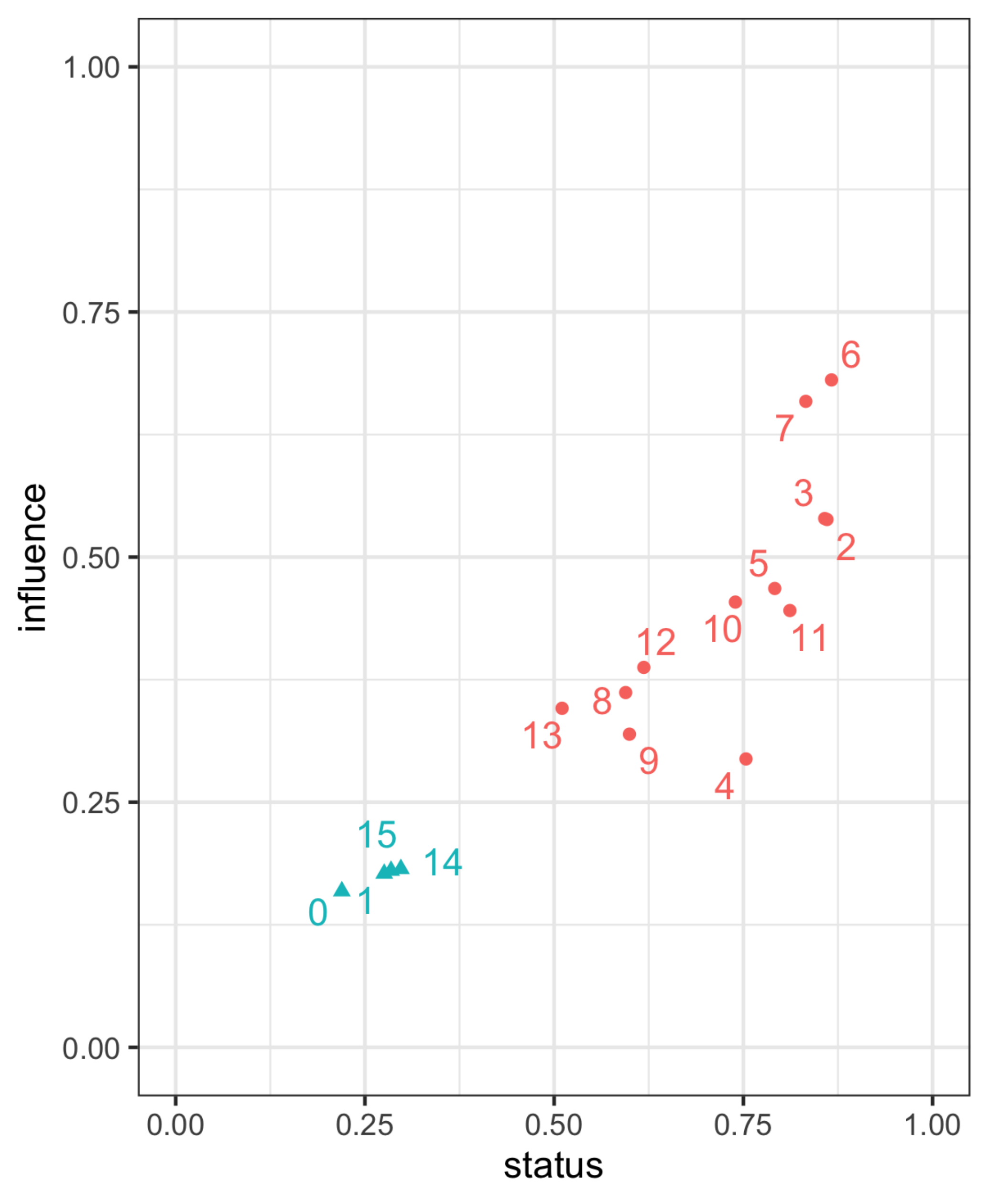}
  \includegraphics[width=2.75in]{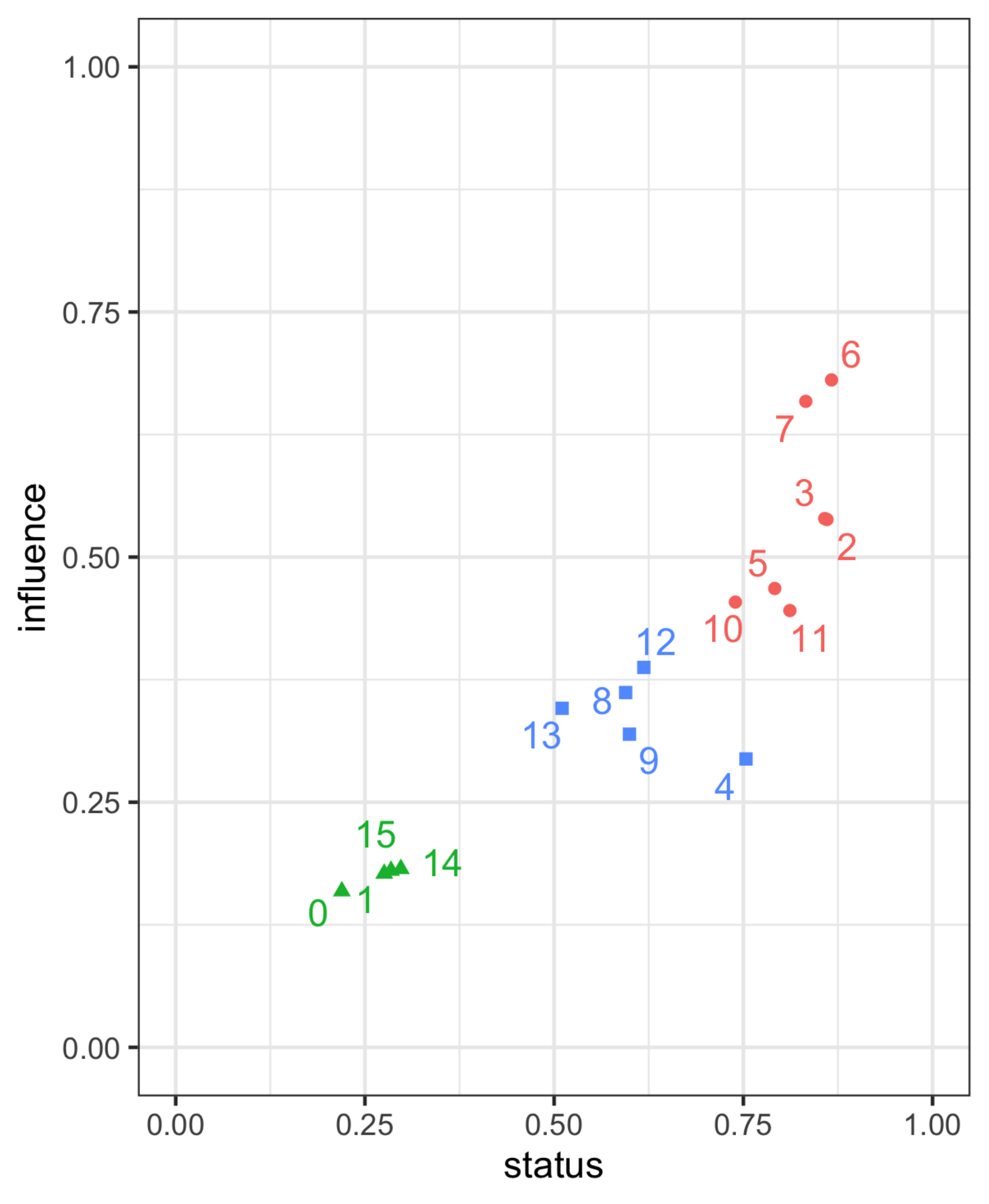}
  \caption{Signed spectral clustering \cite{Kunegis2010} using symmetric Laplacian results for Highland Tribes data for $k=2$ and $k=3$ in status/influence space demonstrate correspondence of radial distance in status/influence space to clustering. Each clusters is represented with a different shape/color.}
  \label{f:Spectral1}
\end{figure}

\section{Conclusion and Future Work}

In this paper, we propose consensus-based quantification of the vertices and edges in a graph. Nearest balanced states of a given network model attitudinal strength and influence in the network through various measures of a network's ability to reach a balanced state with a minimal number of edge sign changes. We introduce the concept of "frustration cloud" and vertex scores that capture vertex relation to the nearest balanced state of the network. The frustration cloud replaces the necessity of determining a balanced state with the least number of sentiment changes by determining a set of nearest states with minimal sentiment disruption. This also trades the NP-hardness of the frustration index for determining fundamental cycles and spanning trees. 

We introduce new metrics that quantify the importance of each balanced state relative to the likelihood it will be become the consensus state. The tree-search methods are shown to provide resolution to the data, while the quality of the resolution appears to be indistinguishable beyond $n=1000$ spanning trees, as seen in Figures \ref{fig:WikiSamplingStatus} and \ref{fig:WikiSamplingInfluence}. The measures of vertex status and influence both demonstrate differences in power dynamics when they are not correlated, as seen in Figure \ref{fig:WikiAllTree} and Figure \ref{fig:WikiNoTree}. These vertex scores provide an alternative to examine existing promotional practices, as indicated in Figure~\ref{fig:RfAnalysis}, and flag anomalous users. 

The status/influence cone provides a view of fairness and power for the data as shown in Figure~\ref{fig:WikiStatusInfluence}. The social network at scale produces a different shaped status/influence cone in Figure~\ref{f:Slashdot2} that indicates a large separation of power as status remains relatively constant when compared to influence. The smaller dataset of Highland Tribes demonstrates quickly reproducible results for the Conservation Law of Controversy (average status is constant regardless of tie-break node) in Figure~\ref{fig:Highland}, as well as the efficacy of our proposed method for spectral clustering in Figure~\ref{f:Spectral1}. 

Consensus modeling has gained traction as a way to model agreement in multi-agent networks in the presence of antagonistic interactions \cite{2013Altafini,She2020}. We plan to apply proposed balancing theory to a multi-agent network agreement to see if we can identify and tune existing policies and detect biased agents (AI modules) in the network. Future research will also focus on comparison of proposed approach to signed and weighted sign graph spectral clustering \cite{Aref2016, Mercado}. We are developing metrics to measure the strength of vertex and edge interactions, a measure that utilizes known promotional outcomes to detect and quantify bias for the majority/minority, and a toolkit for deeper analysis of the proposed metrics.   

\paragraph{Acknowledgements:} 
We would like to thank Data Lab alumni students Joshua Mitchell for initial implementation and graphB 1.0 proof-of-concept and Eric Hull for graphB 2.0 proof-of-concept and code release. We would like to acknowledge Ph.D. student Maria Tomasso and credit her for preliminary data analysis and visualizations.  We would like to thank Texas State for its support through startup funding, computational facilities, and faculty development programs.

\end{document}